\documentclass[12pt]{article}
\usepackage[square]{natbib}
\usepackage{tabularx,array,booktabs} 
\newcommand{\keywords}[1]{\textbf{Keywords:} #1}
\usepackage{algorithm,color}
\usepackage{algorithmic}
\usepackage{subcaption} 
\usepackage{graphicx}
\usepackage{epstopdf}
\usepackage{amsmath}
\usepackage{amsthm}
\usepackage{url}
\usepackage{amssymb}
\usepackage{multirow}
\usepackage{bm}
\usepackage{booktabs}
\usepackage{placeins}
\usepackage{mathrsfs}
\usepackage{fullpage}
\usepackage{setspace}
\usepackage{needspace}
\usepackage[english]{babel}
\usepackage{bibunits}
\usepackage[utf8]{inputenc}

\newtheorem{assumption}{Assumption}
\newtheorem{theorem}{Theorem}

\newtheorem{lemma}{Lemma}

\newcommand{\diag}{\operatorname{diag}}
\newcommand{\expit}{\operatorname{expit}}
\newcommand{\tr}{\operatorname{Tr}}

\newcommand{\logit}{\operatorname{logit}}
\newcommand{\blinf}{\mathrm{bl},\infty}

\begin{document}
	
\title{Linear regression with probabilistic networks for network-linked data}

\author{
	Jingnan Zhang$^\dag$, Shiwen Ye$^\ddag$, Danyang Huang$^\S$ \\ [10pt]
	$^\dag$ Faculty of Business for Science \& Technology, School of Management \\University of Science and Technology of China\\
    $^\ddag$ Department of Statistics and Data Science\\
	The Chinese University of Hong Kong\\
	$^\S$ School of Statistics, Renmin University of China\\
}

\date{}
\maketitle
\onehalfspacing

\begin{abstract}
Network-linked data combine a network with covariates and responses observed for its nodes. Common regression approaches either model dependence among neighboring responses using the observed network or capture similarities associated with network structure without explicitly modeling response dependence. We propose a probabilistic network-linked data analysis framework that incorporates both features while accounting for randomness in network formation. The model uses normalized connection probabilities to weight peer responses and latent positions to capture network cohesion. A two-step likelihood procedure first estimates the network structure and then fits the response model. We establish theoretical guarantees for network recovery in sparse networks with heterogeneous node effects, and prove consistency and joint asymptotic normality of the regression estimators subject to suitable conditions. The analysis also clarifies when estimating the network preserves the classical convergence rate for spatial autoregressive models. Simulation studies demonstrate the estimation performance of the method and assess its inference accuracy. An analysis of diplomatic network data illustrates its practical application.
\end{abstract}

\medskip

\doublespacing

\begin{bibunit}[apalike]
\section{Introduction}

Network data record interactions among units and arise in social, economic and biological applications. Much of statistical network analysis focuses on tasks such as community detection \citep{lei2015consistency}, link prediction \citep{cho2018link} and node classification \citep{bhagat2011node}. In many settings, the network is also observed together with node-level covariates and responses. Following \citet{li2019prediction} and \citet{le2022linear}, we refer to such data as network-linked data. Combining relational structure with unit-level attributes provides a richer basis for prediction and inference, especially when outcomes are related across connected units \citep{michell1996peer,basse2018model}.

The analysis of network-linked data can be traced back to econometrics, where connected units may exhibit statistically dependent outcomes by sharing similar behaviors or characteristics \citep{manski1993identification,jackson2007relating,blume2015linear,yang2017social}. Let $\bm{A}=(a_{ij})_{i,j\in[n]}\in\{0,1\}^{n\times n}$ be the adjacency matrix of an undirected network with $n$ nodes, where $a_{ij}=a_{ji}=1$ indicates the presence of an edge between nodes $i$ and $j$, and $a_{ij}=a_{ji}=0$ otherwise, with $a_{ii}=0$. For each node $i\in[n]$, we also observe a covariate vector $\bm{x}_i\in\mathbb{R}^p$ and a response variable $y_i\in\mathbb{R}$. Let $\bm{X}=(\bm{x}_1,\dots,\bm{x}_n)^\top\in\mathbb{R}^{n\times p}$ and $\bm{y}=(y_1,\dots,y_n)^\top\in\mathbb{R}^n$. The primary objective in the analysis of network-linked data is to model the relationship between $\bm{y}$ and $\bm{X}$ while incorporating the network structure encoded by $\bm{A}$. One of the most widely used models for analyzing network-linked data in econometrics is the social interaction model (SIM) \citep{lee2007identification,bramoulle2009identification,lee2010specification, liu2010gmm}, which assumes that
\begin{align}
\label{Model: SIM}
    y_i=\bm{x}_i^\top\bm\beta+\left(\frac{1}{n_i}\sum_{j\in[n]}a_{ij}y_j\right)\gamma+\left(\frac{1}{n_i}\sum_{j\in[n]}a_{ij}\bm{x}_j^\top\right)\bm\eta_{\rm SIM}+\epsilon_i,
\end{align}
where $n_i=\sum_{j\in[n]}a_{ij}$ and $\epsilon_i$ is the noise term with $\mathbb{E}(\epsilon_i)=0$ and $\mathrm{Var}(\epsilon_i)=\sigma^2$; when $n_i=0$, the two neighborhood averages are defined to be zero. Here, the coefficient \(\gamma\) measures peer effect in the model, while \(\bm\eta_{\rm SIM}\) is in $\mathbb R^p$ and captures contextual effects from neighboring covariates. It reduces to a spatial autoregressive model if $\bm\eta_{\rm SIM}=\bm 0$ \citep{getis2009spatial,huang2019least,zhu2020multivariate}. SIM has also been extended to analyze time-varying network-linked data \citep{zhu2017network,chen2023community}. The identification and interpretation of peer effects through networks have been further studied in a broad econometric literature \citep{goldsmith2013social,lee2016identification,bramoulle2020peer}. Despite their usefulness, model \eqref{Model: SIM} and related approaches based on binary adjacency weights have two limitations.

The first limitation is the use of realized binary adjacency entries as peer weights. After degree normalization, all observed neighbors of node $i$ enter the peer-response term with equal weights, regardless of the strength of their connections Such an assumption is often unrealistic in practice \citep{le2022linear}. Numerical studies in \citet{li2019prediction} and \citet{le2022linear} also show that restrictive specifications of social effects can lead to poor predictive performance. The second limitation is that the response model takes the observed network as given, without modeling the uncertainty in edge formation. Statistical network analysis often views the observed adjacency matrix as one realization from an underlying probabilistic model, thereby accounting for uncertainty in network formation. Typical examples include graphon models \citep{gao2015rate,zhang2017estimating,zhao2022dimension}, as well as two of their special cases: the stochastic block model \citep{holland1983stochastic,karrer2011stochastic} and latent space models \citep{hoff2002latent,hoff2007modeling,ma2020universal,zhang2022joint}. The econometric literature on network formation may also treat networks themselves as outcomes of underlying strategic or stochastic formation processes \citep{depaula2020econometric}. Recent works in econometrics have further considered measurement errors and link misclassification in social network models \citep{lewbel2024ignoring,lewbel2025estimating}. These studies indicate that the observed network need not be treated as a fixed and exact object. However, they do not directly address how the latent network-generating structure can be incorporated into a regression model for network-linked responses.

A closely related regression approach is \citet{le2022linear}, who propose a semi-parametric regression model, referred to as the subspace projection (SP) method in their numerical studies. Specifically, they assume
\begin{align}
\label{Model: Semi}
    \mathbb{E}(\bm{y})=\bm{X}\bm\beta+\bm\mu.
\end{align}
Letting $\bm{P}=\mathbb{E}(\bm{A})$, they further assume that $\bm\mu$
belongs to a subspace spanned by a fixed number of leading eigenvectors
of $\bm P$ to account for network uncertainty and capture network
cohesion. Although this model captures low-dimensional network structure,
it no longer explicitly models peer-response dependence, which plays a
central role in social interaction models. Moreover, because $\bm\mu$ is
introduced as a flexible node-level effect, it may absorb both network
cohesion and other sources of node-specific heterogeneity.

% {\color{magenta}Latent network structure can also be incorporated into spatial models with an explicit peer-response term. \citet{paul2024spatial} correct measurement error in estimated latent covariates in a network application, while retaining normalized observed adjacency weights in the peer-response term. Recently, \citet{hayes2026estimating} specify a latent contagion model that uses normalized expected adjacency weights for peer responses and includes latent positions as covariates. Their approach uses adjacency spectral embedding and two-stage least squares estimation under a fixed-rank, positive-semidefinite expected adjacency model. In the random dot product graph specification, connection probabilities are given by inner products of latent positions. In this work, we seek to model network formation under a logistic latent space specification, with individual-specific node effects and latent similarity entering the connection log-odds. We then study peer-effect estimation and inference under this specification, accounting for estimation errors in both nodal latent positions and connection probabilities.}

In this paper, we propose a probabilistic network-linked data analysis (NLDA) framework that addresses the two limitations of the SIM model in \eqref{Model: SIM} while preserving its ability to model peer-response dependence. We assume that the observed network is generated from a latent space model and use the resulting edge-probability matrix $\bm P$, rather than the realized adjacency matrix $\bm A$, to define the peer-response exposure in the regression model. We normalize these probabilities by the expected degree, so that their sum for each node equals one even when the network is sparse. We further incorporate network cohesion through the latent position vectors, which separates the low-dimensional cohesion component from individual-specific degree heterogeneity. We estimate the network structure through the Bernoulli likelihood and then use the fitted probability matrix and latent positions in a Gaussian quasi-likelihood for the responses.

The main contributions are threefold. First, we develop a regression framework that combines dependence among peer responses with latent network cohesion and accounts for randomness in network formation. The network formation follows a logistic latent space model with individual-specific node effects. Second, building on likelihood methods for latent space models \citep{ma2020universal,nldaLi2025Inference}, we establish maximum and mean-square recovery rates for a local network likelihood maximizer under increasing node-effect heterogeneity. These bounds control the fitted probability matrix and latent positions. Third, we establish consistency and joint asymptotic normality of the regression estimators and provide a consistent covariance estimator for inference. With suitable restrictions on network sparsity and heterogeneity, the regression estimators attain the classical convergence rate for spatial autoregressive models with known network weights \citep{lee2004asymptotic,huang2019least}, even though the network structure must be estimated. These results clarify the additional requirements for accurate regression inference beyond those needed for network recovery.
% {\color{magenta}The analysis jointly controls estimation errors in the probability weights and latent positions obtained from the same logistic network likelihood.}

The remainder of this paper is organized as follows. Section 2 presents the probabilistic network-linked data model and the two-stage estimation procedure. Section 3 studies network recovery and regression inference. Section 4 reports simulation studies, Section 5 analyzes diplomatic network-linked data, and Section 6 concludes. Technical proofs are provided in the Supplementary Material.

\paragraph{Notations.} We introduce some notations to be used throughout the paper. For any positive integer $q$, let $[q]=\{1,\ldots,q\}$. We use $\bm0_q$, $\bm1_q$ and $\bm I_q$ for the zero vector, vector of ones and identity matrix, and omit the dimensions when they are clear from the context. For a vector $\bm v$, $\|\bm v\|_2$ and $\|\bm v\|_\infty$ denote its Euclidean and maximum norms. For a matrix $\bm B=(b_{ij})$, $\|\bm B\|_F$ and $\|\bm B\|_2$ denote its Frobenius and operator norms, $\|\bm B\|_{\max}=\max_{i,j}|b_{ij}|$, and $\|\bm B\|_{2,\infty}=\max_i\|\bm B_{i\cdot}\|_2$. We use $\|\bm B\|_\infty=\max_i\sum_j|b_{ij}|$ for the matrix infinity norm. Inner products are Euclidean for vectors and $\langle\bm B,\bm C\rangle_F=\tr(\bm B^\top\bm C)$ for matrices of the same dimensions. We write $\tr(\cdot)$ for the trace and $\diag(v_1,\ldots,v_q)$ for the diagonal matrix with the listed entries. Eigenvalues and singular values are ordered decreasingly and denoted by $\lambda_k(\cdot)$ and $\sigma_k(\cdot)$; $\lambda_{\min}(\cdot)$ is the smallest eigenvalue. For symmetric matrices, $\preceq$ denotes the positive semidefinite order. We write $\operatorname{col}(\bm B)$ for the column space of $\bm B$ and $\bm M_{\bm B}$ for the orthogonal projector onto $\operatorname{col}(\bm B)^\perp$. When $\bm B$ has full column rank, $\bm M_{\bm B}=\bm I-\bm B(\bm B^\top\bm B)^{-1}\bm B^\top$. We use $O,o,O_p,o_p$ in their usual senses, $a_n\asymp b_n$ when the ratio of two positive sequences is bounded above and away from zero, and $\to_p,\to_d$ for convergence in probability and distribution.

\section{Probabilistic network-linked data model}

We adopt the perspective that
\begin{itemize}
    \item[(i)] the binary neighborhood average in \eqref{Model: SIM} assigns the same weight to every observed neighbor and therefore does not distinguish their connection strengths;
    \item[(ii)] the assumption that the network is perfectly observed is also inappropriate.
\end{itemize}
To address these two issues, we propose a probabilistic model for network-linked data that differs from the semi-parametric model underlying the SP method in \eqref{Model: Semi}. Specifically, we assume that the variables $a_{ij}$, $1\le i<j\le n$, are independent Bernoulli variables with probabilities $p_{ij}$ satisfying
\begin{align}
\label{Model: Latent}
\theta_{ij}=\log\frac{p_{ij}}{1-p_{ij}}
=\alpha_i+\alpha_j+\bm z_i^\top\bm z_j.
\end{align}
We set $p_{ji}=p_{ij}$ and $p_{ii}=0$, and write $\expit(u)=(1+e^{-u})^{-1}$ and $\logit(p)=\log\{p/(1-p)\}$. Here, $\alpha_i\in\mathbb R$ captures the individual-specific effect of node $i$, and $\bm z_i\in\mathbb R^r$ captures homogeneity among nodes, with fixed $r\ll n$. Write $\bm\alpha=(\alpha_1,\ldots,\alpha_n)^\top$ and $\bm Z=(\bm z_1,\ldots,\bm z_n)^\top$. Model \eqref{Model: Latent} is the latent space model studied in \citet{hoff2002latent}, \citet{ma2020universal} and \citet{zhang2022joint}.

Next, we model the network-linked data by assuming
\begin{align}
\label{Model: PMNLD}
y_i=\bm x_i^\top\bm\beta+
\left(\frac{\sum_{j\ne i}p_{ij}y_j}{\sum_{j\ne i}p_{ij}}\right)\gamma
+\bm z_i^\top\bm\eta+\epsilon_i,
\end{align}
for each $i\in[n]$, where $\bm\beta\in\mathbb R^p$, $\gamma\in\mathbb R$, $\bm\eta\in\mathbb R^r$, and $\epsilon_i$ is the error term. Let $\bm P=(p_{ij})_{i,j\in[n]}$, $d_i(\bm P)=\sum_{j\ne i}p_{ij}$ and $\bm D(\bm P)=\diag\{d_1(\bm P),\ldots,d_n(\bm P)\}$. In matrix form, model \eqref{Model: PMNLD} is
\begin{align*}
\bm y=\bm X\bm\beta+\gamma\bm W(\bm P)\bm y+\bm Z\bm\eta+\bm\epsilon,
\end{align*}
where $\bm W(\bm P)=\bm D(\bm P)^{-1}\bm P$ and $\bm\epsilon=(\epsilon_1,\ldots,\epsilon_n)^\top$. The errors $\epsilon_i, i\in[n]$ are independent and identically distributed with mean zero and variance $\sigma^2>0$, and are independent of the network observations.

Models \eqref{Model: Latent} and \eqref{Model: PMNLD} have some compelling features. First, the peer-response term in model \eqref{Model: PMNLD} replaces the binary weights $a_{ij}$ in model \eqref{Model: SIM} with the connection probabilities $p_{ij}$ defined in model \eqref{Model: Latent}. This addresses issues (i) and (ii) by allowing connection strengths to vary across node pairs and treating the observed network as a realization of a probabilistic model. Thus, nodes $j$ and $k$ may exert different levels of influence on node $i$, even when $a_{ij}=a_{ik}=1$. For example, in a social network, individuals may be more strongly influenced by close friends than by casual acquaintances. We normalize these probabilities by the expected degree $d_i(\bm P)$, so the peer-response term is a weighted average over all potential peers. Every row of $\bm W(\bm P)$ sums to one, keeping the total peer-response weight on the same scale as network density changes. Second, we replace $\bm\mu$ in model \eqref{Model: Semi} by $\bm Z\bm\eta$. Whereas the SP formulation constrains $\bm\mu$ to a leading eigenspace of $\bm P$, our network model separates latent similarity, represented by $\bm Z$, from individual-specific variation in connection propensity, represented by $\bm\alpha$. The response component $\bm Z\bm\eta$ therefore captures network cohesion through these latent positions.

The two response components play distinct roles. The term $\gamma\bm W(\bm P)\bm y$ represents dependence on the responses of other nodes, while $\bm Z\bm\eta$ captures latent similarity or shared structural position associated with similar responses, even when $\gamma=0$. Including both components allows the model to account for peer-response dependence and network cohesion; omitting either may lead to model misspecification and obscure the distinction between them.

\subsection{Computing algorithm}

For models \eqref{Model: Latent} and \eqref{Model: PMNLD}, we develop a two-step procedure to estimate the network structure and response parameters. Let
$\bm\Theta(\bm\alpha,\bm Z)
=\bm\alpha\bm1_n^\top+\bm1_n\bm\alpha^\top+\bm Z\bm Z^\top$.
In the first step, we estimate $(\bm\alpha,\bm Z)$ from $\bm A$ and
obtain $\widehat{\bm P}$ from the fitted off-diagonal logits.
In the second step, we substitute $\widehat{\bm Z}$ and
$\widehat{\bm P}$ into model \eqref{Model: PMNLD} to estimate the
response parameters using the observed covariates and responses
$(\bm X,\bm y)$. Before proceeding, there is an identifiability issue to be addressed in model \eqref{Model: Latent}. A widely considered condition is
$\bm Z^\top\bm1_n=\bm0$, which has also been used in the literature
\citep{ma2020universal,zhang2022joint}. We impose this condition while
leaving $\bm\alpha$ unrestricted.

Let $\bm P(\bm\alpha,\bm Z)$ be the probability matrix with zero diagonal.
The network log-likelihood is
\begin{align*}
L_A(\bm\alpha,\bm Z)
=\sum_{1\le i<j\le n}
\left[a_{ij}\theta_{ij}
-\log\{1+\exp(\theta_{ij})\}\right],
\end{align*}
where $\theta_{ij}=\alpha_i+\alpha_j+\bm z_i^\top\bm z_j$.
We use the spectral initialization described in Algorithm~\ref{alg:step1}.
With $\bm J_n=\bm I_n-\bm1_n\bm1_n^\top/n$ and
$\bm P^{(t)}=\bm P(\bm\alpha^{(t)},\bm Z^{(t)})$, the likelihood updates are
\begin{align}
\label{eq: practical network update}
 \bm\alpha^+=\bm\alpha^{(t)}+\frac{\xi}{n}
   (\bm A-\bm P^{(t)})\bm1_n,\
 \bm Z^+=\bm J_n\left\{\bm Z^{(t)}
   +\frac{2\xi}{\max(\|\bm Z^{(t)}\|_2^2,1)}
   (\bm A-\bm P^{(t)})\bm Z^{(t)}\right\}.
\end{align}
Write $\bm\Theta^{(t)}=\bm\Theta(\bm\alpha^{(t)},\bm Z^{(t)})$. We initialize $\xi$ at $\xi_0$ and select it by likelihood backtracking. Algorithm~\ref{alg:step1} summarizes the procedure.

\begin{algorithm}[!thb]
\caption{Network likelihood estimation}
\label{alg:step1}
\begin{algorithmic}[1]
\STATE Input: $\bm A$, $r$, and tuning parameters $c_{\tau,0}$, $\varepsilon_{\rm init}$, $\xi_0$, $\varepsilon_{\rm stop}$ and $T$.
\STATE Compute $\widehat\rho=2\{n(n-1)\}^{-1}\sum_{i<j}a_{ij}$ and set $\lambda_{\rm sp}=c_{\tau,0}\sqrt{\max(n\widehat\rho,\log n)}$.
\STATE From $\bm A=\sum_k\lambda_k\bm u_k\bm u_k^\top$, set $\bm P^{\rm sp}=\sum_{|\lambda_k|\ge\lambda_{\rm sp}}\lambda_k\bm u_k\bm u_k^\top$.
\STATE Set $\Theta^{\rm init}_{ij}=\logit[\min\{1-\varepsilon_{\rm init},\max(\varepsilon_{\rm init},P^{\rm sp}_{ij})\}]$ for $i\ne j$, and $\Theta^{\rm init}_{ii}=0$.
\STATE Set $\bm\alpha^{(0)}=\arg\min_{\bm\alpha}\sum_{i<j}(\Theta^{\rm init}_{ij}-\alpha_i-\alpha_j)^2$.
\STATE Set $R^{\rm init}_{ij}=\Theta^{\rm init}_{ij}-\alpha_i^{(0)}-\alpha_j^{(0)}$ for $i\ne j$, and $R^{\rm init}_{ii}=0$.
\STATE Let $d_1\ge\cdots\ge d_r$ be the largest $r$ eigenvalues of $\bm J_n\bm R^{\rm init}\bm J_n$, with corresponding orthonormal eigenvectors $\bm v_1,\ldots,\bm v_r$.
\STATE Set $\bm Z^{(0)}=\bm J_n\bigl(\sqrt{\max(d_1,0)}\,\bm v_1,\ldots,\sqrt{\max(d_r,0)}\,\bm v_r\bigr)$ and $\xi=\xi_0$.
\FOR{$t=0,\ldots,T-1$}
\STATE Use \eqref{eq: practical network update} with backtracking to obtain $(\bm\alpha^{(t+1)},\bm Z^{(t+1)})$; stop if no step is accepted.
\STATE Stop if $\|\bm\Theta^{(t+1)}-\bm\Theta^{(t)}\|_F/(\|\bm\Theta^{(t)}\|_F+1)<\varepsilon_{\rm stop}$.
%\STATE Set $\xi=\min(1.05\xi,10)$ for the next iteration.
\ENDFOR
\STATE Set $(\widehat{\bm\alpha},\widehat{\bm Z})$ to the iterate with the largest likelihood.
\STATE Set $\widehat{\bm P}=\bm P(\widehat{\bm\alpha},\widehat{\bm Z})$ and $\widehat{\bm W}=\bm W(\widehat{\bm P})$.
\STATE Return $\widehat{\bm\alpha}$, $\widehat{\bm Z}$, $\widehat{\bm P}$ and $\widehat{\bm W}$.
\end{algorithmic}
\end{algorithm}

%We use $\varepsilon_{\rm init}=10^{-4}$ and $\varepsilon_{\rm stop}=10^{-5}$.
%Backtracking accepts the first finite proposal satisfying
%$L_A(\bm\alpha^+,\bm Z^+)\ge L_A(\bm\alpha^{(t)},\bm Z^{(t)})-10^{-10}$
%among at most eight halvings of $\xi$. After acceptance, the trial scale is
%multiplied by $1.05$ and capped at $10$, as specified in Algorithm~\ref{alg:step1}.

After obtaining $\widehat{\bm P}$ and $\widehat{\bm Z}$, we
estimate $\bm\beta$, $\gamma$, $\bm\eta$ and $\sigma^2$ in model
\eqref{Model: PMNLD}. Write
$\bm\zeta=(\bm\beta^\top,\bm\eta^\top,\gamma,\tau)^\top$, where
$\tau=\sigma^{-2}$. For fixed network parameters, the normalized Gaussian
quasi-log-likelihood is
\begin{align}
\label{eq: response quasi likelihood}
 L_y(\bm\zeta;\bm P,\bm Z)
 &=\frac1n\log\det\bm S_\gamma(\bm P)+\frac12\log\tau-\frac{\tau}{2n}
 \|\bm S_\gamma(\bm P)\bm y-\bm X\bm\beta-\bm Z\bm\eta\|_2^2,
\end{align}
where $\bm S_\gamma(\bm P)=\bm I_n-\gamma\bm W(\bm P)$.
Let $\Gamma_\gamma=[-1+\delta_\gamma,1-\delta_\gamma]$ for a fixed
$\delta_\gamma\in(0,1)$. The matrix $\bm W(\bm P)$ is similar to the symmetric
matrix $\bm D(\bm P)^{-1/2}\bm P\bm D(\bm P)^{-1/2}$, so its eigenvalues are real.
Since $\bm W(\bm P)$ has nonnegative entries and each row sums to one,
these eigenvalues lie in $[-1,1]$. Thus, for $\gamma\in\Gamma_\gamma$,
\[
 \lambda_{\min}\{\bm S_\gamma(\bm P)\}
 \ge1-|\gamma|\ge\delta_\gamma>0,
\]
and hence $\det\bm S_\gamma(\bm P)>0$. Instead of jointly updating $\bm\zeta$, we profile out $\bm\beta$, $\bm\eta$ and
$\tau$ and reduce the optimization to a one-dimensional problem in $\gamma$.
Put $\bm X_+=(\bm X,\bm Z)$. For each
$\gamma\in\Gamma_\gamma$, jointly regressing $\bm S_\gamma(\bm P)\bm y$ on
$\bm X_+$ gives
$\operatorname{RSS}_n(\gamma;\bm P,\bm Z)
=\|\bm M_{\bm X_+}\bm S_\gamma(\bm P)\bm y\|_2^2$.
Maximizing $L_y(\bm\zeta;\bm P,\bm Z)$ with respect to
$(\bm\beta,\bm\eta,\tau)$ for fixed $\gamma$ yields
\begin{align}
\label{eq: profile beta tau}
 \begin{pmatrix}\bm\beta_\gamma\\\bm\eta_\gamma\end{pmatrix}
 =(\bm X_+^\top\bm X_+)^{-1}\bm X_+^\top\bm S_\gamma(\bm P)\bm y,\quad
 \tau_\gamma=\frac{n}{\operatorname{RSS}_n(\gamma;\bm P,\bm Z)}.
\end{align}
Substituting these estimates into \eqref{eq: response quasi likelihood}, we obtain the profile quasi-log-likelihood
\begin{align}
\label{eq: response profile criterion}
 L_{y,\mathrm{prof}}(\gamma;\bm P,\bm Z)
 =\frac1n\log\det\bm S_\gamma(\bm P)
 -\frac12\log\left\{\frac{\operatorname{RSS}_n(\gamma;\bm P,\bm Z)}n\right\}
 -\frac12.
\end{align}

Plugging $\widehat{\bm P}$ and $\widehat{\bm Z}$ into \eqref{eq: response profile criterion}, we maximize $L_{y,\mathrm{prof}}(\gamma;\widehat{\bm P},\widehat{\bm Z})$ over a deterministic grid $\Gamma_{\gamma,n}\subset\Gamma_\gamma$ containing both endpoints and having maximum gap at most $n^{-1}$. After obtaining $\widehat\gamma$, we evaluate \eqref{eq: profile beta tau} at $(\widehat\gamma,\widehat{\bm P},\widehat{\bm Z})$ to obtain $\widehat{\bm\beta}$, $\widehat{\bm\eta}$ and $\widehat\tau$. The details are summarized in Algorithm~\ref{alg:step2}. Two cases lead to estimation failure in Algorithm~\ref{alg:step2}: (i) $\operatorname{rank}(\bm X,\widehat{\bm Z})<p+r$, so $(\bm X,\widehat{\bm Z})^\top(\bm X,\widehat{\bm Z})$ is singular and the coefficient estimates in \eqref{eq: profile beta tau} are undefined; or (ii) $\operatorname{RSS}_n(\gamma;\widehat{\bm P},\widehat{\bm Z})=0$ at some $\gamma\in\Gamma_{\gamma,n}$, so the estimate of $\tau$ in \eqref{eq: profile beta tau} is undefined. In either case, Algorithm~\ref{alg:step2} reports an estimation failure. For theoretical purposes, we assign $\widehat{\bm\zeta}=(\bm0_{p+r}^{\top},0,1)^\top$ on these events so that the estimator is defined on all samples. Fortunately, under the conditions specified in the following section, Lemma~\ref{lem: response profile uniform} in the Supplementary Material shows that the probability of either failure tends to zero.

\begin{algorithm}[!thb]
\caption{Estimation of the response parameters}
\label{alg:step2}
\begin{algorithmic}[1]
\STATE Input: $\bm y$, $\bm X$, $\widehat{\bm P}$, $\widehat{\bm Z}$ and $\Gamma_{\gamma,n}$.
\STATE Set $\bm X_+=(\bm X,\widehat{\bm Z})$; if $\operatorname{rank}(\bm X_+)<p+r$, report estimation failure and stop.
\FOR{each $\gamma\in\Gamma_{\gamma,n}$}
\STATE Compute $\operatorname{RSS}_n(\gamma;\widehat{\bm P},\widehat{\bm Z})$; if it is zero, report estimation failure and stop.
\STATE Compute $L_{y,\mathrm{prof}}(\gamma;\widehat{\bm P},\widehat{\bm Z})$ from \eqref{eq: response profile criterion}.
\ENDFOR
\STATE Choose $\widehat\gamma$ to maximize $L_{y,\mathrm{prof}}(\gamma;\widehat{\bm P},\widehat{\bm Z})$ over $\Gamma_{\gamma,n}$, choosing the smallest maximizer in a tie.
\STATE Evaluate \eqref{eq: profile beta tau} at $\widehat\gamma$ to obtain $\widehat{\bm\beta}$, $\widehat{\bm\eta}$ and $\widehat\tau$.
\STATE Return $\widehat{\bm\zeta}=(\widehat{\bm\beta}^\top,\widehat{\bm\eta}^\top,\widehat\gamma,\widehat\tau)^\top$.
\end{algorithmic}
\end{algorithm}

\section{Theory}
\label{sec: theory}

In this section, we investigate the asymptotic properties of the proposed estimators. We first establish error bounds for the estimated network structure and then study consistency and asymptotic normality of the regression estimators. Let $(\bm\alpha^*,\bm Z^*)$ be the true network parameters, where
$\bm Z^{*\top}\bm1_n=\bm0$. Write
$\theta_{ij}^*=\alpha_i^*+\alpha_j^*+\bm z_i^{*\top}\bm z_j^*$,
$\bm P^*=(p_{ij}^*)=\bm P(\bm\alpha^*,\bm Z^*)$ and $\bm W^*=\bm W(\bm P^*)$. We measure
sparsity by $d_n=n\rho_n$, where $\rho_n=2\{n(n-1)\}^{-1}\sum_{1\le i<j\le n}p_{ij}^*$ is the average edge probability. 

Because $(\bm Z,\bm\eta)$ and $(\bm Z\bm O,\bm O^\top\bm\eta)$ represent
the same models for any orthogonal $\bm O$, we align estimated latent
positions and their coefficients when comparing them with the truth.
For any estimates $\widehat{\bm Z}^{\rm raw}$ and $\widehat{\bm\eta}^{\rm raw}$,
let $\widehat{\bm O}\in\mathbb R^{r\times r}$ minimize
$\|\widehat{\bm Z}^{\rm raw}\bm O-\bm Z^*\|_F$ subject to $\bm O^\top\bm O=\bm I_r$,
and define $\widehat{\bm Z}=\widehat{\bm Z}^{\rm raw}\widehat{\bm O}$ and
$\widehat{\bm\eta}=\widehat{\bm O}^\top\widehat{\bm\eta}^{\rm raw}$.
The estimates $\widehat{\bm\alpha}$, $\widehat{\bm\beta}$, $\widehat\gamma$ and
$\widehat\tau$ are unchanged by this transformation. Set
$\widehat{\bm P}=\bm P(\widehat{\bm\alpha},\widehat{\bm Z})$ and
$\widehat{\bm W}=\bm W(\widehat{\bm P})$, and write
$\widehat{\bm\zeta}=(\widehat{\bm\beta}^\top,\widehat{\bm\eta}^\top,\widehat\gamma,\widehat\tau)^\top$
for the aligned response estimate.

\subsection{Network recovery}

We first study the asymptotic properties of the network estimator. The dimensions $p$ and $r\ge1$ are fixed. We impose the following assumption on the network parameters.

\begin{assumption}\label{assump: network}
For fixed constants $m_0>0$, $R_Z>m_0$ and $c_Z>0$,
\begin{align}
\label{eq: hetero primitive bounds}
 \max_i|\alpha_i^*-\tfrac12\logit(\rho_n)|\le b_n-m_0,\
 \|\bm Z^*\|_{2,\infty}\le R_Z-m_0,\
 \lambda_{\min}(\bm Z^{*\top}\bm Z^*)\ge c_Zn,
\end{align}
where $0<\rho_n<1$ and $b_n\ge\max(1,m_0)$ is a deterministic sequence.
\end{assumption}

Assumption~\ref{assump: network} accommodates both network sparsity and node heterogeneity via the node effects $\alpha_i^*, i\in[n]$. When $\rho_n\to0$, $\logit(\rho_n)\to-\infty$. Since the latent positions $\bm{z}_i^*,i\in[n]$ are bounded, we have $\min_i \alpha_i^*\to -\infty$. Note that $b_n$ can increase with $n$, which permits growing deviations of $\alpha_i^*$ from the center $\logit(\rho_n)/2$. Thus, the node effects may diverge at different rates as $n$ increases, which results in node heterogeneity. Next, define
\begin{equation}
\label{eq: hetero scales}
\begin{aligned}
 w_{-,n}&=\min_{|u-\logit(\rho_n)|\le2b_n+2+R_Z^2}\expit(u)\{1-\expit(u)\},\\
 w_{+,n}&=\max_{|u-\logit(\rho_n)|\le2b_n+2+R_Z^2}\expit(u)\{1-\expit(u)\},
\end{aligned}
\end{equation}
where $w_{-,n}$ and $w_{+,n}$ bound the Bernoulli information over the specified logit interval, since $\expit(u)\{1-\expit(u)\}$ is the Fisher information for the logit parameter $u$ of a Bernoulli variable. Let $s_n=nw_{-,n}$ and $K_n=w_{+,n}/w_{-,n}$. For an integer $k\ge1$, define
\begin{align}
\label{eq: hetero rates definitions}
 \delta_{n,k}
 =\left\{K_n^{2k-1}+\sqrt n(K_n/s_n)^{k/2}\right\}
 \sqrt{K_n\log n/s_n}.
\end{align}
The quantity $\delta_{n,k}$ combines the network information scale $s_n$ and the heterogeneity ratio $K_n$ and plays a key role in the asymptotic analysis of the network estimator. We further impose the following growth condition.

\begin{assumption}\label{assump: network growth}
For some fixed integer $k\ge1$,
\begin{align}
\label{eq: hetero growth}
 K_n^2\delta_{n,k}^4s_n\longrightarrow0.
\end{align}
\end{assumption}

Assumption~\ref{assump: network growth} ensures sufficient network information for the recovery bounds below. In particular, Lemma~\ref{lem: network consequences} in the Supplementary Material shows that it implies $s_n\to\infty$ and $\delta_{n,k}\to0$. We then have the following result.

\begin{theorem}\label{thm: first stage}
Suppose Assumptions~\ref{assump: network} and~\ref{assump: network growth} hold. Then, with probability tending to one, there exist aligned local maximizers $(\widehat{\bm\alpha},\widehat{\bm Z})$ of $L_A$ whose nodewise errors $\Delta_i=\{(\widehat\alpha_i-\alpha_i^*)^2+\|\widehat{\bm z}_i-\bm z_i^*\|_2^2\}^{1/2}$ satisfy
\begin{align}
\label{eq: first stage recovery}
 \max_i\Delta_i=O_p(\delta_{n,k}),\quad
 \frac1n\sum_i\Delta_i^2=O_p(s_n^{-1}).
\end{align}
Moreover,
\begin{equation}
\label{eq: propagation rates}
\begin{aligned}
 \|\widehat{\bm P}-\bm P^*\|_{\max}
 =O_p(K_ns_n\delta_{n,k}/n)&,\quad
 \|\widehat{\bm P}-\bm P^*\|_F
 =O_p(K_ns_n^{1/2}),\\
 \|\bm M_{\widehat{\bm Z}}-\bm M_{\bm Z^*}\|_F
 =O_p(s_n^{-1/2})&,\quad
 \|\widehat{\bm W}-\bm W^*\|_F
 =O_p(K_ns_n^{-1/2}).
\end{aligned}
\end{equation}
\end{theorem}

Theorem~\ref{thm: first stage} gives maximum and mean-square error bounds for the aligned network parameters and transfers these bounds to the probability matrix, latent subspace and normalized weight matrix. In particular, suppose $b_n=O(1)$ and $\rho_n\le1-c_\rho$ for a fixed $c_\rho\in(0,1)$, so that $K_n=O(1)$ and $s_n\asymp d_n$. For dense networks with $\rho_n\asymp1$, taking $k=1$ gives $\max_i\Delta_i=O_p\{\sqrt{\log n/n}\}$ by the first bound in \eqref{eq: first stage recovery}. This is the fastest rate provided by Theorem~\ref{thm: first stage}, which is sharper than the uniform bound $O_p(n^{-1/2+\varepsilon})$ for any fixed $\varepsilon>0$ in \citet{nldaLi2025Inference}. Regarding sparsity, retain the preceding bounded-heterogeneity conditions and suppose $d_n\ge c_dn^{\varepsilon_d}$ for fixed $c_d>0$ and any fixed $\varepsilon_d\in(0,1)$. Taking a fixed integer $k\ge1/\varepsilon_d$ yields $\delta_{n,k}=O\{\sqrt{\log n/d_n}\}$ and ensures \eqref{eq: hetero growth}, since $\log^2n/d_n\to0$. Thus uniform parameter consistency holds when $\rho_n\ge c_dn^{-1+\varepsilon_d}$ for any fixed $\varepsilon_d\in(0,1)$, matching the sparsity regime in \citet{nldaLi2025Inference}. Beyond this bounded-heterogeneity setting, Theorem~\ref{thm: first stage} allows $b_n$ to increase with $n$, subject to Assumption~\ref{assump: network growth}, thereby accommodating growing node-effect heterogeneity. In comparison, \citet{nldaLi2025Inference} assume uniformly bounded node effects after separating a common sparsity offset. 

Our theory concerns local likelihood estimators. Establishing convergence guarantees for Algorithm~\ref{alg:step1} with both network sparsity and increasing node-effect heterogeneity would require additional conditions and a separate analysis of initialization and optimization. We leave this investigation for future work. The simulation results provide empirical support for the finite-sample performance of the proposed two-step procedure.

\subsection{Response inference}

We next study the asymptotic properties of the regression estimator obtained by plugging the estimated network structure into the response quasi-likelihood. Let $\bm\zeta^*=(\bm\beta^{*\top},\bm\eta^{*\top},\gamma^*,\tau^*)^\top$ be the true response parameter and $\widehat{\bm\zeta}=(\widehat{\bm\beta}^\top,\widehat{\bm\eta}^\top,\widehat\gamma,\widehat\tau)^\top$ the aligned output of Algorithm~\ref{alg:step2}, where $\tau^*=1/\sigma^{*2}$. Write $\bm S^*=\bm I_n-\gamma^*\bm W^*$ and $\bm\mu_y^*=\bm S^{*-1}(\bm X\bm\beta^*+\bm Z^*\bm\eta^*)$. We use $\mathbb E_\epsilon$, $\operatorname{Var}_\epsilon$ and $\operatorname{Cov}_\epsilon$ for moments over the response errors. More assumptions are required.

\begin{assumption}\label{assump: response}
There exist positive constants $C_X$ and $c_X$, such that $\|\bm X\|_{2,\infty}\le C_X$ and $\lambda_{\min}\{n^{-1}(\bm X,\bm Z^*)^\top(\bm X,\bm Z^*)\}\ge c_X$. The true response parameter $\bm\zeta^*$ belongs to a fixed compact subset of $\mathbb R^{p+r}\times(-1+\delta_\gamma,1-\delta_\gamma)\times(0,\infty)$. For some fixed $\kappa_\gamma>0$,
\begin{align}
\label{eq: gamma identification}
 \liminf_{n\to\infty}\frac1n
 \|\bm M_{(\bm X,\bm Z^*)}\bm W^*\bm\mu_y^*\|_2^2
 \ge\kappa_\gamma.
\end{align}
\end{assumption}

\begin{assumption}\label{assump: response errors}
The errors $\epsilon_1,\ldots,\epsilon_n$
are independent and identically distributed, independent of the network observations,
and, for a fixed $C_\epsilon>0$, satisfy
\[
 \mathbb E_\epsilon\epsilon_i=0,\ 
 \mathbb E_\epsilon\epsilon_i^2=\sigma^{*2},\ 
 \mathbb E_\epsilon\exp\{\epsilon_i^2/(C_\epsilon^2\sigma^{*2})\}\le2.
\]
The error distribution may depend on $n$.
\end{assumption}

Assumption~\ref{assump: response} bounds the covariates and prevents asymptotic collinearity among the columns of $(\bm X,\bm Z^*)$, allowing the covariate and latent-position coefficients to be distinguished. The compactness condition keeps $\gamma^*$ away from $\pm1$ and the error variance bounded above and away from zero. Condition~\eqref{eq: gamma identification} indicates that the information about $\gamma^*$ cannot be absorbed by the coefficients of $\bm X$ and $\bm Z^*$.  The exponential-moment condition in Assumption~\ref{assump: response errors} is equivalent to a uniform sub-Gaussian bound on the standardized errors $\epsilon_i/\sigma^*$, which implies bounded standardized moments of every fixed order.

Let $\bm q_y(\bm\zeta;\bm P,\bm Z)=\nabla_{\bm\zeta}L_y(\bm\zeta;\bm P,\bm Z)$
be the normalized response score of \eqref{eq: response quasi likelihood}.
Define the expected negative Jacobian and scaled score covariance at the truth by
\begin{align}
\label{eq: response information definitions}
 \mathsf J_{y,n}=-\mathbb E_\epsilon\{\nabla_{\bm\zeta}\bm q_y(\bm\zeta^*;\bm P^*,\bm Z^*)\},\quad 
 \bm\Xi_n=n\operatorname{Var}_\epsilon\{\bm q_y(\bm\zeta^*;\bm P^*,\bm Z^*)\}.
\end{align}
When $\mathsf J_{y,n}$ is invertible, write
\begin{align}
\label{eq: joint response covariance}
 \bm V_n=\mathsf J_{y,n}^{-1}\bm\Xi_n\mathsf J_{y,n}^{-\top}.
\end{align}
For joint inference, we estimate the covariance matrix of $\widehat{\bm\zeta}$ by
\begin{align}
\label{eq: joint feasible covariance}
 \widehat{\operatorname{Var}}(\widehat{\bm\zeta})
 =n^{-1}\widehat{\bm V}_n,\quad
 \widehat{\bm V}_n=\widehat{\mathsf J}_{y,n}^{-1}
 \widehat{\bm\Xi}_n\widehat{\mathsf J}_{y,n}^{-\top}.
\end{align}
Here $\widehat{\mathsf J}_{y,n}$ and $\widehat{\bm\Xi}_n$ estimate $\mathsf J_{y,n}$ and $\bm\Xi_n$, respectively, and their sandwich combination $\widehat{\bm V}_n$ estimates $\bm V_n$. The detailed construction is given in Section I of the Supplementary Material, with an explicit formula for $\widehat{\bm\Xi}_n$ in Lemma~\ref{lem: fitted score covariance}. For joint inference, let $\bm C_n=(\bm X,\bm Z^*,\bm W^*\bm\mu_y^*)$ and $h_n=n^{-1}\bm1_n^\top(\bm I_n-\bm M_{\bm C_n})\bm1_n\in[0,1]$. Let $\mu_3=\mathbb E_\epsilon(\epsilon_i^3)$ and $\mu_4=\mathbb E_\epsilon(\epsilon_i^4)$ denote the third and fourth error moments, which may depend on $n$. The additional conditions for joint asymptotic normality are collected below.
\begin{assumption}\label{assump: joint inference}
The error moments $\mu_3,\mu_4$ and the design matrix $\bm C_n$ satisfy
\begin{align}
\label{eq: response nondegeneracy}
 \liminf_{n\to\infty}
 \left(\mu_4-\sigma^{*4}-\frac{\mu_3^2}{\sigma^{*2}}h_n\right)>0.
\end{align}
The network information also satisfies
\begin{align}
\label{eq: inference growth}
 \frac{\sqrt nK_n^2}{s_n}\longrightarrow0.
\end{align}
\end{assumption}
Condition~\eqref{eq: response nondegeneracy} prevents the joint covariance matrix $\bm V_n$ from approaching singularity, thereby ensuring nondegenerate joint inference. Lemma~\ref{lem: response conditioning} of the Supplementary Material shows that this condition keeps the smallest eigenvalue of $\bm V_n$ bounded away from zero. It holds automatically for Gaussian errors: $\mu_3=0$ and $\mu_4=3\sigma^{*4}$, so the expression in \eqref{eq: response nondegeneracy} equals $2\sigma^{*4}$, which is bounded away from zero by Assumption~\ref{assump: response}. Condition~\eqref{eq: inference growth} imposes a further restriction on the relative growth of $s_n$ and $K_n$. Then we have the following result.

\begin{theorem}\label{thm: plugin inference}
Suppose Assumptions~\ref{assump: network}--\ref{assump: response errors} hold. Then the aligned response output satisfies
\begin{align}
\label{eq: response rate all degrees}
 \|\widehat{\bm\zeta}-\bm\zeta^*\|_2
 =O_p(n^{-1/2}+K_n^2/s_n).
\end{align}
If Assumption~\ref{assump: joint inference} also holds, then
\begin{align}
\label{eq: plugin CLT}
 \widehat{\bm V}_n^{-1/2}\sqrt n
 (\widehat{\bm\zeta}-\bm\zeta^*)
 \to_d N(\bm0,\bm I_{p+r+2}).
\end{align}
\end{theorem}

Theorem~\ref{thm: plugin inference} establishes the convergence rate and joint asymptotic normality of the aligned regression estimator. Specifically, the rate in \eqref{eq: response rate all degrees} combines response variation of order
$n^{-1/2}$ and network-estimation error of order $K_n^2/s_n$.
Condition~\eqref{eq: inference growth} makes the latter negligible at the root-$n$
scale. Consequently, the aligned regression estimator attains the classical root-$n$ convergence rate established for spatial autoregressive models with known weights, including the quasi-maximum likelihood estimator of \citet{lee2004asymptotic} and the least squares estimator of \citet{huang2019least}, subject to their respective regularity conditions. Our result retains this rate even though both the weight matrix $\bm W^*$ and the latent positions $\bm Z^*$ must be estimated.

Our sufficient conditions for attaining this rate impose stronger restrictions on network sparsity and node-effect heterogeneity than those for consistency alone. The bound in \eqref{eq: response rate all degrees} gives root-$n$ convergence when $K_n^2/s_n=O(n^{-1/2})$, while \eqref{eq: inference growth} requires this term to be $o(n^{-1/2})$ for the joint normal approximation in \eqref{eq: plugin CLT}. In the bounded-heterogeneity setting described after Theorem~\ref{thm: first stage}, $K_n=O(1)$ and $s_n\asymp d_n$, so \eqref{eq: inference growth} becomes $d_n/\sqrt n\to\infty$, or equivalently $\sqrt n\rho_n\to\infty$. Thus, this sufficient condition allows sparse networks but requires a higher density than that needed for network consistency. Increasing node-effect heterogeneity can reduce $s_n$ and increase $K_n$, further tightening the condition.

\section{Simulation}
\label{sec: simulation}

In this section, we examine the numerical performance of the proposed NLDA method via extensive simulated experiments. We first compare the estimation performance of the proposed NLDA method with some popular existing methods, including the SIM of \cite{bramoulle2009identification}, the RNC of \cite{li2019prediction} and the SP method of \cite{le2022linear}. Next, we evaluate the validity of the inference framework of the proposed NLDA method.

%The estimates in the simulations and data analysis are computed with
%Algorithms~\ref{alg:step1} and~\ref{alg:step2}. For Algorithm~\ref{alg:step1},
%we use $(c_{\tau,0},\xi_0,T)=(2.0,0.75,60)$ in the estimation simulations
%and $(2.5,1.0,500)$ in the inference simulations and the full real-data fit.
%In the simulations, we use $\Gamma_\gamma=[-0.95,0.95]$ with $2n$ equally spaced
%subintervals and the precision profile
%$\tau_\gamma=n/\operatorname{RSS}_n(\gamma)$.

\subsection{Estimation performance}
\label{subsec:sim-estimation}
% In this subsection, we evaluate the practical performance of our method in different settings of network effects and compare it with other three benchmark methods. They are the SIM of \cite{lee2007identification}, the RNC of \cite{li2019prediction}, and the SP of \cite{le2022linear}.

Let $\bm\mu_y^*$ denote the true mean response under the data-generating model and $\widehat{\bm\mu}_y$ its estimate from each method. For the proposed model \eqref{Model: PMNLD}, the true mean is $\bm\mu_y^*=\bm S^{*-1}(\bm X\bm\beta^*+\bm Z^*\bm\eta^*)$, and the NLDA estimate is $\widehat{\bm\mu}_y=\widehat{\bm S}^{-1}(\bm X\widehat{\bm\beta}+\widehat{\bm Z}\widehat{\bm\eta})$, where $\widehat{\bm S}=\bm I_n-\widehat\gamma\widehat{\bm W}$. Following \citet{le2022linear}, we compare the methods using the relative mean squared error 
\[
\mathrm{RelMSE}(\widehat{\bm\mu}_y,\bm\mu_y^*)
=
\frac{\|\widehat{\bm\mu}_y-\bm\mu_y^*\|_2^2}{\|\bm\mu_y^*\|_2^2}.
\]
We consider three synthetic examples to evaluate the performance of all methods in different aspects.

\textbf{Example 1.} With $p=8$ and $r=3$, we first generate $(\bm B_Z)_{ij}\overset{\text{iid}}{\sim}\mathcal N(0,1)$ for $i\in [n]$ and $j\in [r]$, compute the left and right singular vector matrices $\bm U_Z$ and $\bm V_Z$ in the thin singular value decomposition of $\bm J_n\bm B_Z$, and set $\bm Z^*=\sqrt{n/r}\bm U_Z\bm V_Z^\top$. Next, we generate $\widetilde X_{ij}\overset{\text{iid}}{\sim}\mathcal N(0,1)$ for $i\in [n]$ and $j\in [p]$, and set $\bm X=\bm J_n\widetilde{\bm X}$. We then draw $\alpha_i^*\overset{\text{iid}}{\sim}\mathcal N(\mu,1)$ for each $i\in[n]$. Let $\bm P^*$ be determined by \eqref{Model: Latent}, and we sample $\bm A$ from $\bm P^*$. Finally, we draw $\bm\beta^*\sim\mathcal N_p(\bm 0,0.8^2p^{-1}\bm I_p)$, $\bm\eta^*\sim\mathcal N_r(\bm 0,r^{-1}\bm I_r)$, and $\bm\epsilon\sim\mathcal N_n(\bm 0,\bm I_n)$ independently, and generate $\bm y$ from \eqref{Model: PMNLD} with $\gamma^*=0.8$. We vary $\mu\in\{-2.5,-2.0,-1.5\}$ to generate networks with densities ranging approximately from 0.02 to 0.1.

\textbf{Example 2.} The data-generating scheme is the same as in Example 1, except that we vary the latent structure strength $s_z\in\{1,1.5,2\}$ and generate $\bm Z^*=s_z\sqrt{n/r}\bm U_Z\bm V_Z^\top$. We then choose $\mu$ such that the expected network density $\{n(n-1)\}^{-1}\sum_{i\ne j} p_{ij}^* = 0.05$. All other parameters are the same as in Example 1. This yields sparse networks with different strengths of latent low-rank structure of the network data.

\textbf{Example 3.} We generate the network and covariates as in Example 1 with $\mu=-2.0$, but generate the response according to the SIM in \eqref{Model: SIM}, where $\bm\beta^*$ and $\bm\epsilon$ are generated as in Example 1 and $\bm\eta_{\rm SIM}^*\sim\mathcal N_p(\bm 0,p^{-1}\bm I_p)$.
We vary $\gamma^*\in\{0.2,0.4,0.6\}$ to evaluate the robustness of the proposed NLDA method under model misspecification.

We consider $n\in\{1000,2000,4000\}$ in all examples. The averaged $\mathrm{RelMSE}(\widehat{\bm\mu}_y,\bm\mu_y^*)$ and the corresponding standard errors over 50 independent replications are reported in Tables \ref{table:example_1_relmse}-\ref{table:example_3_relmse}.

\begin{table}[!htb]
\centering
\small
\caption{Averaged relative mean squared errors $\mathrm{RelMSE}(\widehat{\bm\mu}_y,\bm\mu_y^*)$ and their standard errors of all competing methods over 50 independent replications in Example 1. The best performer in each scenario is boldfaced.}
\begin{tabularx}{\textwidth}{c c >{\centering\arraybackslash}X >{\centering\arraybackslash}X >{\centering\arraybackslash}X >{\centering\arraybackslash}X}
\toprule
$\mu$ & $n$  & NLDA & SP & RNC & SIM \\
\midrule
\multirow{3}{*}{$\mu=-2.5$} 
  & 1000  & \textbf{0.2003} (0.0032) & 0.2258 (0.0034) & 0.3525 (0.0104) & 0.2974 (0.0046) \\
  & 2000  & \textbf{0.1205} (0.0007) & 0.1811 (0.0011) & 0.3504 (0.0022) & 0.3048 (0.0011) \\
  & 4000  & \textbf{0.0645} (0.0003) & 0.1440 (0.0007) & 0.2492 (0.0013) & 0.3133 (0.0040) \\
\midrule
\multirow{3}{*}{$\mu=-2.0$} 
  & 1000  & \textbf{0.1335} (0.0030) & 0.1870 (0.0036) & 0.3549 (0.0060) & 0.3184 (0.0055) \\
  & 2000  & \textbf{0.0708} (0.0015) & 0.1355 (0.0011) & 0.2545 (0.0018) & 0.3048 (0.0043) \\
  & 4000  & \textbf{0.0337} (0.0003) & 0.1053 (0.0007) & 0.1942 (0.0012) & 0.2814 (0.0022) \\
\midrule
\multirow{3}{*}{$\mu=-1.5$} 
  & 1000  & \textbf{0.0608} (0.0005) & 0.1064 (0.0027) & 0.2371 (0.0040) & 0.2953 (0.0079) \\
  & 2000  & \textbf{0.0341} (0.0003) & 0.0875 (0.0012) & 0.2015 (0.0015) & 0.3113 (0.0066) \\
  & 4000  & \textbf{0.0169} (0.0001) & 0.0723 (0.0011) & 0.1899 (0.0016) & 0.3170 (0.0044) \\
\bottomrule
\end{tabularx}
\label{table:example_1_relmse}
\end{table}

\begin{table}[!htb]
\centering
\small
\caption{Averaged relative mean squared errors $\mathrm{RelMSE}(\widehat{\bm\mu}_y,\bm\mu_y^*)$ and their standard errors of all competing methods over 50 independent replications in Example 2. The best performer in each scenario is boldfaced.}
\begin{tabularx}{\textwidth}{c c >{\centering\arraybackslash}X >{\centering\arraybackslash}X >{\centering\arraybackslash}X >{\centering\arraybackslash}X}
\toprule
$s_z$ & $n$  & NLDA & SP & RNC & SIM \\
\midrule
\multirow{3}{*}{$s_z=1.0$} 
  & 1000  & \textbf{0.1025} (0.0010) & 0.1496 (0.0050) & 0.2928 (0.0072) & 0.2946 (0.0068) \\
  & 2000  & \textbf{0.0539} (0.0004) & 0.1135 (0.0010) & 0.2263 (0.0016) & 0.2838 (0.0023) \\
  & 4000  & \textbf{0.0314} (0.0005) & 0.1040 (0.0009) & 0.1892 (0.0016) & 0.2751 (0.0018) \\
\midrule
\multirow{3}{*}{$s_z=1.5$} 
  & 1000  & \textbf{0.0673} (0.0008) & 0.2072 (0.0016) & 0.2046 (0.0027) & 0.6714 (0.0134) \\
  & 2000  & \textbf{0.0337} (0.0003) & 0.1888 (0.0007) & 0.1734 (0.0012) & 0.6082 (0.0032) \\
  & 4000  & \textbf{0.0168} (0.0001) & 0.1682 (0.0003) & 0.2165 (0.0009) & 0.6752 (0.0072) \\
\midrule
\multirow{3}{*}{$s_z=2.0$} 
  & 1000  & \textbf{0.0457} (0.0004) & 0.2281 (0.0007) & 0.1026 (0.0013) & 0.7389 (0.0048) \\
  & 2000  & \textbf{0.0234} (0.0002) & 0.2205 (0.0003) & 0.1131 (0.0008) & 0.6775 (0.0025) \\
  & 4000  & \textbf{0.0125} (0.0001) & 0.2041 (0.0002) & 0.1786 (0.0009) & 0.7704 (0.0013) \\
\bottomrule
\end{tabularx}
\label{table:example_2_relmse}
\end{table}

\begin{table}[!htb]
\centering
\small
\caption{Averaged relative mean squared errors $\mathrm{RelMSE}(\widehat{\bm\mu}_y,\bm\mu_y^*)$ and their standard errors of all competing methods over 50 independent replications in Example 3. The best performer in each scenario is boldfaced.}
\begin{tabularx}{\textwidth}{c c >{\centering\arraybackslash}X >{\centering\arraybackslash}X >{\centering\arraybackslash}X >{\centering\arraybackslash}X}
\toprule
$\gamma^*$ & $n$  & NLDA & SP & RNC & SIM \\
\midrule
\multirow{3}{*}{$\gamma^*=0.2$} 
  & 1000  & 0.0406 (0.0006) & 0.0418 (0.0007) & 0.1345 (0.0117) & \textbf{0.0208} (0.0079) \\
  & 2000  & 0.0195 (0.0003) & 0.0198 (0.0003) & 0.1364 (0.0142) & \textbf{0.0187} (0.0057) \\
  & 4000  & 0.0088 (0.0001) & 0.0092 (0.0002) & 0.1138 (0.0101) & \textbf{0.0055} (0.0015) \\
\midrule
\multirow{3}{*}{$\gamma^*=0.4$} 
  & 1000  & 0.0407 (0.0005) & 0.0425 (0.0007) & 0.3022 (0.0163) & \textbf{0.0207} (0.0035) \\
  & 2000  & 0.0181 (0.0003) & 0.0190 (0.0003) & 0.2429 (0.0061) & \textbf{0.0127} (0.0024) \\
  & 4000  & 0.0124 (0.0013) & 0.0096 (0.0002) & 0.1609 (0.0026) & \textbf{0.0079} (0.0020) \\
\midrule
\multirow{3}{*}{$\gamma^*=0.6$} 
  & 1000  & \textbf{0.0591} (0.0010) & 0.0634 (0.0010) & 0.4020 (0.0054) & 0.0972 (0.0273) \\
  & 2000  & 0.0208 (0.0003) & 0.0235 (0.0006) & 0.2847 (0.0016) & \textbf{0.0093} (0.0016) \\
  & 4000  & 0.0108 (0.0001) & 0.0119 (0.0002) & 0.1767 (0.0007) & \textbf{0.0081} (0.0014) \\
\bottomrule
\end{tabularx}
\label{table:example_3_relmse}
\end{table}

Tables~\ref{table:example_1_relmse} and~\ref{table:example_2_relmse} show that NLDA achieves the smallest average RelMSE in every setting of Examples 1 and 2. Its RelMSE also decreases as $n$ increases across all three examples; in the correctly specified Examples 1 and 2, this trend is qualitatively consistent with the convergence result in Theorem~\ref{thm: plugin inference}. In Example 1, the estimation error of NLDA decreases substantially as the network becomes denser. SP also improves with network density, while the improvement of RNC is less consistent. SIM remains substantially less accurate than NLDA throughout. In Example 2, increasing $s_z$ reduces the estimation error of NLDA and widens its advantage over SP and SIM. RNC also benefits from stronger latent structure but remains less accurate than NLDA. These results indicate that NLDA can use the latent structure to improve estimation. In Example 3, Table~\ref{table:example_3_relmse} shows that SIM achieves the smallest average RelMSE in most settings, as expected when the data are generated from SIM. Among the misspecified methods, NLDA generally has smaller average RelMSE than SP and RNC. Thus, NLDA remains competitive among the misspecified methods.

\subsection{Inference performance}

We construct 95\% confidence intervals for $\bm\zeta^*$ using the fitted response coefficients and the joint covariance estimator in \eqref{eq: joint feasible covariance}. In each replication, we use the known $\bm Z^*$ to align the raw estimator of $\bm\eta$ and transform its estimated covariance accordingly, as described in Section~\ref{sec: theory}. Coverage for $\bm\eta$ is reported in the aligned coordinates studied in Theorem~\ref{thm: plugin inference}; inference for $\bm\beta$, $\gamma$ and $\tau$ is unchanged by this alignment. The data-generating scheme is the same as in Example~1, except that we vary $\mu\in\{-1.0,-0.75,-0.5\}$. We consider $n\in\{1000,2000\}$. Bias and standard deviation (SD) are computed from the estimation errors. The averaged coverage probability, bias and SD over all $p+r+2=13$ coordinates and 500 independent replications are reported in Table~\ref{table:zeta_inference_mu}.

\begin{table}[!htb]
\centering
\small
\caption{Empirical coverage and estimation error for $\bm\zeta^*$ over 500 independent replications, with latent coordinates aligned to the truth. Entries in the error column are $|\mathrm{Bias}|(\mathrm{SD})$.}
\setlength{\tabcolsep}{3.5pt}
\begin{tabular}{c c c c}
\toprule
 & & \multicolumn{2}{c}{NLDA} \\
\cmidrule(lr){3-4}
$\mu$ & $n$ & Coverage & $|\mathrm{Bias}|(\mathrm{SD})$ \\
\midrule
\multirow{2}{*}{$\mu=-1.0$} & 1000 & 0.9242 & 0.0321 (0.0829) \\
  & 2000 & 0.9343 & 0.0314 (0.0720) \\
\midrule
\multirow{2}{*}{$\mu=-0.75$} & 1000 & 0.9371 & 0.0295 (0.0800) \\
  & 2000 & 0.9403 & 0.0256 (0.0639) \\
\midrule
\multirow{2}{*}{$\mu=-0.5$} & 1000 & 0.9369 & 0.0264 (0.0759) \\
  & 2000 & 0.9482 & 0.0197 (0.0561) \\
\bottomrule
\end{tabular}
\label{table:zeta_inference_mu}
\end{table}

Table~\ref{table:zeta_inference_mu} shows empirical coverage ranging from 92.4\% to 94.8\%. Although coverage remains below the nominal 95\% level, it generally approaches this level as either $n$ or $\mu$ increases. This pattern is consistent with more accurate network estimation: larger $n$ provides more network observations, while larger $\mu$ yields denser networks in the settings considered. Improved network estimation reduces its contribution to the regression estimation error, helping the normal approximation become more accurate. For each value of $\mu$, the absolute bias and SD also decrease as $n$ increases.

\section{Diplomatic network-linked data}

We now apply the proposed method to a diplomatic network-linked dataset constructed from the Correlates of War Diplomatic Exchange data \citep{bayer2006diplomatic} and the World Bank World Development Indicators. The World Bank data are publicly available at \url{https://databank.worldbank.org/source/world-development-indicators}. We use the diplomatic network and country characteristics in 2005. For two distinct countries \(i\) and \(j\), we set \(a_{ij}=a_{ji}=1\) if the original binary exchange indicator records diplomatic representation between the two countries, and zero otherwise, with \(a_{ii}=0\). Restricting attention to countries with complete response and covariate information, we obtain a connected network with adjacency matrix $\bm A$, containing 187 countries and 4,507 edges, with density being 0.2592. Moreover, we construct the response vector $\bm y$ and covariate matrix $\bm X$ as in Table~\ref{tab:diplomacy_yx}. In practice, we apply the natural logarithm to all variables. Thus, $y_i$ is the natural logarithm of life expectancy at birth in years, and $\bm x_i$ includes an intercept together with the centered and scaled versions of log GDP per capita, log population and log urban population percentage, each with empirical mean zero and second moment one.

\begin{table}[!htb]
\centering
\small
\caption{Node-level response and covariates in the diplomatic network analysis.}
\label{tab:diplomacy_yx}
\begin{tabularx}{\textwidth}{c l >{\raggedright\arraybackslash}X}
\toprule
Role & Quantity & Construction \\
\midrule
Response & Life expectancy & Life expectancy at birth in years in 2005. \\
Covariate & GDP per capita & GDP per capita in 2005 (constant 2015 US\$). \\
Covariate & Population & Total population in 2005. \\
Covariate & Urban population percentage & Share of population living in urban areas in 2005. \\
\bottomrule
\end{tabularx}
\end{table}

We select the latent dimension $r=3$ by held-out link prediction \citep{zhang2022joint}. We then assess how well the fitted network model captures regional connectivity in the diplomatic graph. Figure~\ref{fig:diplomacy_network_fit} compares observed edge frequencies with average fitted edge probabilities among five groups of countries defined by continent. The similar patterns in the two panels suggest that the fitted model captures the main regional connectivity structure.

\begin{figure}[!htb]
\centering
\includegraphics[width=0.85\textwidth]{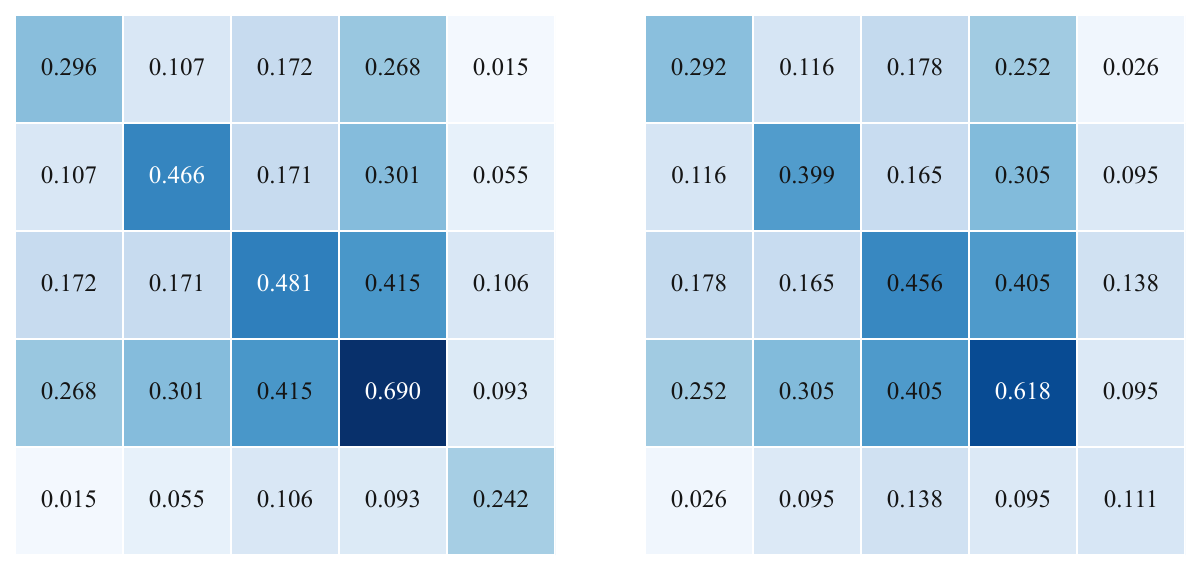}
\caption{Observed edge frequencies (left) and mean fitted edge probabilities (right), grouped by continent: Africa, Americas, Asia, Europe and Oceania, in this order.}
\label{fig:diplomacy_network_fit}
\end{figure}

We compare the predictive performance of NLDA, SP, RNC and SIM over 50 random 80:20 splits of the countries into training and test sets. For NLDA, we fit the network model \eqref{Model: Latent} using the whole network. For each split, we fit the response model \eqref{Model: PMNLD} using only the training covariates and responses together with the fitted network model. For prediction, we use the fitted models and jointly solve the response equations for the test countries. Let $\ell_i=\exp(y_i)$ and $\widehat\ell_i=\exp(\widehat y_i)$ denote the observed and predicted life expectancies in years. We finally evaluate the prediction accuracy on the testing set by relative mean squared error and out-of-sample coefficient of determination:
\[
%\mathrm{rMAE}=\frac{\sum_{i\in\mathcal T}|\widehat\ell_i-\ell_i|}{\sum_{i\in\mathcal T}\ell_i},\mbox{ }
\mathrm{rMSE}=\frac{\sum_{i\in\mathcal T}(\widehat\ell_i-\ell_i)^2}{\sum_{i\in\mathcal T}\ell_i^2},\mbox{ }
R^2=1-\frac{\sum_{i\in\mathcal T}(\widehat\ell_i-\ell_i)^2}{\sum_{i\in\mathcal T}(\ell_i-\bar\ell_{\mathrm{tr}})^2},
\]
where $\mathcal T$ denotes the test set and $\bar\ell_{\mathrm{tr}}$ is the mean life expectancy in the training set. Table~\ref{table:diplomacy_prediction} shows that NLDA achieves the best average predictive performance among the four methods on both measures: the smallest rMSE and the largest $R^2$. Paired comparisons across the 50 splits support these improvements over each competing method, with all one-sided Wilcoxon $p$-values below 0.001.

\begin{table}[!htb]
\centering
\small
\caption{Averaged prediction measures and their standard errors over 50 random 80:20 train-test splits in the diplomatic network analysis. The best performer is boldfaced.}
\label{table:diplomacy_prediction}
\begin{tabular}{l c c c c}
\toprule
Metric & NLDA & SP & RNC & SIM \\
\midrule
%$\mathrm{rMAE}$ & \textbf{0.0492} (0.0010) & 0.0543 (0.0010) & 0.0562 (0.0012) & 0.0509 (0.0013) \\
$\mathrm{rMSE}$ & \textbf{0.0043} (0.0002) & 0.0051 (0.0002) & 0.0058 (0.0003) & 0.0048 (0.0002) \\
$R^2$ & \textbf{0.7729} (0.0083) & 0.7293 (0.0096) & 0.6973 (0.0117) & 0.7500 (0.0104) \\
\bottomrule
\end{tabular}
\end{table}

Next, we examine inference based on the NLDA fit to all 187 countries. Table~\ref{table:diplomacy_main_results} reports estimates of the observed covariate coefficients and the network-effect coefficient $\gamma$. The significantly positive estimate of $\gamma$ suggests a positive association in life expectancy between countries linked by diplomatic ties, after accounting for the observed covariates and latent positions. This association may reflect exchanges of health knowledge and resources through international cooperation. The estimated coefficients of GDP per capita, population and urban population percentage are also positive, with GDP per capita having the largest coefficient. The joint Wald test of $\bm\eta=\bm0$, whose statistic is invariant to the fitted latent orientation, gives a $p$-value of 0.000475, indicating a significant contribution of the latent positions.

\begin{table}[!htb]
\centering
\small
\caption{NLDA estimates for the diplomatic network with $r=3$.}
\label{table:diplomacy_main_results}
\begin{tabular}{l c c c c}
\toprule
Variable & Estimate & SE & $p$-value & 95\% CI \\
\midrule
GDP per capita & 0.083 & 0.008 & $<0.001$ & $[0.067,\,0.099]$ \\
Population & 0.016 & 0.007 & 0.023 & $[0.002,\,0.030]$ \\
Urban population percentage & 0.019 & 0.008 & 0.015 & $[0.004,\,0.035]$ \\
Network effect $\gamma$ & 0.767 & 0.212 & $<0.001$ & $[0.352,\,1.182]$ \\
\bottomrule
\end{tabular}
\end{table}

%----------------------------------------------------------------------------------------

\section{Summary}

In this paper, we develop a probabilistic regression framework for network-linked data. The NLDA model replaces binary adjacency weights by degree-normalized connection probabilities in the peer-response term and includes latent positions as a network-cohesion component. The theory establishes recovery rates for a local network likelihood maximizer under increasing node-effect heterogeneity and consistency of the resulting plug-in regression estimator. Under the sub-Gaussian error and nondegeneracy conditions, $\sqrt nK_n^2/s_n\to0$ ensures joint asymptotic normality of the aligned regression estimator, with a consistently estimable covariance. When node-effect variation is bounded and the average density stays bounded away from one, this condition reduces to $d_n/\sqrt n\to\infty$. The diplomatic network analysis illustrates the application to network-linked data.

\section*{Acknowledgment}

JZ's research is supported in part by the National Key R\&D Program of China, grant 2024YFA1012200, and the National Natural Science Foundation of China, grant 12301388. DH's research is supported in part by the National Natural Science Foundation of China, grant 72471230; the MOE Project of Key Research Institute of Humanities and Social Sciences, grant 22JJD110001; and Big Data and Responsible Artificial Intelligence for National Governance, Renmin University of China. The first two authors contribute equally to this work.

\putbib[NLDA]
\end{bibunit}

\clearpage
\baselineskip=24pt
\setcounter{page}{1}
\setcounter{equation}{0}
\setcounter{lemma}{0}
\setcounter{theorem}{0}
\setcounter{proposition}{0}
\setcounter{corollary}{0}
\setcounter{assumption}{0}
\renewcommand{\thelemma}{A\arabic{lemma}}
\renewcommand{\thetheorem}{A\arabic{theorem}}
\renewcommand{\theproposition}{A\arabic{proposition}}
\renewcommand{\thecorollary}{A\arabic{corollary}}
\renewcommand{\theassumption}{A\arabic{assumption}}
\renewcommand{\theequation}{A\arabic{equation}}
\begin{center}
{\Large\bf Supplementary Materials for ``Linear regression with probabilistic networks for network-linked data"} \\
\medskip
% {\Large\bf ``"}  \\
% \bigskip
\vspace{0.2in} % comment out this line when unblind the author line
\end{center}
% \bigskip

% \spacingset{1.7} 

\begin{bibunit}[apalike]

This supplement is organized as follows. Section I introduces some notations and definitions required for the following proofs of theorems and lemmas. Section II presents the proofs of Theorems~\ref{thm: first stage} and~\ref{thm: plugin inference}. Section III contains the auxiliary lemmas and their proofs.

\section*{Section I. Notation}
Constants $c,C>0$ may vary between inequalities and depend only on the fixed
quantities in the assumptions and the fixed integer $k$.

\paragraph{Network parameters and derivatives.}
Write $\bm\nu=(\bm\alpha,\bm Z)\in\mathbb R^n\times\mathbb R^{n\times r}$ and
$\bm\nu^*=(\bm\alpha^*,\bm Z^*)$ for the network parameter pairs, and set
$\widehat{\bm\nu}=(\widehat{\bm\alpha},\widehat{\bm Z})$.
We measure their differences in node coordinates, $\psi(\bm\nu) = (\alpha_1,\bm z_1^\top,\ldots,\alpha_n,\bm z_n^\top)^\top \in\mathbb R^{n(r+1)}$,
where the same linear map applies to a perturbation $\bm h=(\bm h_\alpha,\bm E_Z)$. For
functions of the network parameters, write $\bm P(\bm\nu)=\bm P(\bm\alpha,\bm Z)$,
$\bm W(\bm\nu)=\bm W\{\bm P(\bm\nu)\}$ and $L_A(\bm\nu)=L_A(\bm\alpha,\bm Z)$. For $i\ne j$,
write $\theta_{ij}(\bm\nu)$ and $p_{ij}(\bm\nu)$ for the logit and edge probability
in \eqref{Model: Latent}.

For a differentiable scalar-, vector- or matrix-valued function $f$ and
$\bm h=(\bm h_\alpha,\bm E_Z)$, define $D_\nu f(\bm\nu)[\bm h] :=\left.\frac{d}{dt}f(\bm\nu+t\bm h)\right|_{t=0}$.
Here $\bm\nu+t\bm h=(\bm\alpha+t\bm h_\alpha,\bm Z+t\bm E_Z)$;
all other arguments of $f$ are held fixed. The notation $D_P$, $D_W$ and
$D_{(W,Z)}$ similarly specifies differentiation with respect to
$\bm P$, $\bm W$ and $(\bm W,\bm Z)$.
For $i\ne j$, the directional derivative of $\theta_{ij}$ is
\begin{align*}
 \dot\theta_{ij}(\bm\nu)[\bm h]
 :=\left.\frac{d}{dt}\theta_{ij}(\bm\nu+t\bm h)\right|_{t=0}
 =h_{\alpha,i}+h_{\alpha,j}
 +\bm E_{Z,i\cdot}\bm Z_{j\cdot}^\top
 +\bm Z_{i\cdot}\bm E_{Z,j\cdot}^\top.
\end{align*}
At the true parameter, write
$\dot\theta_{ij}^*[\bm h]:=\dot\theta_{ij}(\bm\nu^*)[\bm h]$.

Write $\bm\phi_{ij}(\bm\nu)=\nabla_\psi\theta_{ij}(\bm\nu)$,
$\bm\phi_{ij}^*=\bm\phi_{ij}(\bm\nu^*)$ and
$\bm B_{A,ij}=\nabla_\psi^2\theta_{ij}(\bm\nu)$. The matrix $\bm B_{A,ij}$ is constant,
with $\bm I_r$ in the latent-coordinate positions $(i,j)$ and $(j,i)$ and zeros elsewhere. For
$e=(i,j)$ with $1\le i<j\le n$, a subscript $e$ denotes the
corresponding dyad entry in $a_{ij},\theta_{ij},p_{ij},w_{ij},\bm\phi_{ij}$ or
$\bm B_{A,ij}$. Write $w_{ij}(\bm\nu)=p_{ij}(\bm\nu)\{1-p_{ij}(\bm\nu)\}$ and $w_{ij}^*=w_{ij}(\bm\nu^*)$.

Let $m_n=n(n-1)/2$ and $t_n=nw_{+,n}=K_ns_n$. The probability envelopes are
$p_{\pm,n}=\expit\{\vartheta_n\pm(2b_n+2+R_Z^2)\}$, where $\vartheta_n=\logit(\rho_n)$.
Thus $w_{-,n}$ and $w_{+,n}$ in \eqref{eq: hetero scales}
are the minimum and maximum of $p(1-p)$ over $[p_{-,n},p_{+,n}]$. After the alignment
specified in the main article, write $\bm\Delta_Z=\widehat{\bm Z}-\bm Z^*$,
whose $i$-th row is $\bm\Delta_{z,i}^\top=(\widehat{\bm z}_i-\bm z_i^*)^\top$.

\paragraph{Tangent space and network information.}
Define the subspace of centered, aligned parameter directions by
\begin{align}
\label{eq: network tangent coordinates}
 \mathcal T_n=\{\psi(\bm h):
 \bm E_Z^\top\bm1_n=\bm0,
 \ \bm Z^{*\top}\bm E_Z=\bm E_Z^\top\bm Z^*\}
 \subset\mathbb R^{n(r+1)},
\end{align}
where $\bm h=(\bm h_\alpha,\bm E_Z)$, and let $\bm\Pi_{\mathcal T_n}$ be its Euclidean
projector. The Procrustes alignment in Section~\ref{sec: theory} makes
$\bm Z^{*\top}\widehat{\bm Z}$ symmetric and positive semidefinite. Consequently,
$\bm Z^{*\top}\bm\Delta_Z=\bm\Delta_Z^\top\bm Z^*$.
If $\widehat{\bm Z}^{\rm raw\top}\bm1_n=\bm0$, then
$\bm\Delta_Z^\top\bm1_n=\bm0$ as well. Thus the centered, aligned estimation
error satisfies the constraints in \eqref{eq: network tangent coordinates}.

For a vector $\bm u\in\mathbb R^{n(r+1)}$, let
$\bm u_i\in\mathbb R^{r+1}$ contain the coordinates of node $i$.
For a matrix $\bm C\in\mathbb R^{n(r+1)\times n(r+1)}$, let
$\bm C_{ij}\in\mathbb R^{(r+1)\times(r+1)}$ be its $(i,j)$-th node block.
Define the block norms by
\[
 \|\bm u\|_{\blinf}=\max_i\|\bm u_i\|_2,
 \quad
 \|\bm C\|_{\blinf}=\max_i\sum_j\|\bm C_{ij}\|_2,
\]
where ``bl'' stands for ``block'', and $\|\bm C_{ij}\|_2$ is the spectral norm.
In particular, $\|\psi(\bm h)\|_{\blinf}=\max_i\sqrt{h_{\alpha,i}^2+\|\bm E_{Z,i\cdot}\|_2^2}$.

For $\bm H,\bm G\in\mathbb R^{n(r+1)\times n(r+1)}$ and
$\bm u\in\mathbb R^{n(r+1)}$, these norms satisfy
\begin{align}
\label{eq: node product bounds}
 \|\bm H\bm G\|_{\blinf}
 \le\|\bm H\|_{\blinf}\|\bm G\|_{\blinf},\quad
 \|\bm H\bm u\|_{\blinf}
 \le\|\bm H\|_{\blinf}\|\bm u\|_{\blinf}.
\end{align}
For endpoints satisfying $\psi(\bm\nu_\ell-\bm\nu^*)\in\mathcal T_n$, $\ell=1,2$, define
\begin{align}
\label{eq: extended tangent Hessian}
 \mathsf F_A(\bm\nu_1,\bm\nu_2)
 =-\bm\Pi_{\mathcal T_n}\int_0^1\nabla_\psi^2 L_A
 \{\bm\nu_2+v(\bm\nu_1-\bm\nu_2)\} dv\bm\Pi_{\mathcal T_n}
 +s_n(\bm I-\bm\Pi_{\mathcal T_n}),
\end{align}
and set $\mathsf F_0(\bm\nu_1,\bm\nu_2)=\mathbb E_A\{\mathsf F_A(\bm\nu_1,\bm\nu_2)\}$
and $\mathsf F_{0,*}=\mathsf F_0(\bm\nu^*,\bm\nu^*)$, where $\mathbb E_A$,
$\operatorname{Var}_A$ and $\operatorname{Cov}_A$ take moments over the network
observations, with endpoints held fixed when evaluating $\mathsf F_0$. The added term is
zero on $\mathcal T_n$ and equals $s_n\bm I$ on $\mathcal T_n^\perp$.
Define $\mathsf F_*:\mathcal T_n\to\mathcal T_n$ by $\mathsf F_*\bm u=\mathsf F_{0,*}\bm u$ for $\bm u\in\mathcal T_n$.
Define the projected network score at the truth by
\begin{align}
\label{eq: projected network score definition}
 \bm g_A=\bm\Pi_{\mathcal T_n}\nabla_\psi L_A(\bm\nu^*)
 =\bm\Pi_{\mathcal T_n}\sum_{i<j}(a_{ij}-p_{ij}^*)
       \nabla_\psi\theta_{ij}(\bm\nu^*).
\end{align}
\paragraph{Response score.}
For the augmented design $\bm X_+$ in the main article, write
$\bm X_+^*=(\bm X,\bm Z^*)$ and
$\widehat{\bm X}_+=(\bm X,\widehat{\bm Z})$.
For a transition matrix $\bm W\in\mathbb R^{n\times n}$ and
$\gamma\in\Gamma_\gamma$, define
\[
 \bm Q_\gamma(\bm W)=\bm W(\bm I_n-\gamma\bm W)^{-1},\quad
 \bm Q=\bm Q_{\gamma^*}(\bm W^*)=\bm W^*\bm S^{*-1}.
\]
For a full-column-rank $\bm Z\in\mathbb R^{n\times r}$, write
$\bm Z^\dagger=(\bm Z^\top\bm Z)^{-1}\bm Z^\top$ and
$\bm r_y(\bm\zeta;\bm P,\bm Z)=\bm S_\gamma(\bm P)\bm y-\bm X\bm\beta-\bm Z\bm\eta$.
The normalized response score defined in Section~\ref{sec: theory} has the explicit form
\begin{align}
\label{eq: full response score}
 \bm q_y(\bm\zeta;\bm P,\bm Z)=\frac1n
 \begin{pmatrix}
 \tau\bm X^\top\bm r_y(\bm\zeta;\bm P,\bm Z)\\
 \tau\bm Z^\top\bm r_y(\bm\zeta;\bm P,\bm Z)\\
 -\tr\{\bm W(\bm P)\bm S_\gamma(\bm P)^{-1}\}+\tau(\bm W(\bm P)\bm y)^\top\bm r_y(\bm\zeta;\bm P,\bm Z)\\
 n/(2\tau)-\|\bm r_y(\bm\zeta;\bm P,\bm Z)\|_2^2/2
 \end{pmatrix}.
\end{align}
The coordinate order is $(\bm\beta,\bm\eta,\gamma,\tau)$.
For $\bm\nu=(\bm\alpha,\bm Z)$, write
$\bm q_y(\bm\zeta;\bm\nu)=\bm q_y(\bm\zeta;\bm P(\bm\nu),\bm Z)$.
The same convention applies to $\bm r_y$, $\operatorname{RSS}_n$ and
$L_{y,\mathrm{prof}}$: the argument $\bm\nu$ replaces $(\bm P(\bm\nu),\bm Z)$.
Subscripts $\beta$, $\eta$, $\gamma$ and $\tau$ select the corresponding
components of $\bm q_y$.

\paragraph{Residual moments and fitted covariance.}
Define the fitted negative Jacobian by $\widehat{\mathsf J}_{y,n} =-\nabla_{\bm\zeta}\bm q_y(\widehat{\bm\zeta};\widehat{\bm\nu})$,
with differentiation in $\bm\zeta$ holding
$(\bm y,\widehat{\bm P},\widehat{\bm Z})$ fixed. Define the fitted residuals by
$\widehat{\bm\epsilon}=\widehat{\bm S}\bm y-\bm X\widehat{\bm\beta}-\widehat{\bm Z}\widehat{\bm\eta}$, and put $\overline{\widehat\epsilon}=n^{-1}\sum_i\widehat\epsilon_i$ and $\widehat v=n^{-1}\sum_i(\widehat\epsilon_i-\overline{\widehat\epsilon})^2$. For $\widehat v>0$, define
\begin{align}
\label{eq: empirical error distribution}
 q_i=(\widehat\epsilon_i-\overline{\widehat\epsilon})
 \left(\frac{\widehat\tau^{-1}}{\widehat v}\right)^{1/2},\quad
 \widehat F&=\frac1n\sum_{i=1}^n\delta_{q_i},
\end{align}
where $\delta_q$ is the point mass at $q$.
If $\widehat v=0$, let $\widehat F$ put equal masses at
$\pm\widehat\tau^{-1/2}$. In both cases, $\widehat F$ has mean zero and
variance $\widehat\tau^{-1}$.
For $\widehat v>0$, its third and fourth moments are computed directly as
\[
 \widehat\mu_j=\left(\frac{\widehat\tau^{-1}}{\widehat v}\right)^{j/2}
 \frac1n\sum_i(\widehat\epsilon_i-\overline{\widehat\epsilon})^j,
 \quad j=3,4.
\]
Conditional on the observed data, let $\bm e=(e_1,\ldots,e_n)^\top$ have
independent coordinates with common distribution $\widehat F$.
Write $\mathbb E_{\bm e}$, $\operatorname{Var}_{\bm e}$ and
$\operatorname{Cov}_{\bm e}$ for moments over $\bm e$, with all fitted
quantities held fixed. By construction,
\[
 \mathbb E_{\bm e}e_i=0,\quad
 \mathbb E_{\bm e}e_i^2=\widehat\sigma^2,\quad
 \mathbb E_{\bm e}e_i^j=\widehat\mu_j,\quad j=3,4,
\]
where $\widehat\sigma^2=\widehat\tau^{-1}$.
Define the auxiliary response
$\bm y^{(\bm e)}=\widehat{\bm\mu}_y+\widehat{\bm S}^{-1}\bm e$ and
$\widehat{\bm q}_y^{(\bm e)}$ as the full score in
\eqref{eq: full response score} evaluated at
$(\widehat{\bm\zeta};\widehat{\bm P},\widehat{\bm Z})$ with response
$\bm y^{(\bm e)}$. With $\widehat{\bm\Xi}_n=n\operatorname{Var}_{\bm e}(\widehat{\bm q}_y^{(\bm e)})$, the covariance estimator in \eqref{eq: joint feasible covariance} can be written as
\begin{align*}
 \widehat{\bm V}_n
 =n\widehat{\mathsf J}_{y,n}^{-1}
 \operatorname{Var}_{\bm e}(\widehat{\bm q}_y^{(\bm e)})
 \widehat{\mathsf J}_{y,n}^{-\top}.
\end{align*}
Set $\widehat{\bm V}_n=\bm I_{p+r+2}$ if
$\operatorname{rank}(\bm X,\widehat{\bm Z})<p+r$ or
$\min_{\gamma\in\Gamma_{\gamma,n}}\operatorname{RSS}_n(\gamma;\widehat{\bm P},\widehat{\bm Z})=0$
(the failure conditions in Algorithm~\ref{alg:step2}), or if
$\widehat{\mathsf J}_{y,n}$ is singular.
Set the standardized error vector in \eqref{eq: plugin CLT} to
$\bm0_{p+r+2}$ when $\widehat{\bm V}_n\not\succ0$. Let $\widehat{\bm V}_n^{\rm raw}$ denote the estimator in \eqref{eq: joint feasible covariance} computed before alignment, using $(\widehat{\bm Z}^{\rm raw},\widehat{\bm\eta}^{\rm raw})$. Using the alignment specified in Section~\ref{sec: theory}, the two covariance estimates satisfy
\[
 \widehat{\bm V}_n
 =\diag(\bm I_p,\widehat{\bm O}^{\top},1,1)
 \widehat{\bm V}_n^{\rm raw}
 \diag(\bm I_p,\widehat{\bm O},1,1).
\]

\section*{Section II. Proofs of the main theorems}

\noindent\textbf{Proof of Theorem~\ref{thm: first stage}.}
By Lemma~\ref{lem: local likelihood root}, with probability tending to one
there exists a centered, aligned local maximizer
$\overline{\bm\nu}=(\overline{\bm\alpha},\overline{\bm Z})$ satisfying
$\bm\Pi_{\mathcal T_n}\nabla_\psi L_A(\overline{\bm\nu})=\bm0$.
Choose such a maximizer measurably and set
$\widehat{\bm\nu}=\overline{\bm\nu}$; on the exceptional event, set
$\widehat{\bm\nu}=(\bm0_n,\bm0_{n\times r})$.
All subsequent calculations are on the event of the lemma, whose complement
does not affect the stated probability orders. In these aligned coordinates,
$\bm\delta=\psi(\widehat{\bm\nu}-\bm\nu^*)\in\mathcal T_n$ and
$\Delta_i=\|\bm\delta_i\|_2$. The error bounds in that lemma give $\max_i\Delta_i=O_p(\delta_{n,k})$, $\frac1n\sum_i\Delta_i^2=O_p(s_n^{-1})$.
These bounds prove \eqref{eq: first stage recovery}.
We now use these bounds to prove \eqref{eq: propagation rates}.

The fixed margins in \eqref{eq: hetero primitive bounds} and $\max_i\Delta_i=o_p(1)$
give $\max_i|\widehat\alpha_i-\vartheta_n/2|\le b_n$ and
$\|\widehat{\bm Z}\|_{2,\infty}\le R_Z$ with probability tending to one. To transfer the
parameter rates to the network matrices, expand the fitted predictors as $\Delta_{\theta,ij} =\widehat\alpha_i-\alpha_i^*+\widehat\alpha_j-\alpha_j^* +\bm\Delta_{z,i}^\top\widehat{\bm z}_j +\bm z_i^{*\top}\bm\Delta_{z,j}$,
where $\Delta_{\theta,ij}=\theta_{ij}(\widehat{\bm\nu})-\theta_{ij}^*$.
The bounded true
and fitted latent rows give $|\Delta_{\theta,ij}|\le C(\Delta_i+\Delta_j)$, and
therefore $\max_{i<j}|\Delta_{\theta,ij}|=O_p(\delta_{n,k})$.
Using $(x+y)^2\le2(x^2+y^2)$ gives
\begin{align}
\label{eq: dyad logit moment transfer}
 \frac1{m_n}\sum_{i<j}\Delta_{\theta,ij}^2
 \le\frac C{m_n}\sum_{i<j}(\Delta_i^2+\Delta_j^2)
 =\frac{C(n-1)}{m_n}\sum_i\Delta_i^2
 =O_p(s_n^{-1}).
\end{align}
Here $m_n=n(n-1)/2$, so each nodewise moment appears $n-1$ times.
The bounds on the node effects and latent rows imply
that the true and fitted predictors $\theta_{ij}^*$ and
$\widehat\theta_{ij}=\theta_{ij}(\widehat{\bm\nu})$ both lie in
$\{u:|u-\logit(\rho_n)|\le2b_n+2+R_Z^2\}$.
By the definition of $w_{+,n}$ in \eqref{eq: hetero scales},
$\expit'(u)\le w_{+,n}=K_ns_n/n$ throughout this interval.
The mean value theorem therefore gives
$|\widehat p_{ij}-p_{ij}^*|\le w_{+,n}|\Delta_{\theta,ij}|$, and hence
\begin{align*}
 \|\Delta\bm P\|_{\max}
 &\le w_{+,n}\max_{i<j}|\Delta_{\theta,ij}|
 =O_p(K_ns_n\delta_{n,k}/n),&
 \|\Delta\bm P\|_F^2
 &\le2w_{+,n}^2\sum_{i<j}\Delta_{\theta,ij}^2
 =O_p(K_n^2s_n),
\end{align*}
where $\Delta\bm P=\widehat{\bm P}-\bm P^*$. Also
$\|\widehat{\bm Z}-\bm Z^*\|_F=O_p(\sqrt{n/s_n})=o_p(\sqrt n)$. The Gram bound in
\eqref{eq: hetero primitive bounds} and Weyl's inequality thus give
$\sigma_r(\widehat{\bm Z})\ge c\sqrt n$ with probability tending to one, where
$\sigma_r$ is the smallest singular value. Put
$\widehat{\bm Q}_Z=\widehat{\bm Z}(\widehat{\bm Z}^\top\widehat{\bm Z})^{-1/2}$,
whose columns are orthonormal. Since the two latent spaces have rank $r$,
\begin{align*}
 \|\bm M_{\widehat{\bm Z}}-\bm M_{\bm Z^*}\|_F^2
 &=2\|\bm M_{\bm Z^*}\widehat{\bm Q}_Z\|_F^2\\
 &=2\|\bm M_{\bm Z^*}(\widehat{\bm Z}-\bm Z^*)
 (\widehat{\bm Z}^\top\widehat{\bm Z})^{-1/2}\|_F^2
 \le\frac{2\|\widehat{\bm Z}-\bm Z^*\|_F^2}
 {\sigma_r(\widehat{\bm Z})^2}
 =O_p(s_n^{-1}),
\end{align*}
where the second equality uses $\bm M_{\bm Z^*}\bm Z^*=\bm0$.

For the normalized weights, write $\widehat d_i=d_i(\widehat{\bm P})$
and $d_i^*=d_i(\bm P^*)$.
Let $\bm e_i$ be the $i$-th standard basis vector of $\mathbb R^n$.
Since the constraints defining $\mathcal T_n$ concern only the latent directions,
$\psi(\bm e_i,\bm0)\in\mathcal T_n$.
The projected score at the estimator is zero by
Lemma~\ref{lem: local likelihood root}, so
\[
 0=\psi(\bm e_i,\bm0)^\top\nabla_\psi L_A(\widehat{\bm\nu})
 =\frac{\partial L_A(\widehat{\bm\nu})}{\partial\alpha_i}
 =\sum_{j\ne i}(a_{ij}-\widehat p_{ij}).
\]
Thus $\widehat d_i=\sum_{j\ne i}\widehat p_{ij}
=\sum_{j\ne i}a_{ij}$ exactly.
Lemma~\ref{lem: network consequences} gives $d_i^*\ge(n-1)p_{-,n}\ge cs_n$ and
$p_{+,n}/p_{-,n}\le4K_n$. Bernstein's inequality and a union bound then imply
$\max_i|\widehat d_i-d_i^*|/d_i^* =O_p\{\sqrt{\log n/s_n}+\log n/s_n\}=o_p(1)$. Thus
$\widehat d_i\ge d_i^*/2$ simultaneously with probability tending to one. Write
\begin{align}
\label{eq: main transition identity}
 \widehat{\bm W}-\bm W^*
 =\widehat{\bm D}^{-1}\Delta\bm P
 -\widehat{\bm D}^{-1}\diag(\Delta\bm P\bm1_n)\bm W^*,
\end{align}
where $\widehat{\bm D}=\diag(\widehat d_1,\ldots,\widehat d_n)$. The first term has
Frobenius norm at most $Cs_n^{-1}\|\Delta\bm P\|_F=O_p(K_n/\sqrt {s_n})$. For the second term, using $\widehat d_i\ge d_i^*/2$ and
$\sum_j(p_{ij}^*)^2\le p_{+,n}d_i^*$ gives
\begin{align*}
 \|\widehat{\bm D}^{-1}\diag(\Delta\bm P\bm1_n)\bm W^*\|_F^2
 &=\sum_i\frac{(\widehat d_i-d_i^*)^2}{\widehat d_i^2}
       \frac{\sum_j(p_{ij}^*)^2}{(d_i^*)^2}
 \le4p_{+,n}\sum_i\frac{(\widehat d_i-d_i^*)^2}{(d_i^*)^3}.
\end{align*}
On the event under consideration,
$\widehat d_i-d_i^*=\sum_{j\ne i}(a_{ij}-p_{ij}^*)$.
The corresponding weighted sum of observed degree errors has expectation
\begin{align*}
 p_{+,n}\sum_i\frac{
 \operatorname{Var}_A(\sum_{j\ne i}a_{ij})}{(d_i^*)^3}
 \le p_{+,n}\sum_i\frac1{(d_i^*)^2}
 \le C\frac{p_{+,n}/p_{-,n}}{np_{-,n}}
 \le CK_n/s_n.
\end{align*}
Markov's inequality therefore bounds the second term in
\eqref{eq: main transition identity} by
$O_p(\sqrt{K_n/s_n})$. Combining the two terms gives
$\|\widehat{\bm W}-\bm W^*\|_F=O_p(K_n/\sqrt {s_n})$, completing
\eqref{eq: propagation rates}.
\hfill$\blacksquare$

\noindent\textbf{Proof of Theorem~\ref{thm: plugin inference}.}
Lemma~\ref{lem: response score representation} gives
$\mathbb E_\epsilon\bm q_y(\bm\zeta^*;\bm\nu^*)=\bm0$.
By \eqref{eq: response information definitions} and
Lemma~\ref{lem: response conditioning},
\[
 c\bm I\preceq\mathsf J_{y,n}\preceq C\bm I,\
 \mathbb E_\epsilon\|\bm q_y(\bm\zeta^*;\bm\nu^*)\|_2^2
 =\frac{\tr(\bm\Xi_n)}n\le\frac Cn.
\]
In particular, $\|\mathsf J_{y,n}^{-1}\|_2\le C$, and Markov's inequality gives
$\|\bm q_y(\bm\zeta^*;\bm\nu^*)\|_2=O_p(n^{-1/2})$.

For the network perturbation, put
$\widehat{\bm h}=\widehat{\bm\nu}-\bm\nu^*$, so that
$\psi(\widehat{\bm h})\in\mathcal T_n$ by alignment.
Lemma~\ref{lem: supplementary W expansion} gives
\[
 \|\bm q_y(\bm\zeta^*;\widehat{\bm\nu})
 -\bm q_y(\bm\zeta^*;\bm\nu^*)
 -D_\nu\bm q_y(\bm\zeta^*;\bm\nu^*)[\widehat{\bm h}]\|_2
 =O_p(K_n^2/s_n).
\]
Lemma~\ref{lem: supplementary Riesz energy} also gives
$D_\nu q_{y,a}(\bm\zeta^*;\bm\nu^*)[\widehat{\bm h}]
=O_p\{K_n/\sqrt{ns_n}+K_n^2/s_n\}$ for each $a\in[p+r+2]$.
Combining these bounds and using the fixed dimension $p+r+2$, we obtain
\begin{align}
\label{eq: main response score perturbation}
 \|\bm q_y(\bm\zeta^*;\widehat{\bm\nu})-\bm q_y(\bm\zeta^*;\bm\nu^*)\|_2
 =O_p\{K_n/\sqrt{ns_n}+K_n^2/s_n\}.
\end{align}
Lemma~\ref{lem: network consequences} gives
$s_n/(K_n^8\log^2n)\to\infty$.
Since $K_n\ge1$, this implies $K_n/\sqrt{s_n}\to0$ and hence
$K_n/\sqrt{ns_n}=o(n^{-1/2})$.

Lemma~\ref{lem: response profile uniform} gives
\begin{equation}
\label{eq: main profile consistency and score}
 \|\widehat{\bm\zeta}-\bm\zeta^*\|_2=o_p(1),\
 \|\bm q_y(\widehat{\bm\zeta};\widehat{\bm\nu})\|_2=O_p(n^{-1}).
\end{equation}
Put $\bm\zeta_t=\bm\zeta^*+t(\widehat{\bm\zeta}-\bm\zeta^*)$ and
$\overline{\mathsf J}_n=-\int_0^1\nabla_{\bm\zeta}
\bm q_y(\bm\zeta_t;\widehat{\bm\nu})\,dt$.
For a sufficiently small fixed $\delta>0$, let
$\mathcal B_\delta=\{\bm\zeta:\|\bm\zeta-\bm\zeta^*\|_2\le\delta\}$.
By \eqref{eq: main profile consistency and score}, the segment lies in
$\mathcal B_\delta$ with probability tending to one.
Lemma~\ref{lem: fitted response Jacobian} gives
\begin{align*}
 \|-\nabla_{\bm\zeta}\bm q_y(\bm\zeta^*;\widehat{\bm\nu})
 -\mathsf J_{y,n}\|_2&=o_p(1),\\
 \sup_{\substack{\bm\zeta_1,\bm\zeta_2\in\mathcal B_\delta\\
 \bm\zeta_1\ne\bm\zeta_2}}
 \frac{\|\nabla_{\bm\zeta}\bm q_y(\bm\zeta_1;\widehat{\bm\nu})
 -\nabla_{\bm\zeta}\bm q_y(\bm\zeta_2;\widehat{\bm\nu})\|_2}
 {\|\bm\zeta_1-\bm\zeta_2\|_2}&=O_p(1).
\end{align*}
Consequently,
\begin{align*}
 \|\overline{\mathsf J}_n-\mathsf J_{y,n}\|_2
 &\le\sup_{t\in[0,1]}
 \|-\nabla_{\bm\zeta}\bm q_y(\bm\zeta_t;\widehat{\bm\nu})-\mathsf J_{y,n}\|_2\\
 &\le\|-\nabla_{\bm\zeta}\bm q_y(\bm\zeta^*;\widehat{\bm\nu})-\mathsf J_{y,n}\|_2
 +O_p(1)\|\widehat{\bm\zeta}-\bm\zeta^*\|_2=o_p(1).
\end{align*}
Since $\lambda_{\min}(\mathsf J_{y,n})\ge c$,
$\overline{\mathsf J}_n$ is invertible with probability tending to one,
and $\|\overline{\mathsf J}_n^{-1}\|_2=O_p(1)$. The inverse identity gives
\begin{align*}
 \|\overline{\mathsf J}_n^{-1}-\mathsf J_{y,n}^{-1}\|_2
 =\|\overline{\mathsf J}_n^{-1}
 (\mathsf J_{y,n}-\overline{\mathsf J}_n)\mathsf J_{y,n}^{-1}\|_2\le\|\overline{\mathsf J}_n^{-1}\|_2
 \|\mathsf J_{y,n}-\overline{\mathsf J}_n\|_2
 \|\mathsf J_{y,n}^{-1}\|_2=o_p(1).
\end{align*}
Integrating the derivative of $\bm q_y(\bm\zeta_t;\widehat{\bm\nu})$
along $t\in[0,1]$ gives $\overline{\mathsf J}_n(\widehat{\bm\zeta}-\bm\zeta^*)
 =\bm q_y(\bm\zeta^*;\widehat{\bm\nu})-\bm q_y(\widehat{\bm\zeta};\widehat{\bm\nu})$. Consequently, the preceding bounds and \eqref{eq: main response score perturbation} imply
\begin{align*}
 \|\widehat{\bm\zeta}-\bm\zeta^*\|_2
 &\le\|\overline{\mathsf J}_n^{-1}\|_2
 \big\{\|\bm q_y(\bm\zeta^*;\bm\nu^*)\|_2
 +\|\bm q_y(\bm\zeta^*;\widehat{\bm\nu})-\bm q_y(\bm\zeta^*;\bm\nu^*)\|_2+\|\bm q_y(\widehat{\bm\zeta};\widehat{\bm\nu})\|_2\big\}\\
& =O_p(n^{-1/2}+K_n^2/s_n),
\end{align*}
which proves \eqref{eq: response rate all degrees}.

Now impose Assumption~\ref{assump: joint inference}.
Condition~\eqref{eq: inference growth} and
\eqref{eq: main response score perturbation} give
\[
 \sqrt n\|\bm q_y(\bm\zeta^*;\widehat{\bm\nu})-\bm q_y(\bm\zeta^*;\bm\nu^*)\|_2
 =O_p\{K_n/\sqrt{s_n}+\sqrt nK_n^2/s_n\}=o_p(1).
\]
Multiplying the integrated score identity by $\sqrt n$, using
\eqref{eq: main profile consistency and score} and the inverse bounds above, yields
\begin{align}
\label{eq: main response expansions}
 \sqrt n(\widehat{\bm\zeta}-\bm\zeta^*)
 =\mathsf J_{y,n}^{-1}\sqrt n\bm q_y(\bm\zeta^*;\bm\nu^*)+o_p(1).
\end{align}
Under \eqref{eq: response nondegeneracy},
Lemma~\ref{lem: response conditioning} and
Lemma~\ref{lem: joint covariance} give, respectively,
\begin{align}
\label{eq: main covariance consistency}
 c\bm I_{p+r+2}\preceq\bm V_n\preceq C\bm I_{p+r+2},\quad
 \|\widehat{\bm V}_n-\bm V_n\|_2=o_p(1).
\end{align}

Write $\bm c_{n,i}^\top$ for row $i$ of $\bm C_n$ and
$\bm G_n=n^{-1}\bm C_n^\top\bm C_n$.
Lemma~\ref{lem: response score representation} gives
$\max_i\|\bm c_{n,i}\|_2\le C$, and
Lemma~\ref{lem: response conditioning} gives $\|\bm G_n^{-1}\|_2\le C$.
Define
\[
 \bm U_n=\frac1{\sqrt n}\sum_i
 \begin{pmatrix}\bm G_n^{-1}\bm c_{n,i}\epsilon_i\\
 -\tau^{*2}(\epsilon_i^2-\sigma^{*2})\end{pmatrix}.
\]
Lemma~\ref{lem: response conditioning} also gives
\begin{equation}
\label{eq: main response influence representation}
 \mathsf J_{y,n}^{-1}\sqrt n\bm q_y(\bm\zeta^*;\bm\nu^*)
 =\bm U_n+\bm r_n,\
 \mathbb E_\epsilon\|\bm r_n\|_2^2=o(1).
\end{equation}
For a fixed unit vector $\bm d\in\mathbb R^{p+r+2}$, define
\[
 \xi_{n,i}=\bm d^\top\bm V_n^{-1/2}
 \begin{pmatrix}\bm G_n^{-1}\bm c_{n,i}\epsilon_i\\
 -\tau^{*2}(\epsilon_i^2-\sigma^{*2})\end{pmatrix}.
\]
By Assumption~\ref{assump: response errors}, these variables are independent
and satisfy $\mathbb E_\epsilon\xi_{n,i}=0$.
Assumption~\ref{assump: response errors}, $x^4\le4!e^x$ for $x\ge0$, and the compactness condition
on $\tau^*$ in Assumption~\ref{assump: response} give $\sup_i\mathbb E_\epsilon|\epsilon_i|^8 \le48C_\epsilon^8\sigma^{*8}\le C$.
Together with \eqref{eq: main covariance consistency} and the upper bound
on $\tau^*$ in Assumption~\ref{assump: response}, the bounds on
$\bm c_{n,i}$ and $\bm G_n^{-1}$ give, since
$\|\bm d\|_2=1$,
\begin{align*}
 |\xi_{n,i}|
 \le\|\bm V_n^{-1/2}\|_2
 \big\{\|\bm G_n^{-1}\|_2\|\bm c_{n,i}\|_2|\epsilon_i|
 +\tau^{*2}|\epsilon_i^2-\sigma^{*2}|\big\}\le C\big(|\epsilon_i|+|\epsilon_i^2-\sigma^{*2}|\big).
\end{align*}
Using $(a+b)^4\le8(a^4+b^4)$ for $a,b\ge0$, we obtain
\begin{align*}
 \sup_i\mathbb E_\epsilon|\xi_{n,i}|^4
 \le C\sup_i\big\{\mathbb E_\epsilon|\epsilon_i|^4
 +\mathbb E_\epsilon|\epsilon_i^2-\sigma^{*2}|^4\big\}\le C\sup_i\big\{\mathbb E_\epsilon|\epsilon_i|^4
 +\mathbb E_\epsilon|\epsilon_i|^8+\sigma^{*8}\big\}\le C,
\end{align*}
where $\mathbb E_\epsilon|\epsilon_i|^4
\le(\mathbb E_\epsilon|\epsilon_i|^8)^{1/2}$ and the preceding eighth-moment
bound give the last inequality.

By independence of the errors,
\[
 \operatorname{Var}_\epsilon(\bm U_n)
 =\begin{pmatrix}
 \sigma^{*2}\bm G_n^{-1}
 &-\tau^{*2}\mu_3\bm G_n^{-1}\bar{\bm c}_n\\
 -\tau^{*2}\mu_3\bar{\bm c}_n^\top\bm G_n^{-1}
 &\tau^{*4}(\mu_4-\sigma^{*4})
 \end{pmatrix},
\]
where $\bar{\bm c}_n=n^{-1}\sum_i\bm c_{n,i}$.
The covariance formula in Lemma~\ref{lem: response conditioning} therefore gives
$\|\operatorname{Var}_\epsilon(\bm U_n)-\bm V_n\|_2=o(1)$.
Since $n^{-1/2}\sum_i\xi_{n,i}=\bm d^\top\bm V_n^{-1/2}\bm U_n$,
\begin{align*}
 \left|\frac1n\sum_i\mathbb E_\epsilon\xi_{n,i}^2-1\right|
 &=\left|\bm d^\top\bm V_n^{-1/2}
 \{\operatorname{Var}_\epsilon(\bm U_n)-\bm V_n\}
 \bm V_n^{-1/2}\bm d\right|\\
 &\le\|\bm V_n^{-1/2}\|_2^2
 \|\operatorname{Var}_\epsilon(\bm U_n)-\bm V_n\|_2=o(1).
\end{align*}
It follows that
\[
 \frac{\sum_i\mathbb E_\epsilon|\xi_{n,i}|^4}
 {(\sum_i\mathbb E_\epsilon\xi_{n,i}^2)^2}
 \le\frac{Cn}{(n/2)^2}\longrightarrow0.
\]
Lyapunov's theorem and the Cram\'er--Wold theorem give
$\bm V_n^{-1/2}\bm U_n\to_d N(\bm0,\bm I_{p+r+2})$.
Moreover, \eqref{eq: main response influence representation} implies
\[
 \mathbb E_\epsilon\|\bm V_n^{-1/2}\bm r_n\|_2^2
 \le\|\bm V_n^{-1/2}\|_2^2\mathbb E_\epsilon\|\bm r_n\|_2^2=o(1).
\]
Thus Markov's inequality gives $\|\bm V_n^{-1/2}\bm r_n\|_2=o_p(1)$.
Combining \eqref{eq: main response influence representation} with
\eqref{eq: main response expansions} gives
\begin{align*}
 \bm V_n^{-1/2}\sqrt n(\widehat{\bm\zeta}-\bm\zeta^*)
 =\bm V_n^{-1/2}\mathsf J_{y,n}^{-1}
 \sqrt n\bm q_y(\bm\zeta^*;\bm\nu^*)+o_p(1)
 =\bm V_n^{-1/2}\bm U_n+o_p(1)
 \to_d N(\bm0,\bm I_{p+r+2}).
\end{align*}

Finally, \eqref{eq: main covariance consistency} implies
$\Pr\{c\bm I/2\preceq\widehat{\bm V}_n\preceq2C\bm I\}\to1$.
For positive definite matrices $\bm A,\bm B$, the integral representation
of the inverse square root gives
\[
 \bm A^{-1/2}-\bm B^{-1/2}
 =\frac1\pi\int_0^\infty t^{-1/2}
 (\bm A+t\bm I)^{-1}(\bm B-\bm A)(\bm B+t\bm I)^{-1}\,dt.
\]
On the preceding event, both $\widehat{\bm V}_n$ and $\bm V_n$ have minimum
eigenvalue at least $c/2$. Taking $\bm A=\widehat{\bm V}_n$ and
$\bm B=\bm V_n$ therefore yields
\begin{align*}
 \|\widehat{\bm V}_n^{-1/2}-\bm V_n^{-1/2}\|_2
 \le\frac{\|\widehat{\bm V}_n-\bm V_n\|_2}{\pi}
 \int_0^\infty\frac{t^{-1/2}}{(t+c/2)^2}\,dt=\frac{\|\widehat{\bm V}_n-\bm V_n\|_2}{2(c/2)^{3/2}}=o_p(1),
\end{align*}
where the last step uses \eqref{eq: main covariance consistency}.
Since \eqref{eq: main response expansions} gives
$\sqrt n\|\widehat{\bm\zeta}-\bm\zeta^*\|_2=O_p(1)$,
\begin{align*}
\|\widehat{\bm V}_n^{-1/2}\sqrt n(\widehat{\bm\zeta}-\bm\zeta^*)
 -\bm V_n^{-1/2}\sqrt n(\widehat{\bm\zeta}-\bm\zeta^*)\|_2\le\|\widehat{\bm V}_n^{-1/2}-\bm V_n^{-1/2}\|_2
 \sqrt n\|\widehat{\bm\zeta}-\bm\zeta^*\|_2=o_p(1).
\end{align*}
Slutsky's theorem then proves \eqref{eq: plugin CLT}.
\hfill$\blacksquare$

\section*{Section III. Auxiliary lemmas and proofs}

\begin{lemma}\label{lem: network consequences}
Suppose Assumptions~\ref{assump: network} and~\ref{assump: network growth} hold. Let $\mathcal V_n=\{(\bm\alpha,\bm Z): \max_i|\alpha_i-\vartheta_n/2|\le b_n+1, \ \bm Z^\top\bm1_n=\bm0,\ \|\bm Z\|_{2,\infty}\le R_Z\}$. Then for each $\bm\nu\in\mathcal V_n$ and $i\ne j$, $p_{ij}(\bm\nu)$ belongs to $[p_{-,n},p_{+,n}]$ and
$w_{ij}(\bm\nu)=p_{ij}(\bm\nu)\{1-p_{ij}(\bm\nu)\}$ belongs to $[w_{-,n},w_{+,n}]$. Moreover, we have
\begin{align}
\label{eq: envelope implications}
 p_{+,n}/p_{-,n}\le4K_n,\ 
 \min_i &\ d_i\{\bm P(\bm\nu)\} \ge (n-1)p_{-,n}\ge cs_n,\ 
 \max_{i\ne j}\ W_{ij}\{\bm P(\bm\nu)\}\le CK_n/n,\\
\label{eq: basic growth consequences}
 &s_n/(K_n^8\log^2n)\to\infty,\ K_n^8/n\to0,\ 
 K_n\delta_{n,k}\to0.
\end{align}
With probability tending to one, it also holds true simultaneously that 
\begin{align}
\label{eq: adjacency spectral bound}
 \|\bm A-\bm P^*\|_2\le C\sqrt {t_n}, \
 \max_i\sum_{j\ne i}|a_{ij}-p_{ij}^*|\le Ct_n.
\end{align}
\end{lemma}

\noindent\textbf{Proof of Lemma~\ref{lem: network consequences}.}
For each $\bm\nu\in\mathcal V_n$, the triangle and Cauchy--Schwarz inequalities, together with the defining bounds of $\mathcal V_n$, give
\begin{align*}
 |\theta_{ij}(\bm\nu)-\vartheta_n|
 =\left|(\alpha_i-\vartheta_n/2)+(\alpha_j-\vartheta_n/2)
       +\bm z_i^\top\bm z_j\right|\le2(b_n+1)+R_Z^2.
\end{align*}
The definitions of $p_{\pm,n}$ and $w_{\pm,n}$ therefore give the stated envelopes.
To bound the probability ratio, first suppose $p_{-,n}>1/2$. Since
$p_{+,n}<1$ and $K_n\ge1$, we have $\frac{p_{+,n}}{p_{-,n}}<2\le4K_n$.
Now suppose $p_{-,n}\le1/2$. If $p_{+,n}\le1/2$, then
$w_{+,n}\ge p_{+,n}(1-p_{+,n})\ge p_{+,n}/2$.
If $p_{+,n}>1/2$, the interval $[p_{-,n},p_{+,n}]$ contains $1/2$,
so $w_{+,n}=1/4$ and $p_{+,n}\le1=4w_{+,n}$.
Thus $p_{+,n}\le4w_{+,n}$ in both cases. Also,
$w_{-,n}\le p_{-,n}(1-p_{-,n})\le p_{-,n}$, and hence $\frac{p_{+,n}}{p_{-,n}}\le\frac{4w_{+,n}}{w_{-,n}}=4K_n$.
Summing the lower probability bound gives
$d_i\{\bm P(\bm\nu)\}\ge(n-1)p_{-,n}\ge(n-1)w_{-,n}\ge s_n/2$ for $n\ge2$.
Consequently, for $i\ne j$,
\[
 W_{ij}\{\bm P(\bm\nu)\}
 \le\frac{p_{+,n}}{(n-1)p_{-,n}}
 \le\frac{8K_n}{n}.
\]

We next derive \eqref{eq: basic growth consequences} from \eqref{eq: hetero growth} in  Assumption~\ref{assump: network growth}. Since $K_n\ge1$ and $k\ge1$, definition~\eqref{eq: hetero rates definitions} gives
\[
 \delta_{n,k}\ge K_n^{2k-1}\sqrt{K_n\log n/s_n}
 \ge K_n^{3/2}\sqrt{\log n/s_n}.
\]
Raising this inequality to the fourth power and multiplying by $K_n^2s_n$ yields
\[
 \frac{K_n^8\log^2n}{s_n}
 \le K_n^2\delta_{n,k}^4s_n\longrightarrow0,
\]
which proves $s_n/(K_n^8\log^2n)\to\infty$.
Moreover, $w_{-,n}\le1/4$, so $s_n=nw_{-,n}\le n/4$. Consequently,
\[
 \frac{K_n^8}{n}\le\frac{K_n^8}{4s_n}
 =\frac{1}{4\log^2n}\frac{K_n^8\log^2n}{s_n}\longrightarrow0.
\]
The same bounds imply $K_n/\sqrt{s_n}\le\sqrt{K_n^8/s_n}\to0$. Taking the square root of
\eqref{eq: hetero growth} gives $K_n\delta_{n,k}^2\sqrt{s_n}\to0$, and hence $(K_n\delta_{n,k})^2 =(K_n\delta_{n,k}^2\sqrt{s_n})\frac{K_n}{\sqrt{s_n}} \longrightarrow0$.
To prove \eqref{eq: adjacency spectral bound}, let $\bm B=\bm A-\bm P^*$.
Then $B_{ij}$, $i<j$, are independent,
$B_{ji}=B_{ij}$, $B_{ii}=0$, and $\mathbb E_A B_{ij}=0,\ |B_{ij}|\le1,\ \mathbb E_A B_{ij}^2=p_{ij}^*(1-p_{ij}^*)=w_{ij}^*$.
The true parameters belong to $\mathcal V_n$, so the curvature envelope gives $\max_i\sum_{j\ne i}\mathbb E_A B_{ij}^2 =\max_i\sum_{j\ne i}w_{ij}^* \le(n-1)w_{+,n}\le t_n$.
For some positive constants $C_0$ and $c_0$, it follows from Corollary~3.12 and Remark~3.13 of \citet{nldaBandeira2016Sharp} that
\[
 \mathbb P_A\{\|\bm B\|_2\ge C_0\sqrt{t_n}+u\}
 \le n\exp(-u^2/c_0)
\]
for all $u\ge0$. Taking $u=\sqrt{4c_0\log n}$ and noting that
\[
 \frac{t_n}{\log n}=\frac{K_ns_n}{\log n}\ge\frac{s_n}{\log n}
 =\frac{s_n}{K_n^8\log^2n}\,K_n^8\log n\longrightarrow\infty,
\]
we have $\|\bm A-\bm P^*\|_2\le C\sqrt{t_n}$ with probability at least $1-n^{-3}$
for all sufficiently large $n$.

For the absolute row sums, the Bernoulli distribution gives
\begin{align*}
 \mathbb E_A|B_{ij}|
 &=p_{ij}^*(1-p_{ij}^*)+(1-p_{ij}^*)p_{ij}^*=2w_{ij}^*,\\
 \operatorname{Var}_A(|B_{ij}|)
 &=\mathbb E_A B_{ij}^2-(\mathbb E_A|B_{ij}|)^2
 =w_{ij}^*-4(w_{ij}^*)^2\le w_{ij}^*.
\end{align*}
For each fixed $i$, the variables $Y_{ij}=|B_{ij}|-2w_{ij}^*$, $j\ne i$, are independent and centered. Since $|B_{ij}|\in[0,1]$, we have $|Y_{ij}|=\big| |B_{ij}|-\mathbb{E}_A|B_{ij}|\big|\le1$, and the preceding calculation gives $\sum_{j\ne i}\operatorname{Var}_A(Y_{ij})\le t_n$. Bernstein's inequality therefore yields, for $x>0$,
\[
 \mathbb P_A\left\{\sum_{j\ne i}Y_{ij}\ge x\right\}
 \le\exp\left\{-\frac{x^2}{2(t_n+x/3)}\right\}.
\]
Also $\sum_{j\ne i}\mathbb E_A|B_{ij}|=2\sum_{j\ne i}w_{ij}^*\le2t_n$. Since $t_n/\log n\to\infty$, choosing $x=t_n$ and taking a union bound over the $n$ rows gives
\[
 \mathbb P_A\left\{\max_i\sum_{j\ne i}|B_{ij}|\ge3t_n\right\}
 \le n\exp(-3t_n/8)\longrightarrow0.
\]
Combining the two probability bounds proves
\eqref{eq: adjacency spectral bound} simultaneously.
\hfill$\blacksquare$

\begin{lemma}\label{lem: primitive network Fisher}
Suppose Assumption~\ref{assump: network} holds. Then, for sufficiently large $n$
and every $\psi(\bm h)\in\mathcal T_n$,
\begin{align}
\label{eq: offdiag derivative geometry}
 cn\|\psi(\bm h)\|_2^2
 \le\sum_{i<j}\dot\theta_{ij}^*[\bm h]^2
 \le Cn\|\psi(\bm h)\|_2^2.
\end{align}
The Fisher information on $\mathcal T_n$ satisfies, for
$\psi(\bm h),\psi(\bm g)\in\mathcal T_n$,
\begin{equation}
\label{eq: network Fisher information}
\begin{aligned}
 \langle\psi(\bm h),\mathsf F_*\psi(\bm g)\rangle
 =\sum_{i<j}w_{ij}^*\dot\theta_{ij}^*[\bm h]\dot\theta_{ij}^*[\bm g],\ 
 cs_n\bm I_{\mathcal T_n}
 \preceq\mathsf F_*\preceq Ct_n\bm I_{\mathcal T_n}.
\end{aligned}
\end{equation}
The orthogonal projector onto $\mathcal T_n$ satisfies
\begin{align}
\label{eq: network projector bounds}
 \max_{i,j}\|[\bm I-\bm\Pi_{\mathcal T_n}]_{ij}\|_2
 \le C/n,\
 \|\bm\Pi_{\mathcal T_n}\|_{\blinf}\le C.
\end{align}
\end{lemma}

\noindent\textbf{Proof of Lemma~\ref{lem: primitive network Fisher}.} For $\psi(\bm h)\in\mathcal T_n$, let $\bm L=\bm h_\alpha\bm1_n^\top+\bm1_n\bm h_\alpha^\top$ and $\bm B=\bm E_Z\bm Z^{*\top}+\bm Z^*\bm E_Z^\top$. Since $\bm Z^{*\top}\bm1_n=\bm E_Z^\top\bm1_n=\bm0$, we have $\bm B\bm1_n=\bm0$ and hence $\langle\bm L,\bm B\rangle_F=\bm h_\alpha^\top\bm B\bm1_n+\bm1_n^\top\bm B\bm h_\alpha=0$.
Expanding the squared Frobenius norms gives
\begin{align*}
 \|\bm L\|_F^2
 =2n\|\bm h_\alpha\|_2^2+2(\bm1_n^\top\bm h_\alpha)^2,\ \|\bm B\|_F^2
 =2\tr(\bm E_Z^\top\bm E_Z\bm Z^{*\top}\bm Z^*) +2\tr\{(\bm E_Z^\top\bm Z^*)^2\}.
\end{align*}
The alignment condition makes $\bm E_Z^\top\bm Z^*$ symmetric, so
$\tr\{(\bm E_Z^\top\bm Z^*)^2\}=\|\bm E_Z^\top\bm Z^*\|_F^2\ge0$.
Also $\bm Z^{*\top}\bm Z^*\succeq c_Zn\bm I_r$ and
$\|\bm Z^*\|_2\le\|\bm Z^*\|_F\le R_Z\sqrt n$. Therefore,
\begin{align*}
 2n\|\bm h_\alpha\|_2^2\le\|\bm L\|_F^2\le4n\|\bm h_\alpha\|_2^2,\ 2c_Zn\|\bm E_Z\|_F^2 \le\|\bm B\|_F^2 \le4\|\bm E_Z\|_F^2\|\bm Z^*\|_2^2
 \le4R_Z^2n\|\bm E_Z\|_F^2.
\end{align*}
Here the upper bounds use $(\bm1_n^\top\bm h_\alpha)^2\le n\|\bm h_\alpha\|_2^2$
and $\|\bm E_Z\bm Z^{*\top}+\bm Z^*\bm E_Z^\top\|_F\le2\|\bm E_Z\|_F\|\bm Z^*\|_2$.
The diagonal contribution satisfies
\[
 \sum_i(L_{ii}+B_{ii})^2
 =4\sum_i(h_{\alpha,i}+\bm E_{Z,i\cdot}\bm Z_{i\cdot}^{*\top})^2
 \le8\|\bm h_\alpha\|_2^2+8R_Z^2\|\bm E_Z\|_F^2.
\]
Since $\bm L+\bm B$ is symmetric and $\langle\bm L,\bm B\rangle_F=0$,
\begin{align*}
 \sum_{i<j}(L_{ij}+B_{ij})^2
 &=\tfrac12\left\{\|\bm L\|_F^2+\|\bm B\|_F^2
                  -\sum_i(L_{ii}+B_{ii})^2\right\}\\
 &\ge(n-4)\|\bm h_\alpha\|_2^2
      +(c_Zn-4R_Z^2)\|\bm E_Z\|_F^2.
\end{align*}
For $n\ge\max\{8,8R_Z^2/c_Z\}$, the lower bound is at least
$\tfrac12\min(1,c_Z)n\|\psi(\bm h)\|_2^2$.
Dropping the diagonal subtraction instead gives the upper bound
$2\max(1,R_Z^2)n\|\psi(\bm h)\|_2^2$.
Since $L_{ij}+B_{ij}=\dot\theta_{ij}^*[\bm h]$, this proves \eqref{eq: offdiag derivative geometry}.

For the Fisher information, differentiating each likelihood term twice gives
\begin{align*}
 -\nabla_\psi^2 L_A(\bm\nu)
 =\sum_{i<j}\left[w_{ij}(\bm\nu)\bm\phi_{ij}(\bm\nu)\bm\phi_{ij}(\bm\nu)^\top
       -\{a_{ij}-p_{ij}(\bm\nu)\}\bm B_{A,ij}\right].
\end{align*}
At $\bm\nu^*$, $\mathbb E_A(a_{ij}-p_{ij}^*)=0$.
The added term in \eqref{eq: extended tangent Hessian} vanishes on
$\mathcal T_n$, and
$\bm\phi_{ij}^{*\top}\psi(\bm h)=\dot\theta_{ij}^*[\bm h]$.
Taking expectation and restricting to $\mathcal T_n$ therefore gives the first result of \eqref{eq: network Fisher information}.
Assumption~\ref{assump: network} gives
$|\theta_{ij}^*-\vartheta_n|\le2b_n+R_Z^2$, so the definitions in
\eqref{eq: hetero scales} imply $w_{-,n}\le w_{ij}^*\le w_{+,n}$.
Setting $\bm g=\bm h$ and applying \eqref{eq: offdiag derivative geometry} gives $cs_n\|\psi(\bm h)\|_2^2 \le\langle\psi(\bm h),\mathsf F_*\psi(\bm h)\rangle \le Ct_n\|\psi(\bm h)\|_2^2$.
This proves the Fisher bounds.

To prove \eqref{eq: network projector bounds}, define in the latent coordinates $\mathcal C=\{\bm1_n\bm a^\top:\bm a\in\mathbb R^r\}$, $\mathcal R=\{\bm Z^*\bm\Omega:\bm\Omega^\top=-\bm\Omega\}$.
A matrix $\bm E_Z$ is orthogonal to $\mathcal C$ if and only if
$\bm E_Z^\top\bm1_n=\bm0$. Since
$\langle\bm E_Z,\bm Z^*\bm\Omega\rangle_F=\tr(\bm E_Z^\top\bm Z^*\bm\Omega)$,
it is orthogonal to every element of $\mathcal R$ if and only if
$\bm E_Z^\top\bm Z^*$ is symmetric. These are exactly the constraints defining
$\mathcal T_n$ in \eqref{eq: network tangent coordinates}.
Moreover, $\bm1_n^\top\bm Z^*=\bm0$ implies $\mathcal C\perp\mathcal R$.
The node-effect coordinates are unrestricted, so
$\bm I-\bm\Pi_{\mathcal T_n}$ is zero on those coordinates and is the sum of
the projectors onto $\mathcal C$ and $\mathcal R$ on the latent coordinates.
Write $\bm E_Z=\bm E_0+\bm E_{\mathcal C}+\bm E_{\mathcal R}$ for its
orthogonal decomposition, where $\bm E_0\in\mathcal C^\perp\cap\mathcal R^\perp$,
$\bm E_{\mathcal C}\in\mathcal C$ and $\bm E_{\mathcal R}\in\mathcal R$. Then
\begin{align*}
 (\bm I-\bm\Pi_{\mathcal T_n})\psi(\bm h_\alpha,\bm E_Z)
 =\psi(\bm h_\alpha,\bm E_Z)-\psi(\bm h_\alpha,\bm E_0)=\psi(\bm0,\bm E_{\mathcal C}+\bm E_{\mathcal R}).
\end{align*}

We next bound $[\bm I-\bm\Pi_{\mathcal T_n}]_{ij}$ directly.
Represent a latent matrix $\bm E$ by
$(\bm E_{1\cdot},\ldots,\bm E_{n\cdot})^\top\in\mathbb R^{nr}$,
so its projection matrices have $r\times r$ node blocks.
Projection onto $\mathcal C$ replaces every row by the mean of all rows,
$\bm E\mapsto\bm1_n\bm1_n^\top\bm E/n$; hence each block is $\bm I_r/n$.

For $r\ge2$, choose a basis $\bm\Omega_1,\ldots,\bm\Omega_q$ of the
skew-symmetric $r\times r$ matrices, where $q=r(r-1)/2$, with
$\langle\bm\Omega_\ell,\bm\Omega_m\rangle_F=\mathbf1\{\ell=m\}$.
Define
\[
 \bm R_i=(\bm\Omega_1^\top\bm z_i^*,\ldots,\bm\Omega_q^\top\bm z_i^*)
 \in\mathbb R^{r\times q},\qquad
 \bm R=\begin{pmatrix}\bm R_1\\ \vdots\\ \bm R_n\end{pmatrix}
 \in\mathbb R^{nr\times q}.
\]
For $\bm\Omega(\bm a)=\sum_\ell a_\ell\bm\Omega_\ell$,
$\bm R\bm a$ is the vector representation of $\bm Z^*\bm\Omega(\bm a)$.
Thus the column space of $\bm R$ represents $\mathcal R$.
Since $\bm Z^*$ has full column rank, $\bm R\bm a=\bm0$ implies
$\bm\Omega(\bm a)=\bm0$ and then $\bm a=\bm0$.
The projection onto this column space is therefore
$\bm R(\bm R^\top\bm R)^{-1}\bm R^\top$.
Combining the two projections in the full node coordinates gives
\[
 [\bm I-\bm\Pi_{\mathcal T_n}]_{ij}
 =\begin{pmatrix}
 0&\bm0_{1\times r}\\
 \bm0_{r\times1}&\bm I_r/n+\bm R_i(\bm R^\top\bm R)^{-1}\bm R_j^\top
 \end{pmatrix}.
\]
The first row and column correspond to the unrestricted node-effect coordinates.
It remains to bound the two factors involving $\bm R$.
By Assumption~\ref{assump: network} and
$\|\bm\Omega(\bm a)\|_F^2=\sum_\ell a_\ell^2=\|\bm a\|_2^2$,
\begin{align*}
 \|\bm R\bm a\|_2^2
 &=\tr\{\bm\Omega(\bm a)^\top\bm Z^{*\top}\bm Z^*\bm\Omega(\bm a)\}
 \ge c_Zn\|\bm a\|_2^2,\\
 \|\bm R_i\|_2^2
 &\le\|\bm R_i\|_F^2
 =\sum_\ell\|\bm\Omega_\ell^\top\bm z_i^*\|_2^2
 \le\sum_\ell\|\bm\Omega_\ell\|_F^2\|\bm z_i^*\|_2^2
 \le qR_Z^2.
\end{align*}
The first inequality gives $\|(\bm R^\top\bm R)^{-1}\|_2\le(c_Zn)^{-1}$.
Consequently,
\[
\|[\bm I-\bm\Pi_{\mathcal T_n}]_{ij}\|_2
 \le\frac1n+\|\bm R_i\|_2\|(\bm R^\top\bm R)^{-1}\|_2\|\bm R_j\|_2 
 \le\frac1n+\frac{qR_Z^2}{c_Zn}\le\frac Cn.
\]
For $r=1$, we have $\mathcal R=\{\bm0\}$, so the same bound holds with the second term omitted.
Finally,
\[
 \|\bm\Pi_{\mathcal T_n}\|_{\blinf}
 \le\|\bm I\|_{\blinf}+\|\bm I-\bm\Pi_{\mathcal T_n}\|_{\blinf}
 \le1+\max_i\sum_j C/n\le C.
\]
The desired results follow immediately. \hfill$\blacksquare$

\begin{lemma}\label{lem: local network curvature}
Suppose Assumptions~\ref{assump: network} and~\ref{assump: network growth} hold. Let
\[
 \mathcal N_n=\left\{\bm\nu\in\mathcal V_n:
 \psi(\bm\nu-\bm\nu^*)\in\mathcal T_n,\quad
 \|\psi(\bm\nu-\bm\nu^*)\|_2\le\varepsilon\sqrt n/K_n\right\},
\]
where $\varepsilon>0$ is a sufficiently small fixed constant.
For $\bm\nu_1,\bm\nu_2\in\mathcal N_n$, write
$\mathsf F_0=\mathsf F_0(\bm\nu_1,\bm\nu_2)$ as defined in
\eqref{eq: extended tangent Hessian}.
Uniformly over these endpoints, we have 
\begin{align}
\label{eq: population structure}
 cs_n\bm I\preceq\mathsf F_0\preceq Ct_n\bm I,\quad
 \max_{i\ne j}\|\mathsf F_{0,ij}\|_2\le Ct_n/n,\quad
 \|\mathsf F_0^{-1}\|_{\blinf}\le CK_n/s_n.
\end{align}
Moreover, for every $\bm\nu\in\mathcal N_n$, it holds true that
\begin{align}
\label{eq: population Hessian variation}
 \|\mathsf F_0(\bm\nu,\bm\nu^*)-\mathsf F_{0,*}\|_2
 \le Ct_n\|\psi(\bm\nu-\bm\nu^*)\|_{\blinf}.
\end{align}
\end{lemma}

\noindent\textbf{Proof of Lemma~\ref{lem: local network curvature}.}
For $\bm\nu\in\mathcal N_n$, write $\bm H_0(\bm\nu)=-\mathbb E_A\nabla_\psi^2L_A(\bm\nu)$ and
$\bm J=\diag(0,\bm I_r)$. Taking expectation in the likelihood Hessian
calculated in Lemma~\ref{lem: primitive network Fisher} gives
\begin{align}
\label{eq: exact Hessian identity}
 \bm H_0(\bm\nu)
 =\sum_{i<j}w_{ij}(\bm\nu)\bm\phi_{ij}(\bm\nu)\bm\phi_{ij}(\bm\nu)^\top
 +\{\bm P(\bm\nu)-\bm P^*\}\otimes\bm J.
\end{align}
Its node blocks are, with $\bm b_i=(1,\bm z_i^\top)^\top$,
\begin{align}
\label{eq: population Hessian node blocks}
 [\bm H_0(\bm\nu)]_{ii}=\sum_{j\ne i}w_{ij}(\bm\nu)\bm b_j\bm b_j^\top,\quad
 [\bm H_0(\bm\nu)]_{ij}=w_{ij}(\bm\nu)\bm b_j\bm b_i^\top
   +\{p_{ij}(\bm\nu)-p_{ij}^*\}\bm J,\ i\ne j.
\end{align}

We first prove the curvature bound in \eqref{eq: population structure}.
For $\psi(\bm h)\in\mathcal T_n$, the directional derivative formula gives
\begin{align*}
 \left\{\sum_{i<j}
 (\dot\theta_{ij}(\bm\nu)[\bm h]-\dot\theta_{ij}^*[\bm h])^2\right\}^{1/2}
 \le\|\bm E_Z(\bm Z-\bm Z^*)^\top+(\bm Z-\bm Z^*)\bm E_Z^\top\|_F\le2\|\bm Z-\bm Z^*\|_2\|\bm E_Z\|_F.
\end{align*}
By \eqref{eq: offdiag derivative geometry} and
$\|\bm Z-\bm Z^*\|_2\le\varepsilon\sqrt n/K_n\le\varepsilon\sqrt n$,
the triangle inequality therefore yields
\[
 cn\|\psi(\bm h)\|_2^2
 \le\sum_{i<j}\dot\theta_{ij}(\bm\nu)[\bm h]^2
 \le Cn\|\psi(\bm h)\|_2^2
\]
for sufficiently small $\varepsilon$. For $\psi(\bm h)\in\mathcal T_n$, the quadratic
form of the first term in \eqref{eq: exact Hessian identity} is
\[
 \psi(\bm h)^\top
 \left\{\sum_{i<j}w_{ij}(\bm\nu)\bm\phi_{ij}(\bm\nu)\bm\phi_{ij}(\bm\nu)^\top\right\}
 \psi(\bm h)
 =\sum_{i<j}w_{ij}(\bm\nu)\dot\theta_{ij}(\bm\nu)[\bm h]^2,
\]
since $\bm\phi_{ij}(\bm\nu)^\top\psi(\bm h)=\dot\theta_{ij}(\bm\nu)[\bm h]$.
By Lemma~\ref{lem: network consequences}, $w_{-,n}\le w_{ij}(\bm\nu)\le w_{+,n}$.
Combining these weight bounds with the preceding derivative bounds gives
\begin{align*}
 cs_n\|\psi(\bm h)\|_2^2
 &=cnw_{-,n}\|\psi(\bm h)\|_2^2
 \le w_{-,n}\sum_{i<j}\dot\theta_{ij}(\bm\nu)[\bm h]^2\le\sum_{i<j}w_{ij}(\bm\nu)\dot\theta_{ij}(\bm\nu)[\bm h]^2\\
 &\le w_{+,n}\sum_{i<j}\dot\theta_{ij}(\bm\nu)[\bm h]^2\le Cnw_{+,n}\|\psi(\bm h)\|_2^2
 =Ct_n\|\psi(\bm h)\|_2^2.
\end{align*}

For the remaining term, put $\bm a=\bm\alpha-\bm\alpha^*$,
$\bm E=\bm Z-\bm Z^*$ and $d_\nu=\|\psi(\bm\nu-\bm\nu^*)\|_2$.
Writing $\bm\Theta(\bm\nu)=\bm\Theta(\bm\alpha,\bm Z)$, the full logit difference satisfies
\begin{align*}
 \bm\Theta(\bm\nu)-\bm\Theta(\bm\nu^*)
 =\bm a\bm1_n^\top+\bm1_n\bm a^\top
  +\bm E\bm Z^\top+\bm Z^*\bm E^\top.
\end{align*}
Since $\|\bm Z\|_2+\|\bm Z^*\|_2\le2R_Z\sqrt n$, its Frobenius norm is
at most $2\sqrt n\|\bm a\|_2+2R_Z\sqrt n\|\bm E\|_F\le C\sqrt n d_\nu$.
The derivative of $\expit$ is at most $w_{+,n}$ between the endpoint logits.
Applying the mean-value theorem entrywise, with the diagonal of $\bm P$ kept zero, gives
\begin{align*}
 \|\bm P(\bm\nu)-\bm P^*\|_2
 \le\|\bm P(\bm\nu)-\bm P^*\|_F
 \le Cw_{+,n}\sqrt n d_\nu
 \le C\varepsilon t_n/K_n=C\varepsilon s_n.
\end{align*}
This also bounds the spectral norm of the remaining term in
\eqref{eq: exact Hessian identity}, since $\|\bm J\|_2=1$.
Choose $\varepsilon$ small enough to preserve the lower curvature bound.
The set $\mathcal N_n$ is convex, so the bound holds along every segment between
its endpoints. Projecting and integrating as in \eqref{eq: extended tangent Hessian}
gives the required bounds on $\mathcal T_n$; on $\mathcal T_n^\perp$,
the added term is $s_n\bm I$. This proves the full matrix curvature bound.

We next bound the off-diagonal blocks of $\mathsf F_0$.
For any two probabilities in $[p_{-,n},p_{+,n}]$, their difference is at most
$4w_{+,n}$: below $1/2$, $p\le2w_{+,n}$; above $1/2$, $1-p\le2w_{+,n}$;
if the interval contains $1/2$, $w_{+,n}=1/4$.
Since $\|\bm b_i\|_2^2=1+\|\bm z_i\|_2^2\le1+R_Z^2$ and
$w_{ij}(\bm\nu)\le w_{+,n}$, the diagonal formula in
\eqref{eq: population Hessian node blocks} gives
\[
 \|[\bm H_0(\bm\nu)]_{ii}\|_2
 \le\sum_{j\ne i}w_{ij}(\bm\nu)\|\bm b_j\|_2^2
 \le Ct_n.
\]
For $i\ne j$, both $p_{ij}(\bm\nu)$ and $p_{ij}^*$ belong to
$[p_{-,n},p_{+,n}]$, so the preceding probability bound gives
$|p_{ij}(\bm\nu)-p_{ij}^*|\le4w_{+,n}$.
Using the off-diagonal formula and $\|\bm J\|_2=1$, we obtain
\begin{align*}
 \|[\bm H_0(\bm\nu)]_{ij}\|_2
 &\le w_{ij}(\bm\nu)\|\bm b_j\|_2\|\bm b_i\|_2
       +|p_{ij}(\bm\nu)-p_{ij}^*|\|\bm J\|_2\\
 &\le w_{+,n}(1+R_Z^2)+4w_{+,n}
 \le Cw_{+,n}=Ct_n/n.
\end{align*}
Let $\bm\Pi_{\mathcal T_n^\perp}=\bm I-\bm\Pi_{\mathcal T_n}$. By \eqref{eq: network projector bounds},
for all $i,j$,
\[
 \|[\bm\Pi_{\mathcal T_n^\perp}\bm H_0]_{ij}\|_2
 \le\frac Cn\sum_\ell\|[\bm H_0]_{\ell j}\|_2\le Ct_n/n,
 \quad
 \|[\bm\Pi_{\mathcal T_n^\perp}\bm H_0\bm\Pi_{\mathcal T_n^\perp}]_{ij}\|_2
 \le\sum_\ell\frac{Ct_n}{n}\frac Cn\le Ct_n/n.
\]
The bound for $\bm H_0\bm\Pi_{\mathcal T_n^\perp}$ follows by transposing $\bm\Pi_{\mathcal T_n^\perp}\bm H_0$.
For $\bm\nu(v)=\bm\nu_2+v(\bm\nu_1-\bm\nu_2)$, definition
\eqref{eq: extended tangent Hessian} and $\bm\Pi_{\mathcal T_n}=\bm I-\bm\Pi_{\mathcal T_n^\perp}$ give
\begin{align*}
 \mathsf F_{0,ij}
 ={}&\int_0^1\bigl\{
 [\bm H_0\{\bm\nu(v)\}]_{ij}
 -[\bm\Pi_{\mathcal T_n^\perp}\bm H_0\{\bm\nu(v)\}]_{ij}\\
 &\quad-[\bm H_0\{\bm\nu(v)\}\bm\Pi_{\mathcal T_n^\perp}]_{ij}
 +[\bm\Pi_{\mathcal T_n^\perp}\bm H_0\{\bm\nu(v)\}\bm\Pi_{\mathcal T_n^\perp}]_{ij}
 \bigr\}\,dv+s_n[\bm\Pi_{\mathcal T_n^\perp}]_{ij}.
\end{align*}
For $i\ne j$, each of the four blocks in the integrand has norm at most
$Ct_n/n$, uniformly in $v\in[0,1]$. Since $\|[\bm\Pi_{\mathcal T_n^\perp}]_{ij}\|_2\le C/n$ and
$s_n\le t_n$,
\[
 \|\mathsf F_{0,ij}\|_2
 \le\frac{4Ct_n+Cs_n}{n}\le\frac{Ct_n}{n},
\]
which proves the off-diagonal bound in \eqref{eq: population structure}.

To bound the inverse, let $\mathsf R_i$ be block row $i$ of $\mathsf F_0$
with its diagonal block set to zero. Then
$\|\mathsf R_i\|_F\le C\sqrt n(t_n/n)=Ct_n/\sqrt n$,
and block row $i$ of $\mathsf F_0\mathsf F_0^{-1}=\bm I$ gives
\[
 [\mathsf F_0^{-1}]_{i\cdot}
 =\mathsf F_{0,ii}^{-1}[\bm I]_{i\cdot}
 -\mathsf F_{0,ii}^{-1}\mathsf R_i\mathsf F_0^{-1}.
\]
The curvature bound gives
$\|\mathsf F_{0,ii}^{-1}\|_2+\|\mathsf F_0^{-1}\|_2\le C/s_n$.
Using Cauchy--Schwarz to sum the $n$ blocks in the second term, we obtain
\[
 \sum_j\|[\mathsf F_0^{-1}]_{ij}\|_2
 \le\frac C{s_n}
 +\sqrt n\|\mathsf F_{0,ii}^{-1}\mathsf R_i\mathsf F_0^{-1}\|_F
 \le\frac C{s_n}+\frac{Ct_n}{s_n^2}
 \le\frac{CK_n}{s_n}.
\]
Taking the maximum over $i$ completes \eqref{eq: population structure}.

For \eqref{eq: population Hessian variation}, put
$\delta_\nu=\|\psi(\bm\nu-\bm\nu^*)\|_{\blinf}$.
The identity
$\theta_{ij}(\bm\nu)-\theta_{ij}^*
=a_i+a_j+(\bm z_i-\bm z_i^*)^\top\bm z_j
+\bm z_i^{*\top}(\bm z_j-\bm z_j^*)$
gives $\max_{i<j}|\theta_{ij}(\bm\nu)-\theta_{ij}^*|\le C\delta_\nu$.
The first and second derivatives of $\expit$ are bounded in absolute value by $w_{+,n}$ on the logit envelope, so
\[
 \max_{i<j}\{|p_{ij}(\bm\nu)-p_{ij}^*|+|w_{ij}(\bm\nu)-w_{ij}^*|\}
 \le Cw_{+,n}\delta_\nu.
\]
For $\bm b_i^*=(1,\bm z_i^{*\top})^\top$, we have
$\|\bm b_i-\bm b_i^*\|_2\le\delta_\nu$ and
$\|\bm b_i\|_2+\|\bm b_i^*\|_2\le C$.
For any node indices $\ell,m$ and $i\ne j$, expand the weighted product as
\begin{align*}
 w_{ij}(\bm\nu)\bm b_\ell\bm b_m^\top
       -w_{ij}^*\bm b_\ell^*\bm b_m^{*\top}=\{w_{ij}(\bm\nu)-w_{ij}^*\}\bm b_\ell\bm b_m^\top
       +w_{ij}^*(\bm b_\ell-\bm b_\ell^*)\bm b_m^\top
       +w_{ij}^*\bm b_\ell^*(\bm b_m-\bm b_m^*)^\top.
\end{align*}
The preceding bounds show that each term has spectral norm at most
$Cw_{+,n}\delta_\nu$.
For a diagonal block, apply this bound with $\ell=m=j$ to each summand in
\eqref{eq: population Hessian node blocks}:
\begin{align*}
 \|[\bm H_0(\bm\nu)-\bm H_0(\bm\nu^*)]_{ii}\|_2
 \le\sum_{j\ne i}
 \|w_{ij}(\bm\nu)\bm b_j\bm b_j^\top-w_{ij}^*\bm b_j^*\bm b_j^{*\top}\|_2\le C(n-1)w_{+,n}\delta_\nu\le Ct_n\delta_\nu.
\end{align*}
For $i\ne j$, apply the same bound with $\ell=j,m=i$. So
\begin{align*}
 \|[\bm H_0(\bm\nu)-\bm H_0(\bm\nu^*)]_{ij}\|_2
 \le\|w_{ij}(\bm\nu)\bm b_j\bm b_i^\top
           -w_{ij}^*\bm b_j^*\bm b_i^{*\top}\|_2
       +|p_{ij}(\bm\nu)-p_{ij}^*|\|\bm J\|_2\le Cw_{+,n}\delta_\nu,
\end{align*}
where the last inequality uses the preceding probability bound and $\|\bm J\|_2=1$.
Let $\bm B=\bm H_0(\bm\nu)-\bm H_0(\bm\nu^*)$.
For a vector $\bm u$ with node blocks $\bm u_i$, symmetry gives
$\|\bm B_{ij}\|_2=\|\bm B_{ji}\|_2$, and hence
\begin{align*}
 |\bm u^\top\bm B\bm u|
 &\le\sum_{i,j}\|\bm B_{ij}\|_2\|\bm u_i\|_2\|\bm u_j\|_2\le\frac12\sum_{i,j}\|\bm B_{ij}\|_2
       (\|\bm u_i\|_2^2+\|\bm u_j\|_2^2)\\
& =\sum_i\|\bm u_i\|_2^2\sum_j\|\bm B_{ij}\|_2\le\left(\max_i\sum_j\|\bm B_{ij}\|_2\right)\|\bm u\|_2^2.
\end{align*}
Taking the supremum over $\|\bm u\|_2=1$ and using the diagonal and
 off-diagonal bounds above yields
\begin{align*}
 \|\bm H_0(\bm\nu)-\bm H_0(\bm\nu^*)\|_2
 \le\max_i\left\{\|\bm B_{ii}\|_2+\sum_{j\ne i}\|\bm B_{ij}\|_2\right\}\le Ct_n\delta_\nu.
\end{align*}

To obtain \eqref{eq: population Hessian variation}, let
$\bm\nu(v)=\bm\nu^*+v(\bm\nu-\bm\nu^*)$, $0\le v\le1$.
Convexity gives $\bm\nu(v)\in\mathcal N_n$, and
$\delta_{\nu(v)}=v\delta_\nu$.
In the difference of the two matrices defined in \eqref{eq: extended tangent Hessian},
the added terms $s_n(\bm I-\bm\Pi_{\mathcal T_n})$ cancel. Thus
\[
 \mathsf F_0(\bm\nu,\bm\nu^*)-\mathsf F_{0,*}
 =\bm\Pi_{\mathcal T_n}\int_0^1
   [\bm H_0\{\bm\nu(v)\}-\bm H_0(\bm\nu^*)]\,dv\bm\Pi_{\mathcal T_n}.
\]
Since $\|\bm\Pi_{\mathcal T_n}\|_2\le1$, the bound just proved gives
\begin{align*}
 \|\mathsf F_0(\bm\nu,\bm\nu^*)-\mathsf F_{0,*}\|_2
 \le\int_0^1\|\bm H_0\{\bm\nu(v)\}-\bm H_0(\bm\nu^*)\|_2\,dv\le Ct_n\delta_\nu\int_0^1v\,dv
 \le Ct_n\delta_\nu.
\end{align*}
This proves \eqref{eq: population Hessian variation} by noting that $\delta_\nu=\|\psi(\bm\nu-\bm\nu^*)\|_{\blinf}$. 
\hfill$\blacksquare$

\begin{lemma}\label{lem: observed network curvature}
Suppose Assumptions~\ref{assump: network} and~\ref{assump: network growth} hold.
Let $\mathcal N_n$ be as in Lemma~\ref{lem: local network curvature}.
For endpoints $\bm\nu_1,\bm\nu_2\in\mathcal N_n$, abbreviate the matrices
in \eqref{eq: extended tangent Hessian} as $\mathsf F_A$ and $\mathsf F_0$.
The random difference does not depend on the endpoints:
\begin{align}
\label{eq: exact Hessian noise}
 \mathsf F_A=\mathsf F_0+\bm N_A,\qquad
 \bm N_A=-\bm\Pi_{\mathcal T_n}
 \{(\bm A-\bm P^*)\otimes\diag(0,\bm I_r)\}
 \bm\Pi_{\mathcal T_n}.
\end{align}
In particular, $\|\mathsf F_A-\mathsf F_0\|_2\le\|\bm A-\bm P^*\|_2$.
With probability tending to one, uniformly over $\mathcal N_n^2$,
\begin{align}
\label{eq: observed curvature and inverse}
 cs_n\bm I\preceq\mathsf F_A\preceq Ct_n\bm I,\quad
 \|\mathsf F_A^{-1}\|_{\blinf}
 \le\frac{C}{s_n}\left\{K_n^{2k-1}+\sqrt n(K_n/s_n)^{k/2}\right\}.
\end{align}
\end{lemma}

\noindent\textbf{Proof of Lemma~\ref{lem: observed network curvature}.}
Write $\bm J=\diag(0,\bm I_r)$ and
$\bm H_0(\bm\nu)=-\mathbb E_A\nabla_\psi^2L_A(\bm\nu)$ as in
Lemma~\ref{lem: local network curvature}.
In the observed negative Hessian, the residual term is
$\{\bm P(\bm\nu)-\bm A\}\otimes\bm J$; in its expectation, it is
$\{\bm P(\bm\nu)-\bm P^*\}\otimes\bm J$. The other terms agree, so $-\nabla_\psi^2L_A(\bm\nu)-\bm H_0(\bm\nu) =-(\bm A-\bm P^*)\otimes\bm J$.
This difference is independent of $\bm\nu$. In
\eqref{eq: extended tangent Hessian}, the added terms
$s_n(\bm I-\bm\Pi_{\mathcal T_n})$ cancel when subtracting
$\mathsf F_0$ from $\mathsf F_A$. Projecting and integrating the preceding
identity therefore gives \eqref{eq: exact Hessian noise}.
Since $\|\bm\Pi_{\mathcal T_n}\|_2\le1$ and $\|\bm J\|_2=1$, $\|\bm N_A\|_2 \le\|(\bm A-\bm P^*)\otimes\bm J\|_2 =\|\bm A-\bm P^*\|_2$.
We work on the event where both bounds in
\eqref{eq: adjacency spectral bound} hold, whose probability tends to one
by Lemma~\ref{lem: network consequences}.
The $(i,j)$ block of $(\bm A-\bm P^*)\otimes\bm J$ is
$(a_{ij}-p_{ij}^*)\bm J$, and its diagonal blocks are zero.
Using \eqref{eq: node product bounds} and the projector bound in
\eqref{eq: network projector bounds}, we obtain
\begin{align*}
 \|\bm N_A\|_2&\le C\sqrt{t_n},\quad \|\bm N_A\|_{\blinf}\le\|\bm\Pi_{\mathcal T_n}\|_{\blinf}^2
       \max_i\sum_{j\ne i}|a_{ij}-p_{ij}^*|\|\bm J\|_2
 \le Ct_n.
\end{align*}
By \eqref{eq: basic growth consequences} and $t_n=K_ns_n$,
$\sqrt{t_n}/s_n=\sqrt{K_n/s_n}\to0$.
For every unit vector $\bm u$, \eqref{eq: population structure} thus gives
\[
 cs_n-\|\bm N_A\|_2
 \le\bm u^\top\mathsf F_A\bm u
 \le Ct_n+\|\bm N_A\|_2.
\]
For sufficiently large $n$, the left side is at least $cs_n/2$;
the right side is at most $Ct_n$ after enlarging $C$, since $s_n\le t_n$.
This proves the curvature bounds in \eqref{eq: observed curvature and inverse}
and, in particular, $\|\mathsf F_A^{-1}\|_2\le C/s_n$.

To bound $\|\mathsf F_A^{-1}\|_{\blinf}$, put
$\bm M=\mathsf F_0^{-1}\bm N_A$.
The inverse bounds in \eqref{eq: population structure} give
\[
 \|\bm M\|_{\blinf}
 \le\frac{CK_n}{s_n}\,Ct_n\le CK_n^2,
 \qquad
 \|\bm M\|_2
 \le\frac C{s_n}\,C\sqrt{t_n}\le C\sqrt{K_n/s_n}.
\]
Multiplying $\mathsf F_A=\mathsf F_0+\bm N_A$ on the left by
$\mathsf F_0^{-1}$ and on the right by $\mathsf F_A^{-1}$ gives $\mathsf F_A^{-1}=\mathsf F_0^{-1}-\bm M\mathsf F_A^{-1}$.
Substituting this identity into its last term repeatedly, for the fixed
integer $k$ in Assumption~\ref{assump: network growth}, yields
\[
 \mathsf F_A^{-1}
 =\sum_{\ell=0}^{k-1}(-\bm M)^\ell\mathsf F_0^{-1}
   +(-\bm M)^k\mathsf F_A^{-1}.
\]
For the finite sum, \eqref{eq: node product bounds} and $K_n\ge1$ imply
\[
 \left\|\sum_{\ell=0}^{k-1}(-\bm M)^\ell\mathsf F_0^{-1}\right\|_{\blinf}
 \le\sum_{\ell=0}^{k-1}(CK_n^2)^\ell\frac{CK_n}{s_n}
 \le\frac{CK_n^{2k-1}}{s_n},
\]
where $C$ may depend on the fixed $k$.
For the remainder, we use its spectral norm to bound its block norm.
If $\bm C$ has $(r+1)\times(r+1)$ blocks, its $i$th block row
$\bm C_{i\cdot}$ has $r+1$ rows, so
$\|\bm C_{i\cdot}\|_F\le\sqrt{r+1}\|\bm C_{i\cdot}\|_2
\le\sqrt{r+1}\|\bm C\|_2$.
Consequently, Cauchy--Schwarz gives
\[
 \sum_j\|\bm C_{ij}\|_2
 \le\sqrt n\left(\sum_j\|\bm C_{ij}\|_F^2\right)^{1/2}
 =\sqrt n\|\bm C_{i\cdot}\|_F
 \le\sqrt{n(r+1)}\|\bm C\|_2.
\]
Taking the maximum over $i$ and applying this bound to the remainder yields
\[
 \|(-\bm M)^k\mathsf F_A^{-1}\|_{\blinf}
 \le\sqrt{n(r+1)}\|\bm M\|_2^k\|\mathsf F_A^{-1}\|_2
 \le\frac{C\sqrt n}{s_n}(K_n/s_n)^{k/2}.
\]
Combining the finite sum and remainder bounds proves
\[
 \|\mathsf F_A^{-1}\|_{\blinf}
 \le\frac C{s_n}\left\{K_n^{2k-1}+\sqrt n(K_n/s_n)^{k/2}\right\},
\]
as required in \eqref{eq: observed curvature and inverse}.
The event used above depends only on $\bm A$, and all bounds from
Lemma~\ref{lem: local network curvature} hold uniformly over
$\bm\nu_1,\bm\nu_2\in\mathcal N_n$.
The conclusions therefore hold uniformly over these endpoints on the same event.
\hfill$\blacksquare$

\begin{lemma}\label{lem: local likelihood root}
Suppose Assumptions~\ref{assump: network} and~\ref{assump: network growth} hold.  Then, with probability tending to one, there is a
local maximizer $\overline{\bm\nu}=(\overline{\bm\alpha},\overline{\bm Z})$ of $L_A$ over the centered parameter space, lying in the relative interior of $\mathcal N_n$ and satisfying
$\bm\Pi_{\mathcal T_n}\nabla_\psi L_A(\overline{\bm\nu})=\bm0$.
Its aligned error $\bm\delta=\psi(\overline{\bm\nu}-\bm\nu^*)$
satisfies
\begin{align}
\label{eq: exact local root max rate}
 \|\bm\delta\|_{\blinf}=O_p(\delta_{n,k}),\
 n^{-1}\|\bm\delta\|_2^2=O_p(s_n^{-1}),\
 n^{-1}\sum_i\|\bm\delta_i\|_2^4=O_p(s_n^{-2}),
\end{align}
where $\bm\delta_i\in\mathbb R^{r+1}$ is the $i$-th node block of $\bm\delta$.
\end{lemma}

\noindent\textbf{Proof of Lemma~\ref{lem: local likelihood root}.}
The $i$-th node score at $\bm\nu^*$ is
$\sum_{j\ne i}(a_{ij}-p_{ij}^*)\bm b_j^*$, where
$\bm b_j^*=(1,\bm z_j^{*\top})^\top$.
Each coordinate is a sum of independent centered variables bounded by $C$,
with total variance at most $C\sum_{j\ne i}w_{ij}^*\le Ct_n$.
Bernstein's inequality, a union bound over $n(r+1)$ coordinates, and
\eqref{eq: network projector bounds} give
\[
 \|\bm g_A\|_{\blinf}
 \le C\{\sqrt{t_n\log n}+\log n\}\le C\sqrt{t_n\log n}
\]
with probability tending to one, since $t_n\ge s_n\gg\log n$ by
\eqref{eq: basic growth consequences}.
Work on the event where this bound and the uniform bounds of
Lemma~\ref{lem: observed network curvature} hold.
For a sufficiently large fixed $C_1$, let $\mathcal B_n=\{\bm d\in\mathcal T_n:\|\bm d\|_{\blinf}\le C_1\delta_{n,k}\}$, $\bm\nu(\bm d)=\bm\nu^*+\psi^{-1}(\bm d)$.
For $\bm d\in\mathcal B_n$, each node changes by at most $C_1\delta_{n,k}=o(1)$,
so the fixed margins in \eqref{eq: hetero primitive bounds} in Assumption \ref{assump: network} preserve all
parameter inequalities. Moreover,
$\|\bm d\|_2\le C_1\sqrt n\delta_{n,k}=o(\sqrt n/K_n)$.
Thus $\bm\nu(\bm d)\in\mathcal N_n$. For
$T(\bm d)=\mathsf F_A\{\bm\nu(\bm d),\bm\nu^*\}^{-1}\bm g_A$,
\eqref{eq: observed curvature and inverse} and $t_n=K_ns_n$ give
\[
 \|T(\bm d)\|_{\blinf}
 \le\frac C{s_n}\{K_n^{2k-1}+\sqrt n(K_n/s_n)^{k/2}\}
       \sqrt{t_n\log n}
 =C\delta_{n,k}\le C_1\delta_{n,k},
\]
where the equality uses \eqref{eq: hetero rates definitions}.
Continuity and positive definiteness of $\mathsf F_A$ make $T$ continuous.
Since $\bm g_A\in\mathcal T_n$ and $\mathsf F_A^{-1}$ preserves $\mathcal T_n$,
$T$ maps the compact convex ball $\mathcal B_n$ into itself.
Brouwer's theorem gives $\overline{\bm d}\in\mathcal B_n$ such that
$\overline{\bm d}=T(\overline{\bm d})$.
Define the estimator and its aligned error by $\overline{\bm\nu}=\bm\nu(\overline{\bm d}) =\bm\nu^*+\psi^{-1}(\overline{\bm d})$, $\bm\delta=\psi(\overline{\bm\nu}-\bm\nu^*)=\overline{\bm d}$.
Thus the fixed point is exactly the estimation error, and
$\bm\delta=T(\bm\delta)$.
Let $\bm\nu_v=\bm\nu^*+v(\overline{\bm\nu}-\bm\nu^*)$, $0\le v\le1$.
Integrating the derivative of the score along this segment gives
\[
 \nabla_\psi L_A(\overline{\bm\nu})-\nabla_\psi L_A(\bm\nu^*)
 =\int_0^1\nabla_\psi^2L_A(\bm\nu_v)\bm\delta\,dv.
\]
Since $\bm\delta\in\mathcal T_n$, we have
$\bm\Pi_{\mathcal T_n}\bm\delta=\bm\delta$ and
$(\bm I-\bm\Pi_{\mathcal T_n})\bm\delta=\bm0$.
Projecting the identity and using \eqref{eq: extended tangent Hessian}
and \eqref{eq: projected network score definition} gives
\[
 \bm\Pi_{\mathcal T_n}\nabla_\psi L_A(\overline{\bm\nu})
 =\bm g_A-\mathsf F_A(\overline{\bm\nu},\bm\nu^*)\bm\delta=\bm0,
\]
where the last equality follows from the fixed-point equation
$\bm\delta=\mathsf F_A(\overline{\bm\nu},\bm\nu^*)^{-1}\bm g_A$.
Thus $\overline{\bm\nu}$ has zero directional derivative along every
$\bm h$ with $\psi(\bm h)\in\mathcal T_n$, satisfies the parameter inequalities strictly,
and has $\|\bm\delta\|_{\blinf}=O_p(\delta_{n,k})$.
To show that it is a maximizer on $\mathcal N_n$, take any
$\bm\nu\in\mathcal N_n$. By the definition of this set, $\bm u=\psi(\bm\nu-\overline{\bm\nu}) =\psi(\bm\nu-\bm\nu^*)-\psi(\overline{\bm\nu}-\bm\nu^*) \in\mathcal T_n$.
The segment $\bm\nu_t=\overline{\bm\nu}+t(\bm\nu-\overline{\bm\nu})$,
$0\le t\le1$, stays in $\mathcal N_n$ by convexity.
Since $\bm\Pi_{\mathcal T_n}\bm u=\bm u$ and
$(\bm I-\bm\Pi_{\mathcal T_n})\bm u=\bm0$,
\eqref{eq: extended tangent Hessian} at equal endpoints and
\eqref{eq: observed curvature and inverse} give
\[
 \bm u^\top\nabla_\psi^2L_A(\bm\nu_t)\bm u
 =-\bm u^\top\mathsf F_A(\bm\nu_t,\bm\nu_t)\bm u
 \le-cs_n\|\bm u\|_2^2.
\]
Moreover, $\bm u^\top\nabla_\psi L_A(\overline{\bm\nu})
=\bm u^\top\bm\Pi_{\mathcal T_n}\nabla_\psi L_A(\overline{\bm\nu})=0$.
Taylor's formula therefore yields
\begin{align*}
 L_A(\bm\nu)-L_A(\overline{\bm\nu})
 &=\bm u^\top\nabla_\psi L_A(\overline{\bm\nu})
   +\int_0^1(1-t)\bm u^\top\nabla_\psi^2L_A(\bm\nu_t)\bm u\,dt\\
 &\le-cs_n\|\bm u\|_2^2\int_0^1(1-t)\,dt
 =-\frac{cs_n}{2}\|\bm u\|_2^2\le0.
\end{align*}
Hence $\overline{\bm\nu}$ is a maximizer of $L_A$ on $\mathcal N_n$.

To obtain a local maximizer over the centered parameter space, note that
$\bm Z^{*\top}\overline{\bm Z}$ is symmetric and, for sufficiently small
$\varepsilon$,
\[
 \begin{aligned}
 \lambda_{\min}(\bm Z^{*\top}\overline{\bm Z})
 &\ge c_Zn-\|\bm Z^*\|_2\|\overline{\bm Z}-\bm Z^*\|_F\\
 &\ge c_Zn-R_Z\varepsilon n/K_n
 \ge(c_Z-R_Z\varepsilon)n>0.
 \end{aligned}
\]
Thus the orthogonal Procrustes minimizer $\bm O(\bm Z)$ of
$\|\bm Z\bm O-\bm Z^*\|_F$ is unique and continuous near
$\overline{\bm Z}$, with $\bm O(\overline{\bm Z})=\bm I_r$.
For any centered $\bm\nu=(\bm\alpha,\bm Z)$ sufficiently close to
$\overline{\bm\nu}$, its aligned version
$(\bm\alpha,\bm Z\bm O(\bm Z))$ is still close to $\overline{\bm\nu}$
and satisfies the defining linear constraints of $\mathcal N_n$.
Since $\overline{\bm\nu}$ lies in its relative interior, this aligned
version belongs to $\mathcal N_n$.
Orthogonal alignment preserves centering, row norms and all latent
inner products, so it leaves $L_A$ unchanged.
The maximization on $\mathcal N_n$ therefore implies that
$\overline{\bm\nu}$ is a local maximizer over the centered parameter space.

To prove the second- and fourth-moment rates, first consider the linear term
$\bm\delta^{(1)}=\mathsf F_{0,*}^{-1}\bm g_A$, with node blocks $\bm\delta_i^{(1)}\in\mathbb R^{r+1}$.
For an edge $e=(i,j)$, $i<j$, write
$a_e=a_{ij}$, $p_e^*=p_{ij}^*$, $w_e^*=w_{ij}^*$ and
$\bm\phi_e^*=\bm\phi_{ij}(\bm\nu^*)$.
For coordinate $q\in[n(r+1)]$, let $c_{e,q}$ be coordinate $q$ of
$\mathsf F_{0,*}^{-1}\bm\Pi_{\mathcal T_n}\bm\phi_e^*$, so $(\bm\delta^{(1)})_q=\sum_e c_{e,q}(a_e-p_e^*)$.
The block bounds for the inverse and projector in Lemmas~\ref{lem: primitive network Fisher} and \ref{lem: local network curvature}, together with
$\|\bm\phi_e^*\|_{\blinf}\le C$, give
$\max_e|c_{e,q}|\le CK_n/s_n$.
Edge independence and $\operatorname{Var}_A(a_e)=w_e^*$ give the
information identity
\[
 \operatorname{Cov}_A\{\nabla_\psi L_A(\bm\nu^*)\}
 =\sum_e w_e^*\bm\phi_e^*\bm\phi_e^{*\top}
 =-\mathbb E_A\nabla_\psi^2L_A(\bm\nu^*),
\]
where the last equality follows by taking expectation in the Hessian:
the residual terms proportional to $a_e-p_e^*$ have mean zero.
Using \eqref{eq: extended tangent Hessian} at the true endpoints and
$\bm\Pi_{\mathcal T_n}(\bm I-\bm\Pi_{\mathcal T_n})=\bm0$, we obtain
\[
 \operatorname{Cov}_A(\bm g_A)
 =\bm\Pi_{\mathcal T_n}
   \left(\sum_e w_e^*\bm\phi_e^*\bm\phi_e^{*\top}\right)
   \bm\Pi_{\mathcal T_n}
 =\bm\Pi_{\mathcal T_n}\mathsf F_{0,*}\bm\Pi_{\mathcal T_n}.
\]
The definition of $\mathsf F_{0,*}$ also gives
$\mathsf F_{0,*}\bm\Pi_{\mathcal T_n}
=\bm\Pi_{\mathcal T_n}\mathsf F_{0,*}$, so its inverse commutes with the
projection. Since $\bm\delta^{(1)}=\mathsf F_{0,*}^{-1}\bm g_A$,
\[
 \operatorname{Cov}_A(\bm\delta^{(1)})
 =\mathsf F_{0,*}^{-1}\operatorname{Cov}_A(\bm g_A)\mathsf F_{0,*}^{-1}
 =\bm\Pi_{\mathcal T_n}\mathsf F_{0,*}^{-1}\bm\Pi_{\mathcal T_n}.
\]
For the $q$th standard basis vector $\bm e_q\in\mathbb R^{n(r+1)}$,
independence and this covariance formula give
\begin{align*}
 \sum_e w_e^*c_{e,q}^2
 =\operatorname{Var}_A\{(\bm\delta^{(1)})_q\}
 =\bm e_q^\top\bm\Pi_{\mathcal T_n}\mathsf F_{0,*}^{-1}
       \bm\Pi_{\mathcal T_n}\bm e_q\le\|\mathsf F_{0,*}^{-1}\|_2
       \|\bm\Pi_{\mathcal T_n}\bm e_q\|_2^2
 \le C/s_n,
\end{align*}
where \eqref{eq: population structure} bounds the inverse spectral norm
by $C/s_n$, and $\|\bm\Pi_{\mathcal T_n}\bm e_q\|_2\le1$.
Expanding the fourth power of the centered independent sum yields
\begin{align*}
 \mathbb E_A\{(\bm\delta^{(1)})_q^4\}
 &=\sum_e c_{e,q}^4\mathbb E_A(a_e-p_e^*)^4
   +6\sum_{e<f}c_{e,q}^2c_{f,q}^2w_e^*w_f^*\\
 &\le(\max_e|c_{e,q}|^2)\sum_e w_e^*c_{e,q}^2
       +3\left(\sum_e w_e^*c_{e,q}^2\right)^2
 \le C(K_n^2s_n^{-3}+s_n^{-2})\le Cs_n^{-2},
\end{align*}
where $e<f$ counts each pair of distinct edges once,
$\mathbb E_A(a_e-p_e^*)^4\le w_e^*$, and $K_n^2/s_n\to0$ by
\eqref{eq: basic growth consequences}.
Since $\mathbb E_A(a_e-p_e^*)=0$, the coordinate representation gives
$\mathbb E_A(\bm\delta^{(1)})_q=0$. Hence its second moment equals its
variance, which is at most $C/s_n$, and
\[
 \mathbb E_A\!\left[\frac1n\|\bm\delta^{(1)}\|_2^2\right]
 =\frac1n\sum_{q=1}^{n(r+1)}\mathbb E_A\{(\bm\delta^{(1)})_q^2\}
 \le\frac{C(r+1)}{s_n}\le\frac C{s_n}.
\]
Writing $\delta_{i,a}^{(1)}$ for coordinate $a$ of node block
$\bm\delta_i^{(1)}$, Cauchy--Schwarz gives
$\|\bm\delta_i^{(1)}\|_2^4
\le(r+1)\sum_{a=1}^{r+1}(\delta_{i,a}^{(1)})^4$. Therefore,
\[
 \mathbb E_A\!\left[\frac1n\sum_i\|\bm\delta_i^{(1)}\|_2^4\right]
 \le\frac{r+1}{n}\sum_i\sum_{a=1}^{r+1}
       \mathbb E_A(\delta_{i,a}^{(1)})^4
 \le\frac C{s_n^2},
\]
where $r$ is fixed in both bounds. Applying Markov's inequality gives
\[
 n^{-1}\|\bm\delta^{(1)}\|_2^2=O_p(s_n^{-1}),\qquad
 n^{-1}\sum_i\|\bm\delta_i^{(1)}\|_2^4=O_p(s_n^{-2}).
\]
Bernstein's inequality and a union bound also give
\[
 \|\bm\delta^{(1)}\|_{\blinf}
 =O_p\left(\sqrt{\log n/s_n}+K_n\log n/s_n\right)
 =O_p(\sqrt{\log n/s_n}),
\]
since $K_n^2\log n/s_n\to0$.

To control $\bm\delta-\bm\delta^{(1)}$, define
\[
 \bm\Delta_F=\mathsf F_{0,*}^{-1/2}
 \{\mathsf F_A(\overline{\bm\nu},\bm\nu^*)-\mathsf F_{0,*}\}
 \mathsf F_{0,*}^{-1/2},\qquad
 \bm g_A^\circ=\mathsf F_{0,*}^{-1/2}\bm g_A.
\]
The fixed-point construction above gives
$\|\bm\delta\|_{\blinf}=O_p(\delta_{n,k})$.
Substituting this bound into \eqref{eq: population Hessian variation} yields
\[
 \|\mathsf F_0(\overline{\bm\nu},\bm\nu^*)-\mathsf F_{0,*}\|_2
 \le Ct_n\|\bm\delta\|_{\blinf}=O_p(t_n\delta_{n,k}).
\]
For the random difference, \eqref{eq: adjacency spectral bound} and \eqref{eq: exact Hessian noise} give
\[
 \|\mathsf F_A(\overline{\bm\nu},\bm\nu^*)
       -\mathsf F_0(\overline{\bm\nu},\bm\nu^*)\|_2
 \le\|\bm A-\bm P^*\|_2=O_p(\sqrt{t_n}).
\]
Combining these two bounds gives
$\|\mathsf F_A(\overline{\bm\nu},\bm\nu^*)-\mathsf F_{0,*}\|_2
=O_p(t_n\delta_{n,k}+\sqrt{t_n})$.
The curvature bound gives $\|\mathsf F_{0,*}^{-1/2}\|_2\le C/\sqrt{s_n}$, hence
\[
 \|\bm\Delta_F\|_2
 =O_p\{K_n\delta_{n,k}+\sqrt{K_n/s_n}\}
 =O_p(K_n\delta_{n,k})=o_p(1).
\]
Here \eqref{eq: hetero rates definitions} implies
$\delta_{n,k}\ge\sqrt{K_n\log n/s_n}$, and
$K_n\delta_{n,k}\to0$ follows from \eqref{eq: basic growth consequences}.
The covariance formula above gives
$\operatorname{Cov}_A(\bm g_A^\circ)=\bm\Pi_{\mathcal T_n}$, so
$\mathbb E_A\|\bm g_A^\circ\|_2^2
=\operatorname{tr}(\bm\Pi_{\mathcal T_n})\le n(r+1)$ and
$\|\bm g_A^\circ\|_2=O_p(\sqrt n)$.
Factoring the integrated Hessian in the score identity gives
\begin{align*}
 \bm\delta=\mathsf F_{0,*}^{-1/2}(\bm I+\bm\Delta_F)^{-1}\bm g_A^\circ,\quad
 \bm\delta-\bm\delta^{(1)}=\mathsf F_{0,*}^{-1/2}\{(\bm I+\bm\Delta_F)^{-1}-\bm I\}\bm g_A^\circ.
\end{align*}
On the event $\|\bm\Delta_F\|_2\le1/2$, whose probability tends to one,
$\|(\bm I+\bm\Delta_F)^{-1}-\bm I\|_2
=\|(\bm I+\bm\Delta_F)^{-1}\bm\Delta_F\|_2\le2\|\bm\Delta_F\|_2$.
Consequently,
\[
 \frac1n\|\bm\delta-\bm\delta^{(1)}\|_2^2
 \le\frac C{s_n}\|\bm\Delta_F\|_2^2
          \frac{\|\bm g_A^\circ\|_2^2}{n}
 =O_p(K_n^2\delta_{n,k}^2/s_n)=o_p(s_n^{-1}).
\]
The inequality $\|\bm\delta\|_2^2\le
2\|\bm\delta^{(1)}\|_2^2+2\|\bm\delta-\bm\delta^{(1)}\|_2^2$
therefore gives $n^{-1}\|\bm\delta\|_2^2=O_p(s_n^{-1})$.
Since $\delta_{n,k}\ge\sqrt{\log n/s_n}$, the maximum bounds for
$\bm\delta$ and $\bm\delta^{(1)}$ imply
$\|\bm\delta-\bm\delta^{(1)}\|_{\blinf}=O_p(\delta_{n,k})$.
Thus Assumption~\ref{assump: network growth} gives
\begin{align*}
 \frac1n\sum_i\|\bm\delta_i-\bm\delta_i^{(1)}\|_2^4
 \le\|\bm\delta-\bm\delta^{(1)}\|_{\blinf}^2
        \frac1n\|\bm\delta-\bm\delta^{(1)}\|_2^2
 =O_p(K_n^2\delta_{n,k}^4/s_n)=o_p(s_n^{-2}).
\end{align*}
Finally,
\[
 \frac1n\sum_i\|\bm\delta_i\|_2^4
 \le\frac8n\sum_i\|\bm\delta_i^{(1)}\|_2^4
    +\frac8n\sum_i\|\bm\delta_i-\bm\delta_i^{(1)}\|_2^4
 =O_p(s_n^{-2}),
\]
which completes all three bounds in \eqref{eq: exact local root max rate}.
\hfill$\blacksquare$

\begin{lemma}\label{lem: latent matrix calculus}
Let $\bm Z=\bm Z_n$ and $\bm E=\bm E_n$ be sequences of real
$n\times r$ matrices, where $r\ge1$ is fixed. Suppose
$c\sqrt n\le\sigma_r(\bm Z)\le\|\bm Z\|_2\le C\sqrt n$
and $\|\bm E\|_F=o(\sqrt n)$. Then, for all sufficiently large $n$,
\begin{align}
\label{eq: latent matrix perturbations}
 \|\bm M_{\bm Z+\bm E}-\bm M_{\bm Z}\|_F
 \le C\|\bm E\|_F/\sqrt n.
\end{align}
\end{lemma}
\noindent\textbf{Proof of Lemma~\ref{lem: latent matrix calculus}.}
Set $\bm Z_v=\bm Z+v\bm E$ and $\bm G_v=\bm Z_v^\top\bm Z_v$
for $v\in[0,1]$. Weyl's inequality and $\|\bm E\|_2\le\|\bm E\|_F=o(\sqrt n)$ give
\[
 \sigma_r(\bm Z_v)\ge\sigma_r(\bm Z)-\|\bm E\|_F\ge(c/2)\sqrt n
\]
uniformly in $v$ for sufficiently large $n$. Thus $\bm Z_v$ has full
column rank and $\bm G_v$ is invertible throughout the segment. In particular,
$\bm Z_v^\dagger=\bm G_v^{-1}\bm Z_v^\top$ and $\bm M_{\bm Z_v}=\bm I_n-\bm Z_v\bm Z_v^\dagger$. We differentiate these matrix functions with respect to the scalar parameter $v$, holding $\bm Z$ and $\bm E$ fixed. The product and inverse differentiation rules give
\[
 \frac{d\bm G_v}{dv}=\bm E^\top\bm Z_v+\bm Z_v^\top\bm E,\quad
 \frac{d\bm G_v^{-1}}{dv}
 =-\bm G_v^{-1}(\bm E^\top\bm Z_v+\bm Z_v^\top\bm E)\bm G_v^{-1}.
\]
Substituting into the derivative of $\bm Z_v^\dagger=\bm G_v^{-1}\bm Z_v^\top$ gives
\begin{align*}
 \frac{d\bm Z_v^\dagger}{dv}
 &=\frac{d\bm G_v^{-1}}{dv}\bm Z_v^\top+\bm G_v^{-1}\bm E^\top=-\bm G_v^{-1}(\bm E^\top\bm Z_v+\bm Z_v^\top\bm E)
       \bm Z_v^\dagger+\bm G_v^{-1}\bm E^\top\\
 &=-\bm Z_v^\dagger\bm E\bm Z_v^\dagger
       +\bm G_v^{-1}\bm E^\top(\bm I_n-\bm Z_v\bm Z_v^\dagger)=-\bm Z_v^\dagger\bm E\bm Z_v^\dagger
       +\bm G_v^{-1}\bm E^\top\bm M_{\bm Z_v}.
\end{align*}
Consequently, using $\bm Z_v\bm G_v^{-1}=(\bm Z_v^\dagger)^\top$,
\begin{align*}
 \frac{d\bm M_{\bm Z_v}}{dv}
 &=-\bm E\bm Z_v^\dagger-\bm Z_v\frac{d\bm Z_v^\dagger}{dv}=-\bm E\bm Z_v^\dagger+\bm Z_v\bm Z_v^\dagger\bm E\bm Z_v^\dagger
       -\bm Z_v\bm G_v^{-1}\bm E^\top\bm M_{\bm Z_v}\\
 &=-\bm M_{\bm Z_v}\bm E\bm Z_v^\dagger
       -(\bm Z_v^\dagger)^\top\bm E^\top\bm M_{\bm Z_v}.
\end{align*}
The singular-value bound and the projection property imply, uniformly in $v$, $\|\bm Z_v^\dagger\|_2\le Cn^{-1/2}$, $\|\bm M_{\bm Z_v}\|_2\le1$.
Using $\|\bm A\bm B\bm C\|_F\le\|\bm A\|_2\|\bm B\|_F\|\bm C\|_2$
in the projector derivative formula, we obtain
\begin{align*}
 \left\|\frac{d\bm M_{\bm Z_v}}{dv}\right\|_F
 &\le2\|\bm M_{\bm Z_v}\|_2\|\bm E\|_F\|\bm Z_v^\dagger\|_2
 \le C\|\bm E\|_F/\sqrt n.
\end{align*}
Finally, integrating from $v=0$ to $v=1$ gives
\begin{align*}
 \|\bm M_{\bm Z+\bm E}-\bm M_{\bm Z}\|_F
 &\le\int_0^1\left\|\frac{d\bm M_{\bm Z_v}}{dv}\right\|_F\,dv
 \le C\|\bm E\|_F/\sqrt n,
\end{align*}
which proves \eqref{eq: latent matrix perturbations}.
\hfill$\blacksquare$

\begin{lemma}\label{lem: response score representation}
Suppose Assumptions~\ref{assump: network}--\ref{assump: response errors} hold. Then
$\mathbb E_\epsilon\bm q_y(\bm\zeta^*;\bm\nu^*)=\bm0$.
With $\bm C_n$ defined in Section~\ref{sec: theory},
\begin{align}
\label{eq: reduced score representation}
 \sqrt n\bm q_y(\bm\zeta^*;\bm\nu^*)
 =\frac1{\sqrt n}\begin{pmatrix}
 \tau^*\bm C_n^\top\bm\epsilon\\
 -\tfrac12\sum_i(\epsilon_i^2-\sigma^{*2})
 \end{pmatrix}+\bm r_{q,n}, \quad
 \mathbb E_\epsilon\|\bm r_{q,n}\|_2^2\le CK_n/n.
\end{align}
The entries of $\bm C_n$ are bounded. For any deterministic $\bm u\in\mathbb R^n$ and symmetric
$\bm B\in\mathbb R^{n\times n}$,
\begin{align}
\label{eq: linear quadratic variance bound}
 \operatorname{Var}_\epsilon(\bm u^\top\bm\epsilon+
 \bm\epsilon^\top\bm B\bm\epsilon)
 \le C(\|\bm u\|_2^2+\|\bm B\|_F^2).
\end{align}
\end{lemma}

\noindent\textbf{Proof of Lemma~\ref{lem: response score representation}.}
We first prove bounds for any nonnegative row-stochastic matrix $\bm T$
with $\max_{i,j}T_{ij}\le CK_n/n$. If
$\max_{k,\ell}(\bm T^j)_{k\ell}\le CK_n/n$, then $(\bm T^{j+1})_{i\ell} =\sum_k T_{ik}(\bm T^j)_{k\ell} \le\frac{CK_n}{n}\sum_kT_{ik}=\frac{CK_n}{n}$.
Starting from $j=1$, induction gives the same entry bound for every positive
power of $\bm T$. Therefore,
\begin{equation}
\begin{aligned}
\label{eq: stochastic powers resolvents}
 \|\bm T^j\|_2&\le\|\bm T^j\|_F\le C\sqrt{K_n},\quad j\ge1,\\
 \sup_{\gamma\in\Gamma_\gamma}\|\bm T(\bm I_n-\gamma\bm T)^{-1}\|_F
 &\le C\sqrt{K_n},\quad
 \sup_{\gamma\in\Gamma_\gamma}\|(\bm I_n-\gamma\bm T)^{-1}\|_2
 \le C\sqrt{K_n},\\
 \sup_{\gamma\in\Gamma_\gamma}\|(\bm I_n-\gamma\bm T)^{-1}\|_{\infty}
 &\le C,\quad
 \sup_{\gamma\in\Gamma_\gamma}\|(\bm I_n-\gamma\bm T)^{-1}\|_F
 \le\sqrt n+C\sqrt{K_n}.
\end{aligned}
\end{equation}
For the first line, use
$\sum_{i,\ell}\{(\bm T^j)_{i\ell}\}^2\le CK_n$.
For the remaining bounds, $\|\bm T^j\|_\infty=1$ and
$|\gamma|\le1-\delta_\gamma$ give the convergent geometric expansions
\[
 (\bm I_n-\gamma\bm T)^{-1}
 =\bm I_n+\sum_{j=1}^\infty\gamma^j\bm T^j,\quad
 \bm T(\bm I_n-\gamma\bm T)^{-1}
 =\sum_{j=0}^\infty\gamma^j\bm T^{j+1}.
\]
The triangle inequality and the first line of
\eqref{eq: stochastic powers resolvents} yield
\begin{align*}
 \|\bm T(\bm I_n-\gamma\bm T)^{-1}\|_F
 &\le C\sqrt{K_n}\sum_{j=0}^\infty|\gamma|^j\le C\sqrt{K_n},\\
 \|(\bm I_n-\gamma\bm T)^{-1}-\bm I_n\|_F
 &\le C\sqrt{K_n}\sum_{j=1}^\infty|\gamma|^j\le C\sqrt{K_n}.
\end{align*}
The second bound gives the stated spectral and Frobenius bounds after
adding $\|\bm I_n\|_2=1$ and $\|\bm I_n\|_F=\sqrt n$, respectively,
using $K_n\ge1$. Finally,
$\|(\bm I_n-\gamma\bm T)^{-1}\|_\infty
\le\sum_{j=0}^\infty|\gamma|^j\le\delta_\gamma^{-1}\le C$.

At the true parameters, $\bm r_y(\bm\zeta^*;\bm\nu^*)=\bm\epsilon$, $\bm y=\bm\mu_y^*+\bm S^{*-1}\bm\epsilon$, $\bm W^*\bm y=\bm W^*\bm\mu_y^*+\bm Q\bm\epsilon$.
Substituting these identities into \eqref{eq: full response score} and using
$\tau^*\sigma^{*2}=1$ gives
\[
 \sqrt n\bm q_y(\bm\zeta^*;\bm\nu^*)
 =\frac1{\sqrt n}\begin{pmatrix}
 \tau^*\bm X^\top\bm\epsilon\\
 \tau^*\bm Z^{*\top}\bm\epsilon\\
 \tau^*(\bm W^*\bm\mu_y^*)^\top\bm\epsilon
 +\tau^*\{\bm\epsilon^\top\bm Q\bm\epsilon-\sigma^{*2}\tr\bm Q\}\\
 -\tfrac12\sum_i(\epsilon_i^2-\sigma^{*2})
 \end{pmatrix}.
\]
The linear terms in the first three coordinates form
$\tau^*\bm C_n^\top\bm\epsilon$. Thus the decomposition in
\eqref{eq: reduced score representation} holds with
\[
 \bm r_{q,n}
 =\frac{\tau^*}{\sqrt n}
 \{\bm\epsilon^\top\bm Q\bm\epsilon-\sigma^{*2}\tr\bm Q\}
 \bm e_\gamma,
\]
where $\bm e_\gamma\in\mathbb R^{p+r+2}$ is the standard basis vector
corresponding to $\gamma$ in $\bm\zeta$.
Since $\mathbb E_\epsilon\bm\epsilon=\bm0$ and
$\mathbb E_\epsilon(\bm\epsilon\bm\epsilon^\top)=\sigma^{*2}\bm I_n$,
\[
 \mathbb E_\epsilon(\bm\epsilon^\top\bm Q\bm\epsilon)
 =\tr\{\bm Q\,\mathbb E_\epsilon(\bm\epsilon\bm\epsilon^\top)\}
 =\sigma^{*2}\tr\bm Q,\quad
 \mathbb E_\epsilon(\epsilon_i^2-\sigma^{*2})=0.
\]
Every coordinate in the displayed score therefore has expectation zero,
so $\mathbb E_\epsilon\bm q_y(\bm\zeta^*;\bm\nu^*)=\bm0$.

To bound the entries of $\bm C_n$, use
$\bm\mu_y^*=\bm S^{*-1}(\bm X\bm\beta^*+\bm Z^*\bm\eta^*)$.
The inverse infinity-norm bound in \eqref{eq: stochastic powers resolvents}
and row stochasticity of $\bm W^*$ give
\begin{align*}
 \|\bm\mu_y^*\|_\infty
 &\le\delta_\gamma^{-1}
 \bigl(\|\bm X\|_{2,\infty}\|\bm\beta^*\|_2
       +\|\bm Z^*\|_{2,\infty}\|\bm\eta^*\|_2\bigr)\le C,\\
 \max_i\|(\bm C_n)_{i\cdot}\|_2^2
 &\le\|\bm X\|_{2,\infty}^2+\|\bm Z^*\|_{2,\infty}^2
       +\|\bm W^*\|_\infty^2\|\bm\mu_y^*\|_\infty^2\le C,
\end{align*}
where $\|\bm W^*\|_\infty=1$ and the design rows, latent rows and true
response coefficients are bounded by the assumptions.

For symmetric $\bm B$, decompose its centered quadratic form as
\[
 \bm\epsilon^\top\bm B\bm\epsilon-\sigma^{*2}\tr(\bm B)
 =\sum_i B_{ii}(\epsilon_i^2-\sigma^{*2})
 +2\sum_{i<j}B_{ij}\epsilon_i\epsilon_j.
\]
Independence and zero means make distinct off-diagonal terms uncorrelated and
make every off-diagonal term uncorrelated with the diagonal terms and with any
linear form. For symmetric $\bm B,\bm D$ and $\bm u\in\mathbb R^n$, it follows that
\begin{align}
\label{eq: linear quadratic covariance identities}
 \operatorname{Cov}_\epsilon(\bm u^\top\bm\epsilon,
 \bm\epsilon^\top\bm B\bm\epsilon)
 &=\mu_3\sum_i u_iB_{ii},\notag\\
 \operatorname{Cov}_\epsilon(\bm\epsilon^\top\bm B\bm\epsilon,
 \bm\epsilon^\top\bm D\bm\epsilon)
 &=(\mu_4-\sigma^{*4})\sum_i B_{ii}D_{ii}
 +4\sigma^{*4}\sum_{i<j}B_{ij}D_{ij}\notag\\
 &=2\sigma^{*4}\tr(\bm B\bm D)
 +(\mu_4-3\sigma^{*4})\sum_i B_{ii}D_{ii}.
\end{align}
Combining these identities with
$\operatorname{Cov}_\epsilon(\bm u^\top\bm\epsilon,\bm v^\top\bm\epsilon)
=\sigma^{*2}\bm u^\top\bm v$ gives the covariance of any linear--quadratic vector.
Assumption~\ref{assump: response} gives $\tau^*\ge c>0$, hence
$\sigma^{*2}\le c^{-1}$. Using $x^2\le2e^x$ for $x\ge0$ and the
exponential-moment bound in Assumption~\ref{assump: response errors}, followed
by Cauchy--Schwarz, gives
\begin{align*}
 \mu_4
 &\le2C_\epsilon^4\sigma^{*4}
 \mathbb E_\epsilon\exp\{\epsilon_i^2/(C_\epsilon^2\sigma^{*2})\}
 \le4C_\epsilon^4\sigma^{*4}\le C,\\
 |\mu_3|
 &\le\mathbb E_\epsilon|\epsilon_i|^3
 \le(\mathbb E_\epsilon\epsilon_i^2)^{1/2}
     (\mathbb E_\epsilon\epsilon_i^4)^{1/2}
 =\sigma^*\sqrt{\mu_4}\le C.
\end{align*} Since $\sum_i B_{ii}^2\le\|\bm B\|_F^2$ and
$|\sum_i u_iB_{ii}|\le\|\bm u\|_2\|\bm B\|_F$, the same identities prove
\eqref{eq: linear quadratic variance bound}. Applying \eqref{eq: linear quadratic variance bound} to the symmetric part of $\bm Q$ and using \eqref{eq: stochastic powers resolvents} proves the remainder bound in \eqref{eq: reduced score representation}.
\hfill$\blacksquare$

\begin{lemma}\label{lem: fitted score covariance}
Fix the observed data and the fitted parameters
$(\widehat{\bm\zeta},\widehat{\bm P},\widehat{\bm Z})$.
Use the conditional error vector $\bm e$, its moments
$\widehat\sigma^2,\widehat\mu_3,\widehat\mu_4$, and the auxiliary response
$\bm y^{(\bm e)}$ and score $\widehat{\bm q}_y^{(\bm e)}$ defined in Section I.
For each parameter coordinate $a$, its constant, linear and quadratic coefficients in $n\widehat q_{y,a}^{(\bm e)} =\widehat c_{n,a}+\widehat{\bm\ell}_{n,a}^{\top}\bm e +\bm e^\top\widehat{\bm B}_{n,a}\bm e$
are as follows, with $j\in[p]$ and $\ell\in[r]$:
\begin{equation}
\label{eq: fitted response score coefficients}
\begin{aligned}
& \widehat c_{n,\beta_j}=0,\quad
\widehat c_{n,\eta_\ell}=0,\quad
\widehat c_{n,\tau}=\frac{n}{2\widehat\tau},\quad
\widehat c_{n,\gamma}=-\tr(\widehat{\bm W}\widehat{\bm S}^{-1}),\\
 &\widehat{\bm\ell}_{n,\beta_j}=\widehat\tau\bm X_{\cdot j},\quad
 \widehat{\bm\ell}_{n,\eta_\ell}=\widehat\tau\widehat{\bm Z}_{\cdot\ell},\quad
 \widehat{\bm\ell}_{n,\tau}=\bm0,\quad
 \widehat{\bm\ell}_{n,\gamma}=\widehat\tau\widehat{\bm W}\widehat{\bm\mu}_y,\\
 &\widehat{\bm B}_{n,\beta_j}=\bm0,\quad
 \widehat{\bm B}_{n,\eta_\ell}=\bm0,\quad
 \widehat{\bm B}_{n,\tau}=-\bm I_n/2,\quad
 \widehat{\bm B}_{n,\gamma}=\frac{\widehat\tau}{2}
 \{\widehat{\bm S}^{-\top}\widehat{\bm W}^\top
 +\widehat{\bm W}\widehat{\bm S}^{-1}\}.
\end{aligned}
\end{equation}
The matrix $\widehat{\bm\Xi}_n=n\operatorname{Var}_{\bm e}
(\widehat{\bm q}_y^{(\bm e)})$ therefore has the explicit entries
\begin{equation}\label{eq: feasible score covariance}
\begin{aligned}
(\widehat{\bm\Xi}_n)_{ab}=\frac1n\Big[&
 \widehat\sigma^2\widehat{\bm\ell}_{n,a}^\top\widehat{\bm\ell}_{n,b}
 +2\widehat\sigma^4\tr(\widehat{\bm B}_{n,a}\widehat{\bm B}_{n,b})+\widehat\mu_3\sum_i\{\widehat\ell_{n,a,i}(\widehat{\bm B}_{n,b})_{ii}
 +\widehat\ell_{n,b,i}(\widehat{\bm B}_{n,a})_{ii}\}\\
 &+(\widehat\mu_4-3\widehat\sigma^4)
 \sum_i(\widehat{\bm B}_{n,a})_{ii}(\widehat{\bm B}_{n,b})_{ii}\Big].
\end{aligned}
\end{equation}
\end{lemma}
\noindent\textbf{Proof of Lemma~\ref{lem: fitted score covariance}.}
Since $\widehat{\bm S}\widehat{\bm\mu}_y
=\bm X\widehat{\bm\beta}+\widehat{\bm Z}\widehat{\bm\eta}$, we have $\bm r_{y^{(\bm e)}}(\widehat{\bm\zeta};\widehat{\bm\nu})=\bm e$ and $\widehat{\bm W}\bm y^{(\bm e)} =\widehat{\bm W}\widehat{\bm\mu}_y +\widehat{\bm W}\widehat{\bm S}^{-1}\bm e$. Substituting into \eqref{eq: full response score} gives
\[
 n\widehat{\bm q}_y^{(\bm e)}
 =\begin{pmatrix}
 \widehat\tau\bm X^\top\bm e\\
 \widehat\tau\widehat{\bm Z}^\top\bm e\\
 -\tr(\widehat{\bm W}\widehat{\bm S}^{-1})
 +\widehat\tau(\widehat{\bm W}\widehat{\bm\mu}_y)^\top\bm e
 +\widehat\tau\bm e^\top\widehat{\bm W}\widehat{\bm S}^{-1}\bm e\\
 n/(2\widehat\tau)-\bm e^\top\bm e/2
 \end{pmatrix}.
\]
The symmetric part of the gamma quadratic coefficient gives
\eqref{eq: fitted response score coefficients}.
The constants $\widehat c_{n,a}$ do not contribute to covariance, so
\[
 (\widehat{\bm\Xi}_n)_{ab}
 =\frac1n\operatorname{Cov}_{\bm e}\left(
 \widehat{\bm\ell}_{n,a}^{\top}\bm e+\bm e^\top\widehat{\bm B}_{n,a}\bm e,
 \widehat{\bm\ell}_{n,b}^{\top}\bm e+\bm e^\top\widehat{\bm B}_{n,b}\bm e
 \right).
\]
The $e_i$ are independent with mean zero and variance $\widehat\sigma^2$.
Applying \eqref{eq: linear quadratic covariance identities} with
$(\sigma^{*2},\mu_3,\mu_4)$ replaced by
$(\widehat\sigma^2,\widehat\mu_3,\widehat\mu_4)$ and using
$\operatorname{Cov}_{\bm e}(\bm u^\top\bm e,\bm v^\top\bm e)
=\widehat\sigma^2\bm u^\top\bm v$ gives
\eqref{eq: feasible score covariance}.
\hfill$\blacksquare$

\begin{lemma}\label{lem: response operator bounds}
Suppose Assumptions~\ref{assump: network}--\ref{assump: response errors} hold.
Let $\bm W_t=(1-t)\bm W^*+t\widehat{\bm W}$ for $t\in[0,1]$.
With probability tending to one,
\begin{equation}\label{eq: response operator bounds}
\begin{aligned}
 \sup_{t\in[0,1]}\|\bm W_t\|_F+
 \sup_{\substack{t\in[0,1]\\\gamma\in\Gamma_\gamma}}
 \big\{\|(\bm I_n-\gamma\bm W_t)^{-1}\|_2
 +\|\bm W_t(\bm I_n-\gamma\bm W_t)^{-1}\|_F\big\}&\le C\sqrt{K_n},\\
 \sup_{\substack{t\in[0,1]\\\gamma\in\Gamma_\gamma}}
 \|(\bm I_n-\gamma\bm W_t)^{-1}\|_\infty&\le C.
\end{aligned}
\end{equation}
Moreover, $\|\bm\mu_y^*\|_\infty\le C$ and $n^{-1/2}\{\|\bm y\|_2+
\sup_{t\in[0,1]}\|\bm W_t\bm y\|_2\}=O_p(1)$.
\end{lemma}

\noindent\textbf{Proof of Lemma~\ref{lem: response operator bounds}.}
By \eqref{eq: envelope implications} and Theorem~\ref{thm: first stage},
with probability tending to one, both $\bm W^*$ and $\widehat{\bm W}$
are nonnegative and row stochastic, with entries at most $CK_n/n$.
On this event, for every $t\in[0,1]$,
$\bm W_t\bm1_n=\bm1_n$ and
$0\le (\bm W_t)_{ij}\le CK_n/n$.
Applying \eqref{eq: stochastic powers resolvents} to $\bm W_t$ gives
all the matrix bounds; its application to $\bm W^*$ is deterministic.
Also,
\[
 \|\bm\mu_y^*\|_\infty
 \le\|\bm S^{*-1}\|_\infty
 \big(\|\bm X\|_{2,\infty}\|\bm\beta^*\|_2
 +\|\bm Z^*\|_{2,\infty}\|\bm\eta^*\|_2\big)\le C.
\]
Since $\bm y=\bm\mu_y^*+\bm S^{*-1}\bm\epsilon$, we have
\[
 \mathbb E_\epsilon\|\bm S^{*-1}\bm\epsilon\|_2^2
 =\sigma^{*2}\|\bm S^{*-1}\|_F^2\le C(n+K_n)\le Cn,\quad
 \mathbb E_\epsilon\|\bm W^*\bm S^{*-1}\bm\epsilon\|_2^2\le CK_n.
\]
Since $\bm W^*$ is nonnegative and row stochastic,
\[
 \|\bm\mu_y^*\|_2+\|\bm W^*\bm\mu_y^*\|_2
 \le2\sqrt n\|\bm\mu_y^*\|_\infty\le C\sqrt n.
\]
Using $\mathbb E_\epsilon\bm\epsilon=\bm0$ in the decompositions of
$\bm y$ and $\bm W^*\bm y$ therefore gives
\[
 \mathbb E_\epsilon\big(\|\bm y\|_2^2+\|\bm W^*\bm y\|_2^2\big)
 =\|\bm\mu_y^*\|_2^2+\|\bm W^*\bm\mu_y^*\|_2^2
 +\sigma^{*2}\big(\|\bm S^{*-1}\|_F^2
 +\|\bm W^*\bm S^{*-1}\|_F^2\big)\le Cn.
\]
Markov's inequality now yields
$\|\bm y\|_2+\|\bm W^*\bm y\|_2=O_p(\sqrt n)$.
Theorem~\ref{thm: first stage} gives
$\|\widehat{\bm W}-\bm W^*\|_2=O_p(K_n/\sqrt{s_n})=o_p(1)$, so
\[
 \sup_{t\in[0,1]}\|\bm W_t\bm y\|_2
 \le\|\bm W^*\bm y\|_2+
 \|\widehat{\bm W}-\bm W^*\|_2\|\bm y\|_2=O_p(\sqrt n).
\]
This completes the proof of Lemma \ref{lem: response operator bounds}.
\hfill$\blacksquare$

\begin{lemma}\label{lem: response conditioning}
Suppose Assumptions~\ref{assump: network}--\ref{assump: response errors} hold.
Uniformly for large $n$, $c\bm I\preceq\mathsf J_{y,n}\preceq C\bm I$, $\bm0\preceq\bm\Xi_n\preceq C\bm I$.
Moreover,
\begin{align*}
 \bm V_n
 &=\begin{pmatrix}
 \sigma^{*2}\bm G_n^{-1}
 &-\tau^{*2}\mu_3\bm G_n^{-1}\bar{\bm c}_n\\
 -\tau^{*2}\mu_3\bar{\bm c}_n^\top\bm G_n^{-1}
 &\tau^{*4}(\mu_4-\sigma^{*4})
 \end{pmatrix}+o(1),
\end{align*}
where $\bm G_n=n^{-1}\bm C_n^\top\bm C_n$ and
$\bar{\bm c}_n=n^{-1}\bm C_n^\top\bm1_n$, with $\bm C_n$ defined in
Section~\ref{sec: theory}; the $o(1)$ term is measured in operator norm. Here $c\bm I\preceq\bm G_n\preceq C\bm I$.
Writing $\bm c_{n,i}^\top$ for row $i$ of $\bm C_n$, define
\[
 \bm U_n=\frac1{\sqrt n}\sum_i
 \begin{pmatrix}\bm G_n^{-1}\bm c_{n,i}\epsilon_i\\
 -\tau^{*2}(\epsilon_i^2-\sigma^{*2})\end{pmatrix}.
\]
The response score satisfies
\[
 \mathsf J_{y,n}^{-1}\sqrt n\bm q_y(\bm\zeta^*;\bm\nu^*)
 =\bm U_n+\bm r_n,\
 \mathbb E_\epsilon\|\bm r_n\|_2^2=o(1).
\]
If
\eqref{eq: response nondegeneracy} also holds, then
$c\bm I\preceq\bm V_n\preceq C\bm I$ uniformly for large $n$.

\end{lemma}

\noindent\textbf{Proof of Lemma~\ref{lem: response conditioning}.}
Put $\bm B_y=(\bm X,\bm Z^*,\bm W^*\bm y)$ and
$t_\gamma=n^{-1}\tr\{\bm W^*\bm S_\gamma(\bm P^*)^{-1}\}$.
Differentiating \eqref{eq: full response score} gives the symmetric matrix
\begin{align}
\label{eq: full response Hessian}
 -\nabla_{\bm\zeta}\bm q_y(\bm\zeta;\bm\nu^*)
 =\begin{pmatrix}
 \tau\bm B_y^\top\bm B_y/n+t_\gamma'\widetilde{\bm e}_\gamma\widetilde{\bm e}_\gamma^\top
 &-\bm B_y^\top\bm r_y(\bm\zeta;\bm\nu^*)/n\\
 -\bm r_y(\bm\zeta;\bm\nu^*)^\top\bm B_y/n&1/(2\tau^2)
 \end{pmatrix},
\end{align}
where $\widetilde{\bm e}_\gamma=(\bm0_{p+r}^\top,1)^\top\in\mathbb R^{p+r+1}$
is the coordinate vector for $\gamma$ in
$(\bm\beta^\top,\bm\eta^\top,\gamma)^\top$.
Since $\bm B_y=\bm C_n+\bm Q\bm\epsilon\widetilde{\bm e}_\gamma^\top$ at the truth, the expected negative Jacobian is
\begin{align}
\label{eq: full response information}
 \mathsf J_{y,n}=\mathsf J_{y,n}^{(0)}+\bm R_{J,n},\quad
 \mathsf J_{y,n}^{(0)}=\diag\{\tau^*\bm G_n,1/(2\tau^{*2})\}.
\end{align}
The only nonzero corrections are
\[
 R_{J,n;\gamma\gamma}=n^{-1}\{\tr(\bm Q^2)+\|\bm Q\|_F^2\},\quad
 R_{J,n;\gamma\tau}=R_{J,n;\tau\gamma}
 =-\sigma^{*2}\tr(\bm Q)/n.
\]
By Lemma~\ref{lem: response operator bounds}, the first correction is $O(K_n/n)$ because
$|\tr(\bm Q^2)|\le\|\bm Q\|_F^2\le CK_n$.
For the second, the power entry bound gives
$|\tr\bm Q|\le\sum_{j\ge0}|\gamma^*|^j|\tr((\bm W^*)^{j+1})|\le CK_n$ by \eqref{eq: stochastic powers resolvents}.
Thus $\|\bm R_{J,n}\|_2=O(K_n/n)=o(1)$.
Put $\bm u=\bm W^*\bm\mu_y^*$. Then
\[
 \bm G_n=\frac1n\begin{pmatrix}
 \bm X_+^{*\top}\bm X_+^*&\bm X_+^{*\top}\bm u\\
 \bm u^\top\bm X_+^*&\bm u^\top\bm u
 \end{pmatrix},\quad
 \frac{\bm u^\top\bm u}{n}
 -\frac{\bm u^\top\bm X_+^*}{n}
 \left(\frac{\bm X_+^{*\top}\bm X_+^*}{n}\right)^{-1}
 \frac{\bm X_+^{*\top}\bm u}{n}
 =\frac{\|\bm M_{\bm X_+^*}\bm u\|_2^2}{n}.
\]
Assumption~\ref{assump: response} gives
$\lambda_{\min}(\bm X_+^{*\top}\bm X_+^*/n)\ge c_X$, and
\eqref{eq: gamma identification} bounds the displayed Schur complement
below by $\kappa_\gamma/2$ for all sufficiently large $n$. Hence
\[
 \det(\bm G_n)
 =\det(\bm X_+^{*\top}\bm X_+^*/n)\frac{\|\bm M_{\bm X_+^*}\bm u\|_2^2}{n}
 \ge c_X^{p+r}\kappa_\gamma/2>0.
\]
Since $\bm W^*$ is nonnegative and row stochastic,
$|u_i|\le\|\bm\mu_y^*\|_\infty\le C$ by
Lemma~\ref{lem: response operator bounds}. The bounded rows of $\bm X$ and
$\bm Z^*$ therefore give
\[
 \lambda_{\max}(\bm G_n)\le\tr(\bm G_n)
 =\frac1n\sum_i\big(\|\bm x_i\|_2^2+\|\bm z_i^*\|_2^2+u_i^2\big)\le C,
 \quad
 \lambda_{\min}(\bm G_n)
 \ge\frac{\det(\bm G_n)}{\lambda_{\max}(\bm G_n)^{p+r}}\ge c.
\]
Here $\bm G_n$ is a Gram matrix of fixed dimension $p+r+1$.
By Assumption~\ref{assump: response}, $\tau^*$ is bounded above and away from zero.
Thus \eqref{eq: full response information} gives
$c\bm I\preceq\mathsf J_{y,n}^{(0)}\preceq C\bm I$.
Since $\|\bm R_{J,n}\|_2=o(1)$,
\[
 \lambda_{\min}(\mathsf J_{y,n})
 \ge\lambda_{\min}(\mathsf J_{y,n}^{(0)})-\|\bm R_{J,n}\|_2\ge c/2,
 \qquad
 \lambda_{\max}(\mathsf J_{y,n})
 \le\lambda_{\max}(\mathsf J_{y,n}^{(0)})+\|\bm R_{J,n}\|_2\le2C
\]
for all sufficiently large $n$. Finally, Lemma~\ref{lem: response score representation} gives $\mathbb E_\epsilon\bm q_y(\bm\zeta^*;\bm\nu^*)=\bm0$. Using \eqref{eq: reduced score representation}, independence of the errors, and $\tau^*\sigma^{*2}=1$, we bound $\tr(\bm\Xi_n)=\mathbb E_\epsilon\|\sqrt n\bm q_y(\bm\zeta^*;\bm\nu^*)\|_2^2$ by
\begin{align*}
 \tr(\bm\Xi_n)&\le\frac{2\tau^{*2}}n\mathbb E_\epsilon\|\bm C_n^\top\bm\epsilon\|_2^2
 +\frac1{2n}\mathbb E_\epsilon\left\{\sum_i(\epsilon_i^2-\sigma^{*2})\right\}^2
 +2\mathbb E_\epsilon\|\bm r_{q,n}\|_2^2\\
 &\le2\tau^*\tr(\bm G_n)+\frac{\mu_4-\sigma^{*4}}2+\frac{CK_n}n\le C.
\end{align*}
The last bound uses $\mu_4\le C$, as established in the proof of
Lemma~\ref{lem: response score representation}.
Since $\bm\Xi_n$ is a covariance matrix,
$\bm0\preceq\bm\Xi_n\preceq\tr(\bm\Xi_n)\bm I\preceq C\bm I$.

Since $(\mathsf J_{y,n}^{(0)})^{-1}
=\diag\{\tau^{*-1}\bm G_n^{-1},2\tau^{*2}\}$,
multiplication of \eqref{eq: reduced score representation} gives
\[
 (\mathsf J_{y,n}^{(0)})^{-1}\sqrt n\bm q_y(\bm\zeta^*;\bm\nu^*)
 =\bm U_n+(\mathsf J_{y,n}^{(0)})^{-1}\bm r_{q,n},\quad
 \bm U_n=\frac1{\sqrt n}\sum_i
 \begin{pmatrix}\bm G_n^{-1}\bm c_{n,i}\epsilon_i\\
 -\tau^{*2}(\epsilon_i^2-\sigma^{*2})\end{pmatrix},
\]
where $\bm c_{n,i}^\top$ is row $i$ of $\bm C_n$.
To replace $(\mathsf J_{y,n}^{(0)})^{-1}$ by $\mathsf J_{y,n}^{-1}$, use
\[
 \mathsf J_{y,n}^{-1}-(\mathsf J_{y,n}^{(0)})^{-1}
 =-\mathsf J_{y,n}^{-1}\bm R_{J,n}(\mathsf J_{y,n}^{(0)})^{-1},\quad
 \|\mathsf J_{y,n}^{-1}-(\mathsf J_{y,n}^{(0)})^{-1}\|_2\le CK_n/n.
\]
Thus, defining
\[
 \bm r_n=(\mathsf J_{y,n}^{(0)})^{-1}\bm r_{q,n}
 +\{\mathsf J_{y,n}^{-1}-(\mathsf J_{y,n}^{(0)})^{-1}\}
 \sqrt n\bm q_y(\bm\zeta^*;\bm\nu^*),
\]
we obtain
\begin{equation}\label{eq: response influence representation}
 \mathsf J_{y,n}^{-1}\sqrt n\bm q_y(\bm\zeta^*;\bm\nu^*)
 =\bm U_n+\bm r_n,\quad
 \mathbb E_\epsilon\|\bm r_n\|_2^2=o(1).
\end{equation}
Indeed, \eqref{eq: reduced score representation} and the bound
$\mathbb E_\epsilon\|\sqrt n\bm q_y(\bm\zeta^*;\bm\nu^*)\|_2^2
=\tr(\bm\Xi_n)\le C$ give
\begin{align*}
 \mathbb E_\epsilon\|\bm r_n\|_2^2
 &\le2\|(\mathsf J_{y,n}^{(0)})^{-1}\|_2^2
 \mathbb E_\epsilon\|\bm r_{q,n}\|_2^2
 +2\|\mathsf J_{y,n}^{-1}-(\mathsf J_{y,n}^{(0)})^{-1}\|_2^2
 \mathbb E_\epsilon\|\sqrt n\bm q_y(\bm\zeta^*;\bm\nu^*)\|_2^2\\
 &\le C\{K_n/n+(K_n/n)^2\}=o(1).
\end{align*}

Write $\bm U_n=(\bm U_{n,1}^\top,U_{n,2})^\top$.
The summands defining $\bm U_n$ are independent and centered.
Since $n^{-1}\sum_i\bm c_{n,i}\bm c_{n,i}^\top=\bm G_n$ and
$n^{-1}\sum_i\bm c_{n,i}=\bar{\bm c}_n$, their covariance blocks are
\begin{align*}
 \operatorname{Var}_\epsilon(\bm U_{n,1})
 &=\frac{\sigma^{*2}}n\bm G_n^{-1}
 \left(\sum_i\bm c_{n,i}\bm c_{n,i}^\top\right)\bm G_n^{-1}
 =\sigma^{*2}\bm G_n^{-1},\\
 \operatorname{Var}_\epsilon(U_{n,2})
 &=\frac{\tau^{*4}}n\sum_i\mathbb E_\epsilon(\epsilon_i^2-\sigma^{*2})^2
 =\tau^{*4}(\mu_4-\sigma^{*4}),\\
 \operatorname{Cov}_\epsilon(\bm U_{n,1},U_{n,2})
 &=-\frac{\tau^{*2}}n\sum_i\bm G_n^{-1}\bm c_{n,i}
 \mathbb E_\epsilon\{\epsilon_i(\epsilon_i^2-\sigma^{*2})\}
 =-\tau^{*2}\mu_3\bm G_n^{-1}\bar{\bm c}_n.
\end{align*}
These bounds give $\mathbb E_\epsilon\|\bm U_n\|_2^2\le C$.
Both $\bm U_n$ and $\bm r_n$ are centered, so
\eqref{eq: response influence representation} implies
\begin{align*}
 \bm V_n-\operatorname{Var}_\epsilon(\bm U_n)
 &=\mathbb E_\epsilon(\bm U_n\bm r_n^\top+
 \bm r_n\bm U_n^\top+\bm r_n\bm r_n^\top),\\
 \|\bm V_n-\operatorname{Var}_\epsilon(\bm U_n)\|_2
 &\le2\{\mathbb E_\epsilon\|\bm U_n\|_2^2
 \mathbb E_\epsilon\|\bm r_n\|_2^2\}^{1/2}
 +\mathbb E_\epsilon\|\bm r_n\|_2^2=o(1).
\end{align*}
This proves the stated formula for $\bm V_n$.

Let $\bm V_n^{(0)}=\operatorname{Var}_\epsilon(\bm U_n)$.
The Schur complement of its upper-left block $\sigma^{*2}\bm G_n^{-1}$ is
\[
 \tau^{*4}\left\{\mu_4-\sigma^{*4}
 -\frac{\mu_3^2}{\sigma^{*2}}
 \bar{\bm c}_n^\top\bm G_n^{-1}\bar{\bm c}_n\right\}
 =\tau^{*4}\left\{\mu_4-\sigma^{*4}
 -\frac{\mu_3^2}{\sigma^{*2}}h_n\right\},
\]
because
$\bar{\bm c}_n^\top\bm G_n^{-1}\bar{\bm c}_n
=n^{-1}\bm1_n^\top(\bm I_n-\bm M_{\bm C_n})\bm1_n=h_n$.
Condition~\eqref{eq: response nondegeneracy} and the positive lower bound on
$\tau^*$ bound this Schur complement below by a positive constant for all
sufficiently large $n$. Since
$c\bm I\preceq\sigma^{*2}\bm G_n^{-1}\preceq C\bm I$, the block determinant formula gives
\[
 \det(\bm V_n^{(0)})
 =\det(\sigma^{*2}\bm G_n^{-1})\,
 \tau^{*4}\left\{\mu_4-\sigma^{*4}
 -\frac{\mu_3^2}{\sigma^{*2}}h_n\right\}\ge c.
\]
Moreover, $\lambda_{\max}(\bm V_n^{(0)})
\le\tr(\bm V_n^{(0)})=\mathbb E_\epsilon\|\bm U_n\|_2^2\le C$.
As $\bm V_n^{(0)}$ is a covariance matrix of fixed dimension $p+r+2$, $\lambda_{\min}(\bm V_n^{(0)}) \ge\frac{\det(\bm V_n^{(0)})} {\lambda_{\max}(\bm V_n^{(0)})^{p+r+1}}\ge c$.
The bound $\|\bm V_n-\bm V_n^{(0)}\|_2=o(1)$ proved above then yields
\[
 \lambda_{\min}(\bm V_n)
 \ge\lambda_{\min}(\bm V_n^{(0)})-\|\bm V_n-\bm V_n^{(0)}\|_2\ge c/2,
 \quad
 \lambda_{\max}(\bm V_n)
 \le\lambda_{\max}(\bm V_n^{(0)})+\|\bm V_n-\bm V_n^{(0)}\|_2\le2C
\]
for all sufficiently large $n$.
This completes the proof.
\hfill$\blacksquare$

\begin{lemma}\label{lem: response profile uniform}
Suppose Assumptions~\ref{assump: network}--\ref{assump: response errors} hold.
There is a fixed $c>0$ such that the event
\[
 E_n=\left\{\operatorname{rank}(\bm X,\widehat{\bm Z})=p+r,\quad
 \inf_{\gamma\in\Gamma_\gamma} n^{-1}\operatorname{RSS}_n(\gamma;\widehat{\bm\nu})\ge c\right\}
\]
satisfies $\Pr(E_n)\to1$.
The output of Algorithm~\ref{alg:step2} satisfies
\begin{align}
\label{eq: supplementary profile estimator rate}
 \|\widehat{\bm\zeta}-\bm\zeta^*\|_2=O_p(K_n/\sqrt{s_n})=o_p(1).
\end{align}
Moreover,
\begin{align}
\label{eq: supplementary full grid score}
 \|\bm q_y(\widehat{\bm\zeta};\widehat{\bm\nu})\|_2=O_p(n^{-1}).
\end{align}
\end{lemma}

\noindent\textbf{Proof of Lemma~\ref{lem: response profile uniform}.} Put $\bm N=\bm M_{\bm X_+^*}$, $\bm v=\bm N\bm W^*\bm y$, $j_n=n^{-1}\|\bm N\bm W^*\bm\mu_y^*\|_2^2$. Condition~\eqref{eq: gamma identification} in Assumption~\ref{assump: response} gives $j_n\ge c$ for all sufficiently large $n$, while Lemma~\ref{lem: response operator bounds} gives $j_n\le n^{-1}\|\bm W^*\bm\mu_y^*\|_2^2\le\|\bm\mu_y^*\|_\infty^2\le C$.
 For $h=\gamma-\gamma^*$, the identity $\bm N\bm S_\gamma(\bm P^*)\bm y=\bm N\bm\epsilon-h\bm v$ gives the exact quadratic
\begin{align}
\label{eq: quadratic profile identity}
 n^{-1}\operatorname{RSS}_n(\gamma;\bm\nu^*)
 =a_n-2h b_n^{\rm R}+h^2c_n^{\rm R},\
 a_n=\bm\epsilon^\top\bm N\bm\epsilon/n,\
 b_n^{\rm R}=\bm\epsilon^\top\bm v/n,\
 c_n^{\rm R}=\|\bm v\|_2^2/n.
\end{align}
Write $\bm v=\bm v_0+\bm N\bm Q\bm\epsilon$, with
$\bm v_0=\bm N\bm W^*\bm\mu_y^*$.
Since $\|\bm v_0\|_2\le C\sqrt n$, $\|\bm Q\|_F^2\le CK_n$ and
$\operatorname{rank}(\bm I_n-\bm N)=p+r$, the variance bound
\eqref{eq: linear quadratic variance bound} gives
\begin{align}
\label{eq: quadratic profile coefficients}
 a_n-\sigma^{*2}=O_p(n^{-1/2}),\quad
 b_n^{\rm R}=O_p(\sqrt{K_n/n}),\quad
 c_n^{\rm R}-j_n=O_p(\sqrt{K_n/n}).
\end{align}
Here $\mathbb E a_n=\sigma^{*2}\{1-(p+r)/n\}$.
For the second bound, expand
$n b_n^{\rm R}=\bm\epsilon^\top\bm v_0+\bm\epsilon^\top\bm N\bm Q\bm\epsilon$;
the quadratic mean is $O(K_n)$ because
$|\tr(\bm N\bm Q)|\le|\tr\bm Q|+\sqrt{p+r}\|\bm Q\|_F\le CK_n$.
For the third, expand
$n(c_n^{\rm R}-j_n)=2\bm v_0^\top\bm N\bm Q\bm\epsilon
+\bm\epsilon^\top\bm Q^\top\bm N\bm Q\bm\epsilon$.
The linear standard deviation is at most $C\sqrt{nK_n}$, and the
quadratic mean and standard deviation are $O(K_n)$.
Since $K_n/n=o(1)$, these yield the displayed orders.

Let $\widehat{\bm N}=\bm M_{\widehat{\bm X}_+}$.
The latent estimation bound in Theorem~\ref{thm: first stage} and the Gram matrix lower bound in Assumption~\ref{assump: response} imply
\begin{align}
\label{eq: joint design perturbation}
 \lambda_{\min}(\widehat{\bm X}_+^\top\widehat{\bm X}_+/n)\ge c,\quad
 \|\widehat{\bm X}_+^{\dagger}\|_2\le Cn^{-1/2},\quad
 \|\widehat{\bm N}-\bm N\|_F=O_p(s_n^{-1/2})
\end{align}
with probability tending to one. Indeed the normalized Gram difference is
$O_p(s_n^{-1/2})$, so Weyl's inequality gives the first two bounds;
Lemma~\ref{lem: latent matrix calculus}, applied to $\bm X_+^*$, gives the third.
Set $\widehat{\bm u}=\widehat{\bm N}(\bm I_n-\gamma^*\widehat{\bm W})\bm y$
and $\widehat{\bm v}=\widehat{\bm N}\widehat{\bm W}\bm y$.
The fitted RSS is exactly $\|\widehat{\bm u}-h\widehat{\bm v}\|_2^2$.
Since $\bm N\bm S^*\bm y=\bm N\bm\epsilon$ and
$\bm v=\bm N\bm W^*\bm y$, subtraction gives
\begin{align*}
 \widehat{\bm u}-\bm N\bm\epsilon
 &=(\widehat{\bm N}-\bm N)\bm S^*\bm y
 -\gamma^*\widehat{\bm N}(\widehat{\bm W}-\bm W^*)\bm y,\\
 \widehat{\bm v}-\bm v
 &=(\widehat{\bm N}-\bm N)\bm W^*\bm y
 +\widehat{\bm N}(\widehat{\bm W}-\bm W^*)\bm y.
\end{align*}
Lemma~\ref{lem: response operator bounds} gives
$\|\bm y\|_2+\|\bm W^*\bm y\|_2=O_p(\sqrt n)$, and hence
$\|\bm S^*\bm y\|_2\le\|\bm y\|_2+|\gamma^*|\|\bm W^*\bm y\|_2
=O_p(\sqrt n)$.
Using $\|\widehat{\bm N}\|_2\le1$, we obtain
\begin{align*}
 \frac{\|\widehat{\bm u}-\bm N\bm\epsilon\|_2+
 \|\widehat{\bm v}-\bm v\|_2}{\sqrt n}
 &\le\|\widehat{\bm N}-\bm N\|_F
 \frac{\|\bm S^*\bm y\|_2+\|\bm W^*\bm y\|_2}{\sqrt n}
 +(1+|\gamma^*|)\|\widehat{\bm W}-\bm W^*\|_F
 \frac{\|\bm y\|_2}{\sqrt n}\\
 &=O_p(s_n^{-1/2})+O_p(K_n/\sqrt{s_n})
 =O_p(K_n/\sqrt{s_n}),
\end{align*}
where the last line uses \eqref{eq: joint design perturbation},
$\|\widehat{\bm W}-\bm W^*\|_F=O_p(K_n/\sqrt{s_n})$
from Theorem~\ref{thm: first stage}, and $K_n\ge1$.
Put $e_n=K_n/\sqrt{s_n}=o(1)$.
We now bound the errors in the three coefficients of
$n^{-1}\|\widehat{\bm u}-h\widehat{\bm v}\|_2^2$.
By \eqref{eq: quadratic profile coefficients} and Lemma~\ref{lem: response operator bounds}, respectively,
\[
 \frac{\|\bm N\bm\epsilon\|_2^2}{n}
 =a_n=\sigma^{*2}+O_p(n^{-1/2})=O_p(1),\quad
 \|\bm v\|_2=\|\bm N\bm W^*\bm y\|_2
 \le\|\bm W^*\bm y\|_2=O_p(\sqrt n).
\]
The triangle inequality and the preceding difference bounds therefore give
\begin{align*}
 \|\widehat{\bm u}\|_2+\|\widehat{\bm v}\|_2
 \le\|\bm N\bm\epsilon\|_2+\|\bm v\|_2
 +\|\widehat{\bm u}-\bm N\bm\epsilon\|_2+\|\widehat{\bm v}-\bm v\|_2
 =O_p(\sqrt n)+O_p(\sqrt n\,e_n)=O_p(\sqrt n),
\end{align*}
since $e_n=o(1)$.
Using the definitions in \eqref{eq: quadratic profile identity}, we obtain
\begin{align*}
 \left|\frac{\|\widehat{\bm u}\|_2^2}{n}-a_n\right|
 &\le\frac{\|\widehat{\bm u}-\bm N\bm\epsilon\|_2
 (\|\widehat{\bm u}\|_2+\|\bm N\bm\epsilon\|_2)}n=O_p(e_n),\\
 \left|\frac{\widehat{\bm u}^\top\widehat{\bm v}}n-b_n^{\rm R}\right|
 &\le\frac{\|\widehat{\bm u}-\bm N\bm\epsilon\|_2\|\widehat{\bm v}\|_2
 +\|\bm N\bm\epsilon\|_2\|\widehat{\bm v}-\bm v\|_2}n=O_p(e_n),\\
 \left|\frac{\|\widehat{\bm v}\|_2^2}n-c_n^{\rm R}\right|
 &\le\frac{\|\widehat{\bm v}-\bm v\|_2
 (\|\widehat{\bm v}\|_2+\|\bm v\|_2)}n=O_p(e_n).
\end{align*}
Since $s_n\le n/4$ and $K_n\ge1$, we have $\sqrt{K_n/n}\le Ce_n$.
Combining these bounds with \eqref{eq: quadratic profile coefficients} gives
\[
 \frac{\|\widehat{\bm u}\|_2^2}n-\sigma^{*2}=O_p(e_n),\quad
 \frac{\widehat{\bm u}^\top\widehat{\bm v}}n=O_p(e_n),\quad
 \frac{\|\widehat{\bm v}\|_2^2}n-j_n=O_p(e_n).
\]
Consequently, with $h=\gamma-\gamma^*$,
\[
 \frac{\operatorname{RSS}_n(\gamma;\widehat{\bm\nu})}n
 -\sigma^{*2}-j_nh^2
 =\left(\frac{\|\widehat{\bm u}\|_2^2}n-\sigma^{*2}\right)
 -2h\frac{\widehat{\bm u}^\top\widehat{\bm v}}n
 +h^2\left(\frac{\|\widehat{\bm v}\|_2^2}n-j_n\right).
\]
The right side is a quadratic polynomial in $h$ with all three coefficients
$O_p(e_n)$. Since $|h|\le2$ on $\Gamma_\gamma$ and $dh/d\gamma=1$,
its value and its first two derivatives are uniformly $O_p(e_n)$, proving
\begin{align}
\label{eq: quadratic fitted profile control}
 \max_{0\le j\le2}\sup_\gamma
 \left|\partial_\gamma^j
 \{n^{-1}\operatorname{RSS}_n(\gamma;\widehat{\bm\nu})
 -\sigma^{*2}-j_n(\gamma-\gamma^*)^2\}\right|=O_p(e_n).
\end{align}
In particular, the fitted RSS is uniformly bounded below by $cn$.

To control the log-determinant part of the profile criterion, let $\bm T$
be a nonnegative row-stochastic matrix with entries bounded by $CK_n/n$.
The power entry bound used in \eqref{eq: stochastic powers resolvents} gives $|\tr(\bm T^j)|\le CK_n$ for every $j\ge1$.
Since $\bm T$ is row stochastic, $\bm I_n-\gamma\bm T$ is invertible for
$|\gamma|<1$. Its determinant is continuous in $\gamma$ and equals one at zero,
so it is positive on $\Gamma_\gamma$. We use the expansion $\log\det(\bm I_n-\gamma\bm T) =-\sum_{j=1}^{\infty}\frac{\gamma^j}{j}\tr(\bm T^j)$.
On $|\gamma|\le1-\delta_\gamma$, this series and its first three derivative
series converge uniformly, allowing term-by-term differentiation.
The sums of the absolute coefficients are bounded by constants depending only on
$\delta_\gamma$, giving
\begin{align}
\label{eq: normalized logdet bound}
 \max_{0\le j\le3}\sup_{\gamma\in\Gamma_\gamma}
 \left|\partial_\gamma^j\{n^{-1}\log\det(\bm I_n-\gamma\bm T)\}\right|
 \le CK_n/n.
\end{align}
By the entry bound established in the proof of Lemma~\ref{lem: response operator bounds},
with probability tending to one, \eqref{eq: normalized logdet bound} holds simultaneously for
$\bm T=(1-t)\bm W^*+t\widehat{\bm W}$, $t\in[0,1]$.

Put $f_n(\gamma)=n^{-1}\operatorname{RSS}_n(\gamma;\widehat{\bm\nu})$ and
$g_n(\gamma)=\sigma^{*2}+j_n(\gamma-\gamma^*)^2$.
Define
\[
 \ell_n^0(\gamma)=-\tfrac12\log g_n(\gamma)-\tfrac12
 =-\tfrac12\log\{\sigma^{*2}+j_n(\gamma-\gamma^*)^2\}-\tfrac12.
\]
Since $g_n'=2j_n(\gamma-\gamma^*)$ and $g_n''=2j_n$, the bounds
$j_n\le C$ and $|\gamma-\gamma^*|\le2$ give
$\sup_{\gamma\in\Gamma_\gamma}(|g_n'|+|g_n''|)\le C$.
Equation~\eqref{eq: quadratic fitted profile control} then gives
$\sup_{\gamma\in\Gamma_\gamma}(|f_n'|+|f_n''|)=O_p(1)$.
Also $\inf_\gamma g_n\ge\sigma^{*2}\ge c>0$, and
$\inf_\gamma f_n\ge c/2$ with probability tending to one.
The mean-value theorem for $\log$ and the identities $(\log f_n)'=f_n'/f_n$, $(\log f_n)''=f_n''/f_n-(f_n'/f_n)^2$
therefore imply, by \eqref{eq: quadratic fitted profile control},
$\max_{0\le j\le2}\sup_{\gamma\in\Gamma_\gamma}
|\partial_\gamma^j(\log f_n-\log g_n)|=O_p(e_n)$.
Finally,
\[
 L_{y,\mathrm{prof}}(\gamma;\widehat{\bm\nu})-\ell_n^0(\gamma)
 =n^{-1}\log\det(\bm I_n-\gamma\widehat{\bm W})
 -\tfrac12\{\log f_n(\gamma)-\log g_n(\gamma)\}.
\]
Combining this identity with \eqref{eq: normalized logdet bound} and
$K_n/n\le e_n$ gives
\begin{align}
\label{eq: supplementary fitted profile uniform}
 \max_{0\le j\le2}\sup_\gamma
 |\partial_\gamma^j\{L_{y,\mathrm{prof}}(\gamma;\widehat{\bm\nu})
 -\ell_n^0(\gamma)\}|=O_p(e_n).
\end{align}
Equations~\eqref{eq: joint design perturbation} and
\eqref{eq: quadratic fitted profile control} give $\Pr(E_n^c)\to0$.
To prove $\widehat\gamma-\gamma^*=o_p(1)$, define on $E_n$
$\Delta_n=\sup_{\gamma\in\Gamma_\gamma}
|L_{y,\mathrm{prof}}(\gamma;\widehat{\bm\nu})-\ell_n^0(\gamma)|$,
and set $\Delta_n=0$ on $E_n^c$.
Equation~\eqref{eq: supplementary fitted profile uniform} gives
$\Delta_n=O_p(e_n)=o_p(1)$.
The maximum grid gap is at most $n^{-1}$, so choose
$\gamma_n^\circ\in\Gamma_{\gamma,n}$ with
$|\gamma_n^\circ-\gamma^*|\le n^{-1}$.
By $\log(1+x)\le x$ for $x\ge0$,
\[
 0\le\ell_n^0(\gamma^*)-\ell_n^0(\gamma_n^\circ)
 =\frac12\log\left\{1+\frac{j_n}{\sigma^{*2}}
 (\gamma_n^\circ-\gamma^*)^2\right\}
 \le\frac{j_n}{2\sigma^{*2}}|\gamma_n^\circ-\gamma^*|^2\le C/n^2.
\]
On $E_n$, $\widehat\gamma$ maximizes $L_{y,\mathrm{prof}}(\gamma;\widehat{\bm\nu})$
on the grid, so
\begin{align*}
 \ell_n^0(\widehat\gamma)
 \ge L_{y,\mathrm{prof}}(\widehat\gamma;\widehat{\bm\nu})-\Delta_n
 \ge L_{y,\mathrm{prof}}(\gamma_n^\circ;\widehat{\bm\nu})-\Delta_n\ge\ell_n^0(\gamma_n^\circ)-2\Delta_n
 \ge\ell_n^0(\gamma^*)-C/n^2-2\Delta_n.
\end{align*}
For any fixed $\varepsilon>0$, the bounds $j_n\ge c$ and
$\sigma^{*2}\le C$ give, whenever $|\gamma-\gamma^*|\ge\varepsilon$,
\[
 \ell_n^0(\gamma^*)-\ell_n^0(\gamma)
 =\frac12\log\left\{1+\frac{j_n(\gamma-\gamma^*)^2}{\sigma^{*2}}\right\}
 \ge\frac12\log(1+c\varepsilon^2/C)=:c_\varepsilon>0.
\]
Consequently, on $E_n$, $|\widehat\gamma-\gamma^*|\ge\varepsilon$ implies
$2\Delta_n+C/n^2\ge c_\varepsilon$. Hence
\[
 \Pr\{|\widehat\gamma-\gamma^*|\ge\varepsilon\}
 \le\Pr\big(E_n\cap\{2\Delta_n+C/n^2\ge c_\varepsilon\}\big)
 +\Pr(E_n^c)\longrightarrow0.
\]
Thus $\widehat\gamma-\gamma^*=o_p(1)$.
Assumption~\ref{assump: response} places $\gamma^*$ a fixed positive distance
from the endpoints of $\Gamma_\gamma$. Thus, with probability tending to one,
$E_n$ holds and $\widehat\gamma$ has adjacent grid points
$\widehat\gamma-d_-$ and $\widehat\gamma+d_+$, where
$0<d_-,d_+\le n^{-1}$.
Define $\widehat\varphi_n(\gamma)=\partial_\gamma
L_{y,\mathrm{prof}}(\gamma;\widehat{\bm\nu})$.
On this event, the grid maximum property and Taylor's theorem give
\begin{align*}
 0&\ge L_{y,\mathrm{prof}}(\widehat\gamma+d_+;\widehat{\bm\nu})
 -L_{y,\mathrm{prof}}(\widehat\gamma;\widehat{\bm\nu})
 =d_+\widehat\varphi_n(\widehat\gamma)+\tfrac12d_+^2\widehat\varphi_n'(\xi_+),\\
 0&\ge L_{y,\mathrm{prof}}(\widehat\gamma-d_-;\widehat{\bm\nu})
 -L_{y,\mathrm{prof}}(\widehat\gamma;\widehat{\bm\nu})
 =-d_-\widehat\varphi_n(\widehat\gamma)+\tfrac12d_-^2\widehat\varphi_n'(\xi_-),
\end{align*}
where $\xi_+\in(\widehat\gamma,\widehat\gamma+d_+)$ and
$\xi_-\in(\widehat\gamma-d_-,\widehat\gamma)$.
The derivative bounds in \eqref{eq: supplementary fitted profile uniform}
give $\sup_{\gamma\in\Gamma_\gamma}|\widehat\varphi_n'(\gamma)|=O_p(1)$.
Dividing the two inequalities by $d_+$ and $d_-$, respectively, therefore yields
\[
 -\frac{d_-}{2}\sup_{\gamma\in\Gamma_\gamma}|\widehat\varphi_n'(\gamma)|
 \le\widehat\varphi_n(\widehat\gamma)
 \le\frac{d_+}{2}\sup_{\gamma\in\Gamma_\gamma}|\widehat\varphi_n'(\gamma)|, \quad |\widehat\varphi_n(\widehat\gamma)|=O_p(n^{-1}).
\]

Next, direct differentiation gives
\[
 (\ell_n^0)'(\gamma^*)=0,\qquad
 (\ell_n^0)''(\gamma)
 =-\frac{j_n\{\sigma^{*2}-j_n(\gamma-\gamma^*)^2\}}
 {\{\sigma^{*2}+j_n(\gamma-\gamma^*)^2\}^2}.
\]
Since $j_n$ and $\sigma^{*2}$ are bounded above and away from zero,
there are fixed $\delta_0,c_0>0$ such that
$(\ell_n^0)''(\gamma)\le-c_0$ whenever
$|\gamma-\gamma^*|\le\delta_0$.
By consistency, $|\widehat\gamma-\gamma^*|\le\delta_0$ with probability tending to one.
The mean-value theorem and \eqref{eq: supplementary fitted profile uniform} then give
\begin{align*}
 c_0|\widehat\gamma-\gamma^*|
 &\le| (\ell_n^0)'(\widehat\gamma)-(\ell_n^0)'(\gamma^*)|\\
 &\le|\widehat\varphi_n(\widehat\gamma)|+
 \sup_{\gamma\in\Gamma_\gamma}|(\ell_n^0)'(\gamma)-\widehat\varphi_n(\gamma)|
 =O_p(n^{-1})+O_p(e_n)=O_p(e_n),
\end{align*}
where $n^{-1}\le e_n$ follows from $s_n\le n/4$ and $K_n\ge1$.
Hence $\widehat\gamma-\gamma^*=O_p(e_n)$.

For $\bm a=(\bm\beta^\top,\bm\eta^\top)^\top$, on $E_n$ the joint
least-squares solution satisfies
$\widehat{\bm a}=\widehat{\bm X}_+^{\dagger}
(\bm I_n-\widehat\gamma\widehat{\bm W})\bm y$ and
$\widehat{\bm X}_+^{\dagger}\widehat{\bm X}_+=\bm I_{p+r}$.
The response model gives
$(\bm I_n-\gamma^*\bm W^*)\bm y=\bm X_+^*\bm a^*+\bm\epsilon$.
Since $(\widehat{\bm X}_+-\bm X_+^*)\bm a^*
=(\widehat{\bm Z}-\bm Z^*)\bm\eta^*$, we obtain
\begin{align*}
 (\bm I_n-\widehat\gamma\widehat{\bm W})\bm y
 -\widehat{\bm X}_+\bm a^*
 &=\bm\epsilon-(\widehat{\bm Z}-\bm Z^*)\bm\eta^*
 -(\widehat\gamma\widehat{\bm W}-\gamma^*\bm W^*)\bm y\\
 &=\bm\epsilon-(\widehat{\bm Z}-\bm Z^*)\bm\eta^*
 -\gamma^*(\widehat{\bm W}-\bm W^*)\bm y
 -(\widehat\gamma-\gamma^*)\widehat{\bm W}\bm y.
\end{align*}
Multiplying by $\widehat{\bm X}_+^{\dagger}$ therefore gives
\[
 \widehat{\bm a}-\bm a^*=
 \widehat{\bm X}_+^{\dagger}\{\bm\epsilon
 -(\widehat{\bm Z}-\bm Z^*)\bm\eta^*
 -\gamma^*(\widehat{\bm W}-\bm W^*)\bm y
 -(\widehat\gamma-\gamma^*)\widehat{\bm W}\bm y\}.
\]
Conditional on the network, independence of $\bm\epsilon$ and
$\widehat{\bm X}_+$ gives, on the event in
\eqref{eq: joint design perturbation},
\[
 \mathbb E_{\epsilon}\|\widehat{\bm X}_+^{\dagger}\bm\epsilon\|_2^2
 =\sigma^{*2}\tr\{(\widehat{\bm X}_+^{\top}\widehat{\bm X}_+)^{-1}\}
 \le C/n.
\]
Thus $\|\widehat{\bm X}_+^{\dagger}\bm\epsilon\|_2=O_p(n^{-1/2})$.
Using Theorem~\ref{thm: first stage},
Lemma~\ref{lem: response operator bounds},
\eqref{eq: joint design perturbation}, and
$\widehat\gamma-\gamma^*=O_p(e_n)$, the other three terms satisfy
\begin{align*}
 \|\widehat{\bm X}_+^{\dagger}(\widehat{\bm Z}-\bm Z^*)\bm\eta^*\|_2
 &\le\|\widehat{\bm X}_+^{\dagger}\|_2
 \|\widehat{\bm Z}-\bm Z^*\|_F\|\bm\eta^*\|_2
 =O_p(s_n^{-1/2}),\\
 |\gamma^*|\|\widehat{\bm X}_+^{\dagger}(\widehat{\bm W}-\bm W^*)\bm y\|_2
 &\le |\gamma^*|\|\widehat{\bm X}_+^{\dagger}\|_2
 \|\widehat{\bm W}-\bm W^*\|_F\|\bm y\|_2=O_p(e_n),\\
 |\widehat\gamma-\gamma^*|\|\widehat{\bm X}_+^{\dagger}\widehat{\bm W}\bm y\|_2
 &\le |\widehat\gamma-\gamma^*|\|\widehat{\bm X}_+^{\dagger}\|_2
 \|\widehat{\bm W}\bm y\|_2=O_p(e_n).
\end{align*}
Since $n^{-1/2}+s_n^{-1/2}=O(e_n)$, this proves
$\|\widehat{\bm a}-\bm a^*\|_2=O_p(e_n)$.
Equation~\eqref{eq: quadratic fitted profile control}
and $\widehat\tau=n/\operatorname{RSS}_n(\widehat\gamma;\widehat{\bm\nu})$
give $\widehat\tau-\tau^*=O_p(e_n)$.
Combining the bounds for $\widehat{\bm a}-\bm a^*$,
$\widehat\gamma-\gamma^*$ and $\widehat\tau-\tau^*$ proves
\eqref{eq: supplementary profile estimator rate}, since $e_n=K_n/\sqrt{s_n}=o(1)$.
At the algorithm output, define the fitted residual $\widehat{\bm\epsilon} =(\bm I_n-\widehat\gamma\widehat{\bm W})\bm y -\bm X\widehat{\bm\beta}-\widehat{\bm Z}\widehat{\bm\eta}$.
The joint least-squares equations and
$\widehat\tau=n/\operatorname{RSS}_n(\widehat\gamma;\widehat{\bm\nu})
=n/\|\widehat{\bm\epsilon}\|_2^2$ give $\bm X^\top\widehat{\bm\epsilon}=\bm0$, $\widehat{\bm Z}^{\top}\widehat{\bm\epsilon}=\bm0$, $\|\widehat{\bm\epsilon}\|_2^2=n/\widehat\tau$.
Thus the $\bm\beta$-, $\bm\eta$- and $\tau$-components of the full score vanish.
To identify its gamma coordinate, differentiate the scalar profile criterion
along $(\widehat{\bm\beta}_\gamma,\widehat{\bm\eta}_\gamma,\gamma,\widehat\tau_\gamma)$.
Put $\widehat{\bm\zeta}_\gamma=(\widehat{\bm\beta}_\gamma^\top,\widehat{\bm\eta}_\gamma^\top,\gamma,\widehat\tau_\gamma)^\top$.
The chain rule gives
\begin{align*}
 \widehat\varphi_n(\gamma)
 =q_{y,\gamma}(\widehat{\bm\zeta}_\gamma;\widehat{\bm\nu})
 +\bm q_{y,\beta}(\widehat{\bm\zeta}_\gamma;\widehat{\bm\nu})^{\top}\frac{d\widehat{\bm\beta}_\gamma}{d\gamma}+\bm q_{y,\eta}(\widehat{\bm\zeta}_\gamma;\widehat{\bm\nu})^{\top}\frac{d\widehat{\bm\eta}_\gamma}{d\gamma}
 +q_{y,\tau}(\widehat{\bm\zeta}_\gamma;\widehat{\bm\nu})\frac{d\widehat\tau_\gamma}{d\gamma}
 =q_{y,\gamma}(\widehat{\bm\zeta}_\gamma;\widehat{\bm\nu}),
\end{align*}
where the other
three terms vanish by their first-order equations. Since
$\widehat{\bm\zeta}_{\widehat\gamma}=\widehat{\bm\zeta}$, we obtain
$q_{y,\gamma}(\widehat{\bm\zeta};\widehat{\bm\nu})
=\widehat\varphi_n(\widehat\gamma)=O_p(n^{-1})$, proving
\eqref{eq: supplementary full grid score}.

\hfill$\blacksquare$

\begin{lemma}\label{lem: fitted response Jacobian}
Suppose Assumptions~\ref{assump: network}--\ref{assump: response errors} hold.
For any random $\widetilde{\bm\zeta}$ satisfying $\|\widetilde{\bm\zeta}-\bm\zeta^*\|_2=o_p(1)$,
\begin{align}
\label{eq: supplementary fitted Jacobian consistency}
 \|-\nabla_{\bm\zeta}\bm q_y(\widetilde{\bm\zeta};\widehat{\bm\nu})-\mathsf J_{y,n}\|_2 =o_p(1).
\end{align}
For some fixed sufficiently small $\delta>0$, let
$\mathcal B_\delta=\{\bm\zeta\in\mathbb R^{p+r+2}:
\|\bm\zeta-\bm\zeta^*\|_2\le\delta\}$. Then
\begin{align*}
 \sup_{\substack{\bm\zeta_1,\bm\zeta_2\in\mathcal B_\delta,
 \bm\zeta_1\ne\bm\zeta_2}}
 \frac{\|\nabla_{\bm\zeta}\bm q_y(\bm\zeta_1;\widehat{\bm\nu})
 -\nabla_{\bm\zeta}\bm q_y(\bm\zeta_2;\widehat{\bm\nu})\|_2}
 {\|\bm\zeta_1-\bm\zeta_2\|_2}
 =O_p(1).
\end{align*}
\end{lemma}

\noindent\textbf{Proof of Lemma~\ref{lem: fitted response Jacobian}.}
Put $e_n=K_n/\sqrt{s_n}=o(1)$ and write
$\bm H_n(\bm\zeta;\bm\nu)=-\nabla_{\bm\zeta}\bm q_y(\bm\zeta;\bm\nu)$.
We first control $\bm H_n(\bm\zeta^*;\bm\nu^*)-\mathsf J_{y,n}$.
Since $\bm B_y=\bm C_n+\bm Q\bm\epsilon\widetilde{\bm e}_\gamma^\top$ and
$\bm r_y(\bm\zeta^*;\bm\nu^*)=\bm\epsilon$, the random blocks in
\eqref{eq: full response Hessian} satisfy
\begin{align*}
 \bm B_y^\top\bm B_y-\mathbb E_\epsilon(\bm B_y^\top\bm B_y)
 &=\bm C_n^\top\bm Q\bm\epsilon\widetilde{\bm e}_\gamma^\top
 +\widetilde{\bm e}_\gamma\bm\epsilon^\top\bm Q^\top\bm C_n
 +\widetilde{\bm e}_\gamma\widetilde{\bm e}_\gamma^\top\{\bm\epsilon^\top\bm Q^\top\bm Q\bm\epsilon
 -\sigma^{*2}\|\bm Q\|_F^2\},\\
 \bm B_y^\top\bm\epsilon-\mathbb E_\epsilon(\bm B_y^\top\bm\epsilon)
 &=\bm C_n^\top\bm\epsilon
 +\widetilde{\bm e}_\gamma\{\bm\epsilon^\top\bm Q^\top\bm\epsilon
 -\sigma^{*2}\tr(\bm Q)\}.
\end{align*}
Lemma~\ref{lem: response conditioning} gives $\|\bm C_n\|_2\le C\sqrt n$,
and Lemma~\ref{lem: response operator bounds} gives
$\|\bm Q\|_2\le\|\bm Q\|_F\le C\sqrt{K_n}$.
Thus the linear coefficient norms are at most $C\sqrt{nK_n}$ and
$C\sqrt n$, respectively, while
$\|\bm Q^\top\bm Q\|_F\le\|\bm Q\|_2\|\bm Q\|_F\le CK_n$.
Applying \eqref{eq: linear quadratic variance bound} to these terms
(with $(\bm Q+\bm Q^\top)/2$ for the last quadratic form), and dividing by $n$, gives
\[
 \|\bm H_n(\bm\zeta^*;\bm\nu^*)-\mathsf J_{y,n}\|_2
 =O_p\!\left(\sqrt{K_n/n}+K_n/n\right)
 =O_p\!\left(\sqrt{K_n/n}\right),
\]
since the matrix dimension is fixed and $K_n/n=o(1)$.

Next, replace the true inputs by the fitted inputs. Define
$\widehat{\bm B}=(\bm X,\widehat{\bm Z},\widehat{\bm W}\bm y)$ and
$\widehat{\bm r}_*=\bm r_y(\bm\zeta^*;\widehat{\bm\nu})$.
Then
\[
 \widehat{\bm B}-\bm B_y
 =(\bm0,\widehat{\bm Z}-\bm Z^*,(\widehat{\bm W}-\bm W^*)\bm y),\quad
 \widehat{\bm r}_*-\bm\epsilon
 =-(\widehat{\bm Z}-\bm Z^*)\bm\eta^*
 -\gamma^*(\widehat{\bm W}-\bm W^*)\bm y.
\]
Theorem~\ref{thm: first stage} and Lemma~\ref{lem: response operator bounds}
therefore imply
\[
 n^{-1/2}\{\|\widehat{\bm B}-\bm B_y\|_2
 +\|\widehat{\bm r}_*-\bm\epsilon\|_2\}=O_p(e_n),\quad
 n^{-1/2}\{\|\bm B_y\|_2+\|\widehat{\bm B}\|_2
 +\|\bm\epsilon\|_2+\|\widehat{\bm r}_*\|_2\}=O_p(1).
\]
Consequently, the two product blocks obey
\begin{align*}
 n^{-1}\|\widehat{\bm B}^\top\widehat{\bm B}-\bm B_y^\top\bm B_y\|_2
 &\le n^{-1}\|\widehat{\bm B}-\bm B_y\|_2
 (\|\widehat{\bm B}\|_2+\|\bm B_y\|_2)=O_p(e_n),\\
 n^{-1}\|\widehat{\bm B}^\top\widehat{\bm r}_*-\bm B_y^\top\bm\epsilon\|_2
 &\le n^{-1}\{\|\widehat{\bm B}-\bm B_y\|_2\|\widehat{\bm r}_*\|_2
 +\|\bm B_y\|_2\|\widehat{\bm r}_*-\bm\epsilon\|_2\}=O_p(e_n).
\end{align*}
It remains to control the term
$t_\gamma'\widetilde{\bm e}_\gamma\widetilde{\bm e}_\gamma^\top$
in \eqref{eq: full response Hessian}, where
\[
 t_\gamma'
 =\frac1n\tr\{[\bm W^*(\bm I_n-\gamma\bm W^*)^{-1}]^2\}
 =-\frac{\partial^2}{\partial\gamma^2}
 \left\{\frac1n\log\det(\bm I_n-\gamma\bm W^*)\right\}.
\]
The fitted term has the same expression with $\bm W^*$ replaced by
$\widehat{\bm W}$. By \eqref{eq: normalized logdet bound}, both scalar
coefficients are $O_p(K_n/n)$ uniformly in $\gamma$, so their difference
has the same order.
Since $K_n/n\le e_n$, the fitted version of
\eqref{eq: full response Hessian} now yields
\[
 \|\bm H_n(\bm\zeta^*;\widehat{\bm\nu})
 -\bm H_n(\bm\zeta^*;\bm\nu^*)\|_2=O_p(e_n).
\]

It remains to control variation in $\bm\zeta$.
Write $\bm\theta=(\bm\beta^\top,\bm\eta^\top,\gamma)^\top$ and define
\[
 \widehat t_\gamma=n^{-1}\tr\{\widehat{\bm W}
 (\bm I_n-\gamma\widehat{\bm W})^{-1}\}.
\]
Choose the fixed $\delta>0$ sufficiently small that $\tau$ is bounded
away from zero and $\gamma\in\Gamma_\gamma$ on $\mathcal B_\delta$,
as permitted by Assumption~\ref{assump: response}.
The fitted residual satisfies
\[
 \bm r_y(\bm\zeta;\widehat{\bm\nu})
 =\widehat{\bm r}_*-\widehat{\bm B}(\bm\theta-\bm\theta^*),\quad
 \sup_{\bm\zeta\in\mathcal B_\delta}
 n^{-1/2}\|\bm r_y(\bm\zeta;\widehat{\bm\nu})\|_2=O_p(1).
\]
For $\bm\zeta_1,\bm\zeta_2\in\mathcal B_\delta$, the residual difference is
$-\widehat{\bm B}(\bm\theta_1-\bm\theta_2)$.
The third-derivative bound in \eqref{eq: normalized logdet bound} also gives
$\sup_\gamma|\widehat t_\gamma''|=O_p(K_n/n)$.
Applying the mean-value theorem to $\widehat t_\gamma'$ and $1/(2\tau^2)$,
and bounding the other blocks in \eqref{eq: full response Hessian}, yields
\begin{align*}
 \|\bm H_n(\bm\zeta_1;\widehat{\bm\nu})
 -\bm H_n(\bm\zeta_2;\widehat{\bm\nu})\|_2
 &\le\left\{2n^{-1}\|\widehat{\bm B}\|_2^2
 +\sup_{\gamma\in\Gamma_\gamma}|\widehat t_\gamma''|+C\right\}
 \|\bm\zeta_1-\bm\zeta_2\|_2\\
 &=O_p(1)\|\bm\zeta_1-\bm\zeta_2\|_2,
\end{align*}
where the random bound is uniform over both points. This is the asserted
Lipschitz bound. Finally, $\widetilde{\bm\zeta}\in\mathcal B_\delta$ with
probability tending to one, and
\begin{align*}
 \bm H_n(\widetilde{\bm\zeta};\widehat{\bm\nu})-\mathsf J_{y,n}
 ={}&\{\bm H_n(\widetilde{\bm\zeta};\widehat{\bm\nu})
 -\bm H_n(\bm\zeta^*;\widehat{\bm\nu})\}\\
 &+\{\bm H_n(\bm\zeta^*;\widehat{\bm\nu})
 -\bm H_n(\bm\zeta^*;\bm\nu^*)\}
 +\{\bm H_n(\bm\zeta^*;\bm\nu^*)-\mathsf J_{y,n}\}.
\end{align*}
The three terms have norms
$O_p(1)\|\widetilde{\bm\zeta}-\bm\zeta^*\|_2=o_p(1)$,
$O_p(e_n)=o_p(1)$, and $O_p(\sqrt{K_n/n})=o_p(1)$, respectively.
This proves \eqref{eq: supplementary fitted Jacobian consistency} and
completes the proof.
\hfill$\blacksquare$

\begin{lemma}\label{lem: supplementary W expansion}
Suppose Assumptions~\ref{assump: network} and~\ref{assump: network growth} hold. Let $\widehat{\bm h}=\widehat{\bm\nu}-\bm\nu^*$
with $\psi(\widehat{\bm h})\in\mathcal T_n$. Then
\begin{align}
\label{eq: supplementary W first expansion}
 \|\widehat{\bm W}-\bm W^*
 -D_\nu\bm W(\bm\nu^*)[\widehat{\bm h}]\|_F
 &=O_p(K_n/s_n).
\end{align}
For arbitrary
tangent directions $\bm h=(\bm h_\alpha,\bm E_Z)$ at the truth,
\begin{align}
\label{eq: derivative operator bounds}
 \|D_\nu\bm W(\bm\nu^*)[\bm h]\|_F
 \le CK_n n^{-1/2}\|\psi(\bm h)\|_2.
\end{align}
If Assumptions~\ref{assump: response}--\ref{assump: response errors} also hold, then
\begin{align}
\label{eq: supplementary response composition true}
 \bm q_y(\bm\zeta^*;\widehat{\bm\nu})
 -\bm q_y(\bm\zeta^*;\bm\nu^*)
 -D_\nu\bm q_y(\bm\zeta^*;\bm\nu^*)[\widehat{\bm h}]
 =O_p(K_n^2/s_n).
\end{align}
For every score coordinate, the network derivative also satisfies
\begin{align}
\label{eq: response network derivative bound}
 \max_{a\in[p+r+2]}
 \sup_{\substack{\psi(\bm h)\in\mathcal T_n\\\|\psi(\bm h)\|_2=1}}
 |D_\nu q_{y,a}(\bm\zeta^*;\bm\nu^*)[\bm h]|
 =O_p(K_n/\sqrt n).
\end{align}
\end{lemma}

\noindent\textbf{Proof of Lemma~\ref{lem: supplementary W expansion}.}
The derivative with respect to $\bm P$ in a hollow symmetric direction $\bm E_P$ is
\begin{align*}
 D_P\bm W(\bm P)[\bm E_P]
 =\bm D(\bm P)^{-1}\bm E_P
 -\bm D(\bm P)^{-1}\diag(\bm E_P\bm1_n)\bm W(\bm P).
\end{align*}
All dyad arrays in this proof are defined for $i\ne j$ and extended by zero on the
diagonal. Put $r^P_{ij}=(\widehat p_{ij}-p_{ij}^*)/p_{ij}^*$. For an array $g=(g_{ij})$
with $g_{ii}=0$, define
\begin{align*}
 \bar r_i^P&=\sum_jW_{ij}^*r^P_{ij},&
 \mathscr N_W(g)_{ij}&=W_{ij}^*\left(g_{ij}-
 \sum_\ell W_{i\ell}^*g_{i\ell}\right).
\end{align*}
The degree-normalization identity is
\begin{align}
\label{eq: exact weighted normalization}
 \widehat W_{ij}-W_{ij}^*
 =\mathscr N_W(r^P)_{ij}/(1+\bar r_i^P).
\end{align}
Since $|\Delta_{\theta,ij}|\le C(\Delta_i+\Delta_j)$ and
$\max_i\Delta_i=o_p(1)$, the identity
$\log(\widehat p_{ij}/p_{ij}^*)
=\log\expit(\widehat\theta_{ij})-\log\expit(\theta_{ij}^*)$
and $|d\log\expit(x)/dx|\le1$ give
$|r^P_{ij}|\le C|\Delta_{\theta,ij}|$.
Row stochasticity then gives $\max_i|\bar r_i^P|=o_p(1)$.
Write $\ell_{ij}^\theta=\dot\theta_{ij}^*[\widehat{\bm h}],\ r_{ij}^\theta=\bm\Delta_{z,i}^\top\bm\Delta_{z,j},\ \ell_{ij}^P=(1-p_{ij}^*)\ell_{ij}^\theta$.
Then $\Delta_{\theta,ij}=\ell_{ij}^\theta+r_{ij}^\theta$.
For $u=\theta_{ij}^*+t\Delta_{\theta,ij}$ with $t\in[0,1]$,
the mean-value theorem applied to $\log\expit$ gives
\[
 \left|\log\frac{\expit(u)}{p_{ij}^*}\right|
 \le|u-\theta_{ij}^*|\le|\Delta_{\theta,ij}|.
\]
Because $\max_{i\ne j}|\Delta_{\theta,ij}|=o_p(1)$, with probability tending to one,
\[
 |\expit''(u)|\le\expit(u)
 \le p_{ij}^*e^{|\Delta_{\theta,ij}|}\le Cp_{ij}^*
\]
uniformly over $i\ne j$ and $t\in[0,1]$.
Taylor expansion of $\widehat p_{ij}=\expit(\widehat\theta_{ij})$
about $\theta_{ij}^*$, using
$\expit'(\theta_{ij}^*)=p_{ij}^*(1-p_{ij}^*)$, gives
\[
 \widehat p_{ij}-p_{ij}^*
 =p_{ij}^*(1-p_{ij}^*)\Delta_{\theta,ij}
 +\tfrac12\expit''(\xi_{ij})\Delta_{\theta,ij}^2,
\]
where $\xi_{ij}$ lies between $\theta_{ij}^*$ and $\widehat\theta_{ij}$.
Dividing by $p_{ij}^*$ and subtracting
$\ell_{ij}^P=(1-p_{ij}^*)\ell_{ij}^\theta$ therefore yields
\[
r_{ij}^P-\ell_{ij}^P
 =(1-p_{ij}^*)r_{ij}^\theta
 +\frac{\expit''(\xi_{ij})}{2p_{ij}^*}\Delta_{\theta,ij}^2,\quad 
 |r_{ij}^P-\ell_{ij}^P|
 \le |r_{ij}^\theta|+C\Delta_{\theta,ij}^2
 \le C\{\Delta_{\theta,ij}^2+|r_{ij}^\theta|\}.
\]
To bound the squared Taylor remainder, use
$|\Delta_{\theta,ij}|\le C(\Delta_i+\Delta_j)$ and the fourth-moment bound in
Lemma~\ref{lem: local likelihood root} to obtain
\begin{align}
\label{eq: dyad logit fourth moment}
 \frac1{m_n}\sum_{i<j}\Delta_{\theta,ij}^4
 \le\frac C{m_n}\sum_{i<j}(\Delta_i^4+\Delta_j^4)
 =\frac{C(n-1)}{m_n}\sum_i\Delta_i^4
 =O_p(s_n^{-2}).
\end{align}
Since $|\ell_{ij}^\theta|\le C(\Delta_i+\Delta_j)$ and
$|r_{ij}^\theta|\le\|\bm\Delta_{z,i}\|_2\|\bm\Delta_{z,j}\|_2$,
the same nodewise summation and Cauchy--Schwarz give
\begin{align}
\label{eq: dyad linear remainder moments}
 \frac1{m_n}\sum_{i<j}(\ell_{ij}^\theta)^2&=O_p(s_n^{-1}),\\
\label{eq: dyad quadratic remainder moments}
 \begin{split}
 \frac1{m_n}\sum_{i<j}(r_{ij}^\theta)^2
 &\le\frac C{n^2}\|\bm\Delta_Z\|_F^4=O_p(s_n^{-2}),\\
 \frac1{m_n}\sum_{i<j}(r_{ij}^\theta)^4
 &\le\max_i\|\bm\Delta_{z,i}\|_2^4
       \frac1{m_n}\sum_{i<j}(r_{ij}^\theta)^2
 =O_p(\delta_{n,k}^4s_n^{-2}).
 \end{split}
\end{align}
Here $\|\bm\Delta_Z\|_F^2=\sum_i\|\bm\Delta_{z,i}\|_2^2
\le\sum_i\Delta_i^2=O_p(n/s_n)$ gives the second-moment rate in
\eqref{eq: dyad quadratic remainder moments}; its fourth-moment rate also uses
$\max_i\|\bm\Delta_{z,i}\|_2\le\max_i\Delta_i=O_p(\delta_{n,k})$.
These inequalities are deterministic before the stated stochastic rates
are substituted, so they also apply to any other aligned fit with the
same nodewise moment bounds. 

Subtracting the $W_{ij}^*$-weighted row mean does not increase the
weighted second moment. Together with $W_{ij}^*\le CK_n/n$, this gives
\begin{align}
\label{eq: weighted row bound}
 \sum_j\mathscr N_W(g)_{ij}^2
 \le\frac{CK_n}{n}\sum_jW_{ij}^*g_{ij}^2.
\end{align}
Summing \eqref{eq: weighted row bound} and using $W_{ij}^*\le CK_n/n$
gives $\|\mathscr N_W(g)\|_F^2\le CK_n^2n^{-2}\sum_{i,j}g_{ij}^2$.
For $g=r^P-\ell^P$, the Taylor bound and
\eqref{eq: dyad logit fourth moment}--\eqref{eq: dyad quadratic remainder moments}
therefore give
\[
 \|\mathscr N_W(r^P-\ell^P)\|_F^2
 \le\frac{CK_n^2}{n^2}\sum_{i<j}
 \{\Delta_{\theta,ij}^4+(r_{ij}^\theta)^2\}
 =O_p(K_n^2/s_n^2).
\]
Moreover,
\eqref{eq: exact weighted normalization} gives the exact denominator identity
\begin{align}
\label{eq: exact W denominator remainder}
 \widehat W_{ij}-W_{ij}^*-\mathscr N_W(\ell^P)_{ij}
 =\frac{\mathscr N_W(r^P-\ell^P)_{ij}
 -\bar r_i^P\mathscr N_W(\ell^P)_{ij}}
 {1+\bar r_i^P}.
\end{align}
Let $v_i=\sum_jW_{ij}^*\{\Delta_{\theta,ij}^2+(r_{ij}^\theta)^2\}$.
Weighted Cauchy--Schwarz and $|r_{ij}^P|\le C|\Delta_{\theta,ij}|$
give $|\bar r_i^P|^2\le Cv_i$.
Also, \eqref{eq: weighted row bound} and
$\ell_{ij}^P=(1-p_{ij}^*)(\Delta_{\theta,ij}-r_{ij}^\theta)$ give
$\sum_j\mathscr N_W(\ell^P)_{ij}^2\le C(K_n/n)v_i$. Hence
\begin{align*}
 \sum_{i,j}\{\bar r_i^P\mathscr N_W(\ell^P)_{ij}\}^2
 \le\frac{CK_n}{n}\sum_i v_i^2\le\frac{CK_n}{n}\sum_{i,j}W_{ij}^*
 \{\Delta_{\theta,ij}^4+(r_{ij}^\theta)^4\}
 =O_p(K_n^2/s_n^2),
\end{align*}
where Jensen's inequality gives the second inequality, and
\eqref{eq: dyad logit fourth moment} and \eqref{eq: dyad quadratic remainder moments}
give the final order.
Since $\max_i|\bar r_i^P|=o_p(1)$, with probability tending to one
$\min_i(1+\bar r_i^P)\ge1/2$, so division by $1+\bar r_i^P$ in
\eqref{eq: exact W denominator remainder} preserves the $O_p(K_n/s_n)$
Frobenius norm bound for the numerator.
Since $\mathscr N_W(\ell^P)=D_\nu\bm W(\bm\nu^*)[\widehat{\bm h}]$, this proves
\eqref{eq: supplementary W first expansion}.

For a general direction $\bm h$ with $\psi(\bm h)\in\mathcal T_n$, put
$g_{ij}=(1-p_{ij}^*)\dot\theta_{ij}^*[\bm h]$.
Similarly, $\mathscr N_W(g)=D_\nu\bm W(\bm\nu^*)[\bm h]$.
Applying \eqref{eq: weighted row bound} and summing over $i$ gives
\begin{align*}
 \|D_\nu\bm W(\bm\nu^*)[\bm h]\|_F^2
 \le\frac{CK_n}{n}\sum_{i,j}W_{ij}^*g_{ij}^2\le\frac{CK_n^2}{n^2}\sum_{i<j}\dot\theta_{ij}^*[\bm h]^2
 \le\frac{CK_n^2}{n}\|\psi(\bm h)\|_2^2,
\end{align*}
where the second inequality uses $W_{ij}^*\le CK_n/n$ and $1-p_{ij}^*\le1$,
and the last uses the upper bound in \eqref{eq: offdiag derivative geometry}
of Lemma~\ref{lem: primitive network Fisher}.
Taking square roots proves \eqref{eq: derivative operator bounds}.

For the response expansion, regard \eqref{eq: full response score}
as a function $\mathscr Q_{\bm\zeta}(\bm W,\bm Z)$, and define its residual by
$\mathscr R_{\bm\zeta}(\bm W,\bm Z)=(\bm I_n-\gamma\bm W)\bm y-\bm X\bm\beta-\bm Z\bm\eta$.
On the product of the fitted--true segments, all of
$\bm y,\bm W\bm y,\mathscr R_{\bm\zeta}(\bm W,\bm Z)$ have norms $O_p(\sqrt n)$ and
$\|\bm Z\|_2=O(\sqrt n)$. For directions $(\bm E_j,\bm F_j)$ in
$(\bm W,\bm Z)$, put
$\bm u_j=\bm E_j\bm y$ and
$\bm a_j=-\gamma\bm u_j-\bm F_j\bm\eta$.
In the following display, $D_{(W,Z)}q_a$ and $D_{(W,Z)}^2q_a$ denote the $a$-th components
of $D_{(W,Z)}\mathscr Q_{\bm\zeta}(\bm W,\bm Z)[\bm E_1,\bm F_1]$ and
$D_{(W,Z)}^2\mathscr Q_{\bm\zeta}(\bm W,\bm Z)[(\bm E_1,\bm F_1),(\bm E_2,\bm F_2)]$.
With $\bm\zeta$ fixed, the first and second derivatives are
\begin{equation}\label{eq: response input first derivatives}
\begin{aligned}
 D_{(W,Z)}q_\beta&=\tau\bm X^\top\bm a_1/n,\\
 D_{(W,Z)}q_\eta&=\tau(\bm F_1^\top\mathscr R_{\bm\zeta}(\bm W,\bm Z)+\bm Z^\top\bm a_1)/n,\\
 D_{(W,Z)}q_\gamma&=-\frac1n\tr\{D_W\bm Q_\gamma(\bm W)[\bm E_1]\}
 +\frac{\tau}{n}\{\bm u_1^\top\mathscr R_{\bm\zeta}(\bm W,\bm Z)+(\bm W\bm y)^\top\bm a_1\},\\
 D_{(W,Z)}q_\tau&=-\mathscr R_{\bm\zeta}(\bm W,\bm Z)^\top\bm a_1/n,
\end{aligned}
\end{equation}
\begin{equation}\label{eq: response input second derivatives}
\begin{aligned}
 D_{(W,Z)}^2q_\beta&=0,\\
 D_{(W,Z)}^2q_\eta&=\tau(\bm F_1^\top\bm a_2+\bm F_2^\top\bm a_1)/n,\\
 D_{(W,Z)}^2q_\gamma&=-\frac1n\tr\{D_W^2\bm Q_\gamma(\bm W)[\bm E_1,\bm E_2]\}
 +\frac{\tau}{n}(\bm u_1^\top\bm a_2+\bm u_2^\top\bm a_1),\\
 D_{(W,Z)}^2q_\tau&=-\bm a_1^\top\bm a_2/n.
\end{aligned}
\end{equation}
Write $d_j=\|\bm E_j\|_F+n^{-1/2}\|\bm F_j\|_F$.
Since $\|\bm u_j\|_2+\|\bm a_j\|_2=O_p(\sqrt n)d_j$,
Cauchy--Schwarz bounds the terms outside the traces in
\eqref{eq: response input first derivatives} by $O_p(1)d_1$
and those in \eqref{eq: response input second derivatives} by $O_p(1)d_1d_2$.
For the trace term put $\bm T_\gamma(\bm W)=(\bm I_n-\gamma\bm W)^{-1}$.
Differentiating the inverse and using $\bm I_n+\gamma\bm W\bm T_\gamma(\bm W)=\bm T_\gamma(\bm W)$ gives
\begin{equation*}\begin{aligned}
D_W\bm Q_\gamma(\bm W)[\bm E]&=\bm T_\gamma(\bm W)\bm E\bm T_\gamma(\bm W),\\
 D_W^2\bm Q_\gamma(\bm W)[\bm E_1,\bm E_2]
 &=\gamma\bm T_\gamma(\bm W)(\bm E_1\bm T_\gamma(\bm W)\bm E_2+
\bm E_2\bm T_\gamma(\bm W)\bm E_1)\bm T_\gamma(\bm W).
\end{aligned}
\end{equation*}
Cyclic invariance of the trace and \eqref{eq: stochastic powers resolvents} imply
\begin{equation*}\begin{aligned}
n^{-1}|\tr D_W\bm Q_\gamma(\bm W)[\bm E]|&\le CK_n n^{-1/2}\|\bm E\|_F,\\
 n^{-1}|\tr D_W^2\bm Q_\gamma(\bm W)[\bm E_1,\bm E_2]|
 &\le CK_n^{3/2}n^{-1}\|\bm E_1\|_F\|\bm E_2\|_F.
\end{aligned}
\end{equation*}
Indeed the first trace is at most $\|\bm T_\gamma(\bm W)^2\|_F\|\bm E\|_F$,
and each second-order trace is at most
$\|\bm T_\gamma(\bm W)\|_2^3\|\bm E_1\|_F\|\bm E_2\|_F$.
These bounds hold along the entire transition segment.
Since $K_n/\sqrt n$ and $K_n^{3/2}/n$ are bounded by
\eqref{eq: basic growth consequences}, the trace derivatives satisfy
the same first- and second-order bounds as the other terms.
Consequently, uniformly on these segments,
\begin{equation}\label{eq: response Frechet bounds}
\begin{aligned}
 \|D_{(W,Z)}\mathscr Q_{\bm\zeta}(\bm W,\bm Z)[\bm E_1,\bm F_1]\|_2&\le O_p(1)d_1,\\
 \|D_{(W,Z)}^2\mathscr Q_{\bm\zeta}(\bm W,\bm Z)[(\bm E_1,\bm F_1),(\bm E_2,\bm F_2)]\|_2
 &\le O_p(1)d_1d_2.
\end{aligned}
\end{equation}
Put $\Delta\bm W=\widehat{\bm W}-\bm W^*$ and
$\bm R_W=\Delta\bm W-D_\nu\bm W(\bm\nu^*)[\widehat{\bm h}]$.
Since $\bm q_y(\bm\zeta^*;\bm\nu)
=\mathscr Q_{\bm\zeta^*}(\bm W(\bm\nu),\bm Z)$ and the $\bm Z$-component
of $\widehat{\bm h}$ is $\bm\Delta_Z$, the chain rule gives
\[
 D_\nu\bm q_y(\bm\zeta^*;\bm\nu^*)[\widehat{\bm h}]
 =D_{(W,Z)}\mathscr Q_{\bm\zeta^*}(\bm W^*,\bm Z^*)
 [D_\nu\bm W(\bm\nu^*)[\widehat{\bm h}],\bm\Delta_Z].
\]
Expand $\mathscr Q_{\bm\zeta^*}(\widehat{\bm W},\widehat{\bm Z})$
about $(\bm W^*,\bm Z^*)$ by Taylor's theorem. The linear term is
\[
 D_{(W,Z)}\mathscr Q_{\bm\zeta^*}(\bm W^*,\bm Z^*)[\Delta\bm W,\bm\Delta_Z].
\]
Subtracting the displayed chain-rule identity therefore yields
\begin{align*}
 \bm q_y(\bm\zeta^*;\widehat{\bm\nu})-\bm q_y(\bm\zeta^*;\bm\nu^*)
 -D_\nu\bm q_y(\bm\zeta^*;\bm\nu^*)[\widehat{\bm h}]=D_{(W,Z)}\mathscr Q_{\bm\zeta^*}(\bm W^*,\bm Z^*)[\bm R_W,\bm0]
 +\bm R_Q,
\end{align*}
where the Taylor remainder satisfies
$\|\bm R_Q\|_2\le O_p(1)(\|\Delta\bm W\|_F+
 n^{-1/2}\|\bm\Delta_Z\|_F)^2$ by \eqref{eq: response Frechet bounds}.
Using $\|\bm R_W\|_F=O_p(K_n/s_n)$ from
\eqref{eq: supplementary W first expansion}, the norm of the right-hand side is at most
\[
 O_p(1)\{K_n/s_n+(\|\Delta\bm W\|_F+
 n^{-1/2}\|\bm\Delta_Z\|_F)^2\}=O_p(K_n^2/s_n).
\]
This proves \eqref{eq: supplementary response composition true}.
Finally, take $\bm h=(\bm h_\alpha,\bm E_Z)$ with
$\psi(\bm h)\in\mathcal T_n$ and $\|\psi(\bm h)\|_2=1$.
Then $\|\bm E_Z\|_F^2\le\|\bm h_\alpha\|_2^2+\|\bm E_Z\|_F^2=1$,
and \eqref{eq: derivative operator bounds} gives
$\|D_\nu\bm W(\bm\nu^*)[\bm h]\|_F\le CK_n/\sqrt n$.
The chain rule and the first bound in \eqref{eq: response Frechet bounds} yield
\begin{align*}
 \|D_\nu\bm q_y(\bm\zeta^*;\bm\nu^*)[\bm h]\|_2
 &=\|D_{(W,Z)}\mathscr Q_{\bm\zeta^*}(\bm W^*,\bm Z^*)
 [D_\nu\bm W(\bm\nu^*)[\bm h],\bm E_Z]\|_2\\
 &\le O_p(1)\{\|D_\nu\bm W(\bm\nu^*)[\bm h]\|_F
 +n^{-1/2}\|\bm E_Z\|_F\}\\
 &\le O_p(1)(CK_n/\sqrt n+n^{-1/2})=O_p(K_n/\sqrt n),
\end{align*}
where $K_n\ge1$ and the same $O_p(1)$ factor applies to all such directions.
Each coordinate is bounded in absolute value by this vector norm, so taking
the supremum over $\bm h$ and the maximum over the score coordinates proves \eqref{eq: response network derivative bound}.
\hfill$\blacksquare$

\begin{lemma}\label{lem: supplementary Riesz energy}
Suppose Assumptions~\ref{assump: network}--\ref{assump: response errors} hold. For each $a\in[p+r+2]$, there is a unique direction $\bm v_a$ with $\psi(\bm v_a)\in\mathcal T_n$ satisfying, for every $\bm h$ with $\psi(\bm h)\in\mathcal T_n$,
\begin{align}
\label{eq: supplementary Riesz equation}
 \psi(\bm v_a)^\top\mathsf F_*\psi(\bm h)=D_\nu q_{y,a}(\bm\zeta^*;\bm\nu^*)[\bm h].
\end{align}
The direction satisfies
\begin{equation}\label{eq: supplementary Riesz direction bounds}
 \psi(\bm v_a)^\top\mathsf F_*\psi(\bm v_a)=O_p(K_n^2/(ns_n)),\quad
 \|\psi(\bm v_a)\|_2=O_p(K_n/(\sqrt n s_n)).
\end{equation}
Moreover,
\begin{equation}\label{eq: supplementary target functional expansion}
\begin{gathered}
 \psi(\bm v_a)^\top\bm g_A=O_p(K_n/\sqrt{ns_n}),\quad D_\nu q_{y,a}(\bm\zeta^*;\bm\nu^*)[\widehat{\bm h}]
  =\psi(\bm v_a)^\top\bm g_A+O_p(K_n^2/s_n).
\end{gathered}
\end{equation}
\end{lemma}

\noindent\textbf{Proof of Lemma~\ref{lem: supplementary Riesz energy}.}
Fix $a\in[p+r+2]$. With $\bm\zeta^*$, $\bm y$ and $\bm X$ held fixed,
define the gradient of $q_{y,a}$ in the node ordering of $\psi$ by
\[
 \nabla_\psi q_{y,a}(\bm\zeta^*;\bm\nu^*)
 :=\left.\psi\left(
 \left(\frac{\partial q_{y,a}(\bm\zeta^*;\bm\nu)}{\partial\alpha_i}\right)_{i\in[n]},
 \left(\frac{\partial q_{y,a}(\bm\zeta^*;\bm\nu)}{\partial z_{ik}}\right)_{i\in[n],\,k\in[r]}
 \right)\right|_{\bm\nu=\bm\nu^*}.
\]
Define its orthogonal projection onto $\mathcal T_n$ by $\bm g_a:=\bm\Pi_{\mathcal T_n} \nabla_\psi q_{y,a}(\bm\zeta^*;\bm\nu^*)\in\mathcal T_n$.
For every $\psi(\bm h)\in\mathcal T_n$,
$\bm\Pi_{\mathcal T_n}\psi(\bm h)=\psi(\bm h)$, so the chain rule and
symmetry of $\bm\Pi_{\mathcal T_n}$ give
\[
 D_\nu q_{y,a}(\bm\zeta^*;\bm\nu^*)[\bm h]
 =\{\nabla_\psi q_{y,a}(\bm\zeta^*;\bm\nu^*)\}^\top\psi(\bm h)
 =\bm g_a^\top\psi(\bm h).
\]
Since $\bm g_a\in\mathcal T_n$, equation~\eqref{eq: response network derivative bound}
gives $\|\bm g_a\|_2=O_p(K_n/\sqrt n)$. The
information operator $\mathsf F_*$ has minimum eigenvalue at least $cs_n$.
Hence the required direction is defined by
$\psi(\bm v_a)=\mathsf F_*^{-1}\bm g_a$.
Substitution proves \eqref{eq: supplementary Riesz equation}, and positive
definiteness of $\mathsf F_*$ makes this direction unique.
Moreover,
\begin{align*}
 \psi(\bm v_a)^\top\mathsf F_*\psi(\bm v_a)
 &=\bm g_a^\top\mathsf F_*^{-1}\bm g_a
 \le C\|\bm g_a\|_2^2/s_n=O_p(K_n^2/(ns_n)),\\
 \|\psi(\bm v_a)\|_2&\le C\|\bm g_a\|_2/s_n
 =O_p(K_n/(\sqrt n s_n)).
\end{align*}
This proves \eqref{eq: supplementary Riesz direction bounds}.
Conditional on $\bm y$, the direction is fixed and the edges remain independent. Therefore, by \eqref{eq: supplementary Riesz direction bounds},
\begin{align*}
 \operatorname{Var}\{\psi(\bm v_a)^\top\bm g_A\mid\bm y\}
 =\sum_{i<j}w_{ij}^*
       \{\dot\theta_{ij}^*[\bm v_a]\}^2
 =\psi(\bm v_a)^\top\mathsf F_*\psi(\bm v_a)
 =O_p(K_n^2/(ns_n)).
\end{align*}
Since $\mathbb E\{\psi(\bm v_a)^\top\bm g_A\mid\bm y\}=0$,
conditional Chebyshev gives, for $L>0$,
\[
 \Pr\!\left\{\left|\psi(\bm v_a)^\top\bm g_A\right|
 >\frac{LK_n}{\sqrt{ns_n}}\,\middle|\,\bm y\right\}
 \le\min\!\left\{1,\frac{ns_n}{L^2K_n^2}
 \psi(\bm v_a)^\top\mathsf F_*\psi(\bm v_a)\right\}.
\]
For $M>0$, define the event on which the conditional variance is at most
$MK_n^2/(ns_n)$ by
\[
 \mathcal V_{a,n}(M)
 :=\left\{\psi(\bm v_a)^\top\mathsf F_*\psi(\bm v_a)
 \le\frac{MK_n^2}{ns_n}\right\}.
\]
This event is determined by $\bm y$.
Splitting the tail probability over $\mathcal V_{a,n}(M)$ and its complement,
and applying the conditional bound on $\mathcal V_{a,n}(M)$, gives
\[
\begin{aligned}
 &\Pr\!\left\{\left|\psi(\bm v_a)^\top\bm g_A\right|
 >\frac{LK_n}{\sqrt{ns_n}}\right\}\\
 &\quad\le\mathbb E\!\left[\mathbf 1_{\mathcal V_{a,n}(M)}
 \Pr\!\left\{\left|\psi(\bm v_a)^\top\bm g_A\right|
 >\frac{LK_n}{\sqrt{ns_n}}\,\middle|\,\bm y\right\}\right]
 +\Pr\{\mathcal V_{a,n}(M)^c\}\\
 &\quad\le\frac{M}{L^2}+\Pr\{\mathcal V_{a,n}(M)^c\}.
\end{aligned}
\]
By \eqref{eq: supplementary Riesz direction bounds},
$\lim_{M\to\infty}\limsup_{n\to\infty}\Pr\{\mathcal V_{a,n}(M)^c\}=0$.
Taking $\limsup_{n\to\infty}$ first, then letting $L\to\infty$ and $M\to\infty$
in the preceding bound, we obtain
\[
 \lim_{L\to\infty}\limsup_{n\to\infty}
 \Pr\!\left\{\left|\psi(\bm v_a)^\top\bm g_A\right|
 >\frac{LK_n}{\sqrt{ns_n}}\right\}=0,
\]
which proves the first assertion in \eqref{eq: supplementary target functional expansion}.

The estimator constructed in Lemma~\ref{lem: local likelihood root} has zero
projected score. Since $\psi(\bm v_a)\in\mathcal T_n$,
$D_\nu L_A(\widehat{\bm\nu})[\bm v_a]=0$.
Write $\bm V_{Z,a}$ for the latent component of $\bm v_a$, and put
$q_{a,ij}=\dot\theta_{ij}^*[\bm v_a]$ and
$r^v_{a,ij}=\bm v_{a,z_i}^\top\bm\Delta_{z,j}
+\bm\Delta_{z,i}^\top\bm v_{a,z_j}$.
Then $\dot\theta_{ij}(\widehat{\bm\nu})[\bm v_a]=q_{a,ij}+r^v_{a,ij}$.
With $\ell_{ij}^\theta$ and $r_{ij}^\theta$ as in
\eqref{eq: dyad linear remainder moments}--\eqref{eq: dyad quadratic remainder moments}, define
$R^P_{ij}=\widehat p_{ij}-p_{ij}^*-w_{ij}^*\ell_{ij}^\theta$.
On the predictor interval in \eqref{eq: hetero scales},
$|\expit''(u)|\le\expit'(u)\le w_{+,n}$.
Applying the Taylor expansion of $\expit(\widehat\theta_{ij})$ about
$\theta_{ij}^*$,
and using $\Delta_{\theta,ij}=\ell_{ij}^\theta+r_{ij}^\theta$, yields $R^P_{ij}=w_{ij}^*r_{ij}^\theta +\tfrac12\expit''(\xi_{ij})\Delta_{\theta,ij}^2$,
where $\xi_{ij}$ lies between $\theta_{ij}^*$ and $\widehat\theta_{ij}$.
This gives the first bound below; the second follows from the mean-value theorem:
\[
 |R^P_{ij}|\le Cw_{+,n}(|r_{ij}^\theta|+\Delta_{\theta,ij}^2),
 \quad |\widehat p_{ij}-p_{ij}^*|\le Cw_{+,n}|\Delta_{\theta,ij}|.
\]
For a dyad array $f$, abbreviate
$\|f\|_{\mathcal E,2}^2=m_n^{-1}\sum_{i<j}f_{ij}^2$.
The moment bounds in \eqref{eq: dyad logit moment transfer},
\eqref{eq: dyad logit fourth moment} and
\eqref{eq: dyad quadratic remainder moments} yield
\begin{equation}
\label{eq: target probability remainder norms}
 \|R^P\|_{\mathcal E,2}=O_p(w_{+,n}/s_n),\quad
 \|\widehat{\bm P}-\bm P^*\|_{\mathcal E,2}=O_p(w_{+,n}/\sqrt{s_n}).
\end{equation}
The weight bound $w_{ij}^*\ge w_{-,n}$, the Fisher identity in
\eqref{eq: network Fisher information}, and
\eqref{eq: supplementary Riesz direction bounds} give
$\|q_a\|_{\mathcal E,2}=O_p\{K_n/(ns_n)\}$.
Since $\bm V_{Z,a}$ is the latent component of $\bm v_a$,
\eqref{eq: supplementary Riesz direction bounds} gives $\|\bm V_{Z,a}\|_F^2\le\|\psi(\bm v_a)\|_2^2 =O_p\!\left(\frac{K_n^2}{ns_n^2}\right)$.
Combining this with Cauchy--Schwarz and
$\|\bm\Delta_Z\|_F^2=O_p(n/s_n)$ from Theorem~\ref{thm: first stage}, we obtain
\begin{equation}
\label{eq: mixed direction dyad bound}
 \|q_a\|_{\mathcal E,2}=O_p\{K_n/(ns_n)\},\quad
 \|r_a^v\|_{\mathcal E,2}^2
 \le Cn^{-2}\|\bm V_{Z,a}\|_F^2\|\bm\Delta_Z\|_F^2
 =O_p\{K_n^2/(n^2s_n^3)\}.
\end{equation}
For the first bound use
$w_{-,n}\|q_a\|_{\mathcal E,2}^2
\le m_n^{-1}\psi(\bm v_a)^\top\mathsf F_*\psi(\bm v_a)$ and $s_n=nw_{-,n}$.
Substituting $\dot\theta_{ij}(\widehat{\bm\nu})[\bm v_a]=q_{a,ij}+r^v_{a,ij}$
into the score equation $D_\nu L_A(\widehat{\bm\nu})[\bm v_a]=0$
and writing $a_{ij}-\widehat p_{ij}=(a_{ij}-p_{ij}^*)-(\widehat p_{ij}-p_{ij}^*)$
yields
\[
 0=\sum_{i<j}
 \{(a_{ij}-p_{ij}^*)-(\widehat p_{ij}-p_{ij}^*)\}
 (q_{a,ij}+r^v_{a,ij}).
\]
The two sums involving $q_{a,ij}$, which enter the preceding equation
with signs $+$ and $-$, respectively, satisfy
\begin{align*}
 \sum_{i<j}(a_{ij}-p_{ij}^*)q_{a,ij}
 &=\psi(\bm v_a)^\top\bm g_A,\\
 \sum_{i<j}(\widehat p_{ij}-p_{ij}^*)q_{a,ij}
 &=\psi(\bm v_a)^\top\mathsf F_*\psi(\widehat{\bm h})
 +\sum_{i<j}q_{a,ij}R^P_{ij}.
\end{align*}
Rearranging gives
\begin{equation*}\begin{aligned}
\psi(\bm v_a)^\top\mathsf F_*\psi(\widehat{\bm h})
 -\psi(\bm v_a)^\top\bm g_A
 ={}&\tr\{\bm V_{Z,a}^\top(\bm A-\bm P^*)\bm\Delta_Z\}\\
 &-\sum_{i<j}\{q_{a,ij}R^P_{ij}
 +r^v_{a,ij}(\widehat p_{ij}-p_{ij}^*)\}.
\end{aligned}
\end{equation*}
Symmetry gives the trace identity for the noise term, whose absolute
value is bounded by
\[
 \|\bm V_{Z,a}\|_F\|\bm A-\bm P^*\|_2\|\bm\Delta_Z\|_F
 =O_p\{(K_n/(\sqrt n s_n))\sqrt{K_ns_n}\sqrt{n/s_n}\}
 =O_p(K_n^{3/2}/s_n).
\]
The other two terms are at most
$m_n\|q_a\|_{\mathcal E,2}\|R^P\|_{\mathcal E,2}$ and
$m_n\|r_a^v\|_{\mathcal E,2}\|\widehat{\bm P}-\bm P^*\|_{\mathcal E,2}$.
Both are $O_p(K_n^2/s_n)$ by
\eqref{eq: target probability remainder norms}--\eqref{eq: mixed direction dyad bound}
and $w_{+,n}/w_{-,n}=K_n$.
Applying \eqref{eq: supplementary Riesz equation} with $\bm h=\widehat{\bm h}$
therefore proves the second assertion in \eqref{eq: supplementary target functional expansion}.
\hfill$\blacksquare$

\begin{lemma}\label{lem: joint covariance}
Suppose Assumptions~\ref{assump: network}--\ref{assump: response errors} hold.
The residual moment estimates defined with \eqref{eq: empirical error distribution}
and the aligned covariance estimator in \eqref{eq: joint feasible covariance} satisfy
\begin{equation}
\label{eq: joint covariance consistency}
 \widehat v-\sigma^{*2}=o_p(1),\quad
 \max_{j=3,4}|\widehat\mu_j-\mu_j|=o_p(1),\quad
 \|\widehat{\bm V}_n-\bm V_n\|_2=o_p(1).
\end{equation}
If \eqref{eq: response nondegeneracy} also holds, then
$\Pr\{\widehat{\bm V}_n\succ0\}\longrightarrow1$.
\end{lemma}

\noindent\textbf{Proof of Lemma~\ref{lem: joint covariance}.}
All latent quantities in the proof use the aligned coordinates. The fitted product
$\widehat{\bm Z}\widehat{\bm\eta}$, and hence the structural residual and its
empirical distribution, are invariant to this alignment. Lemma~\ref{lem: response profile uniform}
gives $\|\widehat{\bm\zeta}-\bm\zeta^*\|_2=O_p(K_n/\sqrt{s_n})=o_p(1)$.
For $\bm u\in\mathbb R^n$, write
$\|\bm u\|_{n,4}=(n^{-1}\sum_i|u_i|^4)^{1/4}$.
The exact normalization identity \eqref{eq: exact weighted normalization}
and $|\Delta_{\theta,ij}|\le C(\Delta_i+\Delta_j)$ imply
\begin{align}
\label{eq: weight row consistency}
 a_n^W:=\max_i\sum_j|\widehat W_{ij}-W_{ij}^*|
 \le C\max_i\Delta_i=o_p(1).
\end{align}
Indeed the relative probability errors are uniformly at most
$C\max_i\Delta_i$, and the denominator in that identity is bounded away
from zero.
The proof of Lemma~\ref{lem: response score representation} establishes
$\|\bm\mu_y^*\|_\infty\le C$ and
$\mathbb E_\epsilon|\epsilon_i|^4=\mu_4\le C$; its bound
\eqref{eq: stochastic powers resolvents} gives $\|\bm S^{*-1}\|_\infty\le C$.
Applying Minkowski's inequality to
$\bm y=\bm\mu_y^*+\bm S^{*-1}\bm\epsilon$ therefore yields
\[
 \sup_i(\mathbb E_\epsilon|y_i|^4)^{1/4}
 \le\|\bm\mu_y^*\|_\infty
 +\|\bm S^{*-1}\|_\infty\sup_j(\mathbb E_\epsilon|\epsilon_j|^4)^{1/4}
 \le C.
\]
Conditional on the network, a second application of Minkowski's inequality
therefore gives
\[
 \mathbb E_\epsilon\|(\widehat{\bm W}-\bm W^*)\bm y\|_{n,4}^4
 \le C(a_n^W)^4,\quad
 \mathbb E_\epsilon\|\widehat{\bm W}\bm y\|_{n,4}^4\le C.
\]
Conditional Markov's inequality yields $o_p(1)$ and $O_p(1)$ for the
respective fourth norms.
The exact residual identity is
\[
 \widehat{\bm\epsilon}-\bm\epsilon
 =-\bm X\Delta\bm\beta-\widehat{\bm Z}\Delta\bm\eta
 -\bm\Delta_Z\bm\eta^*-\Delta\gamma\widehat{\bm W}\bm y
 -\gamma^*(\widehat{\bm W}-\bm W^*)\bm y.
\]
Here $\Delta\bm\beta,\Delta\bm\eta,\Delta\gamma$ denote fitted minus true
values. Their consistency, bounded design rows and $\max_i\Delta_i=o_p(1)$
give $\|\widehat{\bm\epsilon}-\bm\epsilon\|_{n,4}=o_p(1)$.
Put $d_i=\widehat\epsilon_i-\epsilon_i$. For $1\le j\le4$, the factorization
of $x^j-y^j$ and $|\widehat\epsilon_i|\le|\epsilon_i|+|d_i|$ give
\[
 |\widehat\epsilon_i^j-\epsilon_i^j|
 \le j|d_i|(|\epsilon_i|+|d_i|)^{j-1}.
\]
Since $\mathbb E_\epsilon\|\bm\epsilon\|_{n,4}^4 =\frac1n\sum_i\mathbb E_\epsilon|\epsilon_i|^4=\mu_4\le C$,
Markov's inequality gives $\|\bm\epsilon\|_{n,4}=O_p(1)$.
H\"older's inequality now yields
\[
 \left|\frac1n\sum_i(\widehat\epsilon_i^j-\epsilon_i^j)\right|
 \le C\|\widehat{\bm\epsilon}-\bm\epsilon\|_{n,4}
 (\|\bm\epsilon\|_{n,4}+\|\widehat{\bm\epsilon}-\bm\epsilon\|_{n,4})^{j-1}=o_p(1).
\]
The bounded eighth moments and independence give
$n^{-1}\sum_i\epsilon_i^j-\mathbb E_\epsilon\epsilon_i^j=O_p(n^{-1/2})$.
Thus $\overline{\widehat\epsilon}=o_p(1)$ and
$\widehat v-\sigma^{*2}=o_p(1)$. Since $\sigma^{*2}\ge c>0$ by
Assumption~\ref{assump: response}, $\Pr(\widehat v=0)\to0$.
Combining $\widehat\tau^{-1}-\sigma^{*2}=o_p(1)$ with centering and rescaling
as in \eqref{eq: empirical error distribution} proves consistency of
$\widehat\mu_3$ and $\widehat\mu_4$.

The fitted mean can be controlled by the same row bound. Specifically,
\begin{equation}\label{eq: fitted mean row consistency}
\begin{aligned}
&\widehat{\bm\mu}_y-\bm\mu_y^*
 =\widehat{\bm S}^{-1}\{\bm X\Delta\bm\beta+
 \widehat{\bm Z}\Delta\bm\eta+\bm\Delta_Z\bm\eta^*
 +(\bm S^*-\widehat{\bm S})\bm\mu_y^*\},\\
 &\|\widehat{\bm\mu}_y-\bm\mu_y^*\|_\infty+
 \|\widehat{\bm W}\widehat{\bm\mu}_y-\bm W^*\bm\mu_y^*\|_\infty
 =o_p(1).
\end{aligned}
\end{equation}
The first line follows by multiplying by $\widehat{\bm S}$.
For the second, \eqref{eq: response operator bounds} with $t=1$ and
$\gamma=\widehat\gamma\in\Gamma_\gamma$ gives
$\|\widehat{\bm S}^{-1}\|_\infty\le C$; Lemma~\ref{lem: response operator bounds}
also gives $\|\bm\mu_y^*\|_\infty\le C$.
Since $|\gamma^*|\le1-\delta_\gamma$ by Assumption~\ref{assump: response} and
$\bm S^*-\widehat{\bm S}
=\Delta\gamma\widehat{\bm W}+\gamma^*(\widehat{\bm W}-\bm W^*)$,
\[
 \|(\bm S^*-\widehat{\bm S})\bm\mu_y^*\|_\infty
 \le (|\Delta\gamma|+|\gamma^*|a_n^W)\|\bm\mu_y^*\|_\infty
 \le C(|\Delta\gamma|+a_n^W)=o_p(1),
\]
where the last equality uses \eqref{eq: supplementary profile estimator rate}
and \eqref{eq: weight row consistency}.
The other three terms inside the braces are $o_p(1)$ in maximum norm by
\eqref{eq: supplementary profile estimator rate}, the maximum node error bound
in \eqref{eq: first stage recovery}, and the design bound in
Assumption~\ref{assump: response}.
Finally, row stochasticity of $\widehat{\bm W}$ gives
\[
 \|\widehat{\bm W}\widehat{\bm\mu}_y-\bm W^*\bm\mu_y^*\|_\infty
 \le\|\widehat{\bm\mu}_y-\bm\mu_y^*\|_\infty
 +a_n^W\|\bm\mu_y^*\|_\infty=o_p(1).
\]
This also gives $\|\widehat{\bm\mu}_y\|_\infty=O_p(1)$.

It remains to prove consistency of the full-score covariance and Jacobian.
Put $\widehat{\bm Q}=\widehat{\bm W}\widehat{\bm S}^{-1}$.
Condition on the observed data, so that the fitted quantities are fixed
and only $\bm e$ is random. By construction,
$\mathbb E_{\bm e}e_i=0$ and $\mathbb E_{\bm e}e_i^2=\widehat\sigma^2$;
the moment consistency proved above gives
$\widehat\sigma^2,\widehat\mu_3,\widehat\mu_4=O_p(1)$.
A conditional covariance calculation analogous to
\eqref{eq: linear quadratic covariance identities}, with
$(\sigma^{*2},\mu_3,\mu_4)$ replaced by
$(\widehat\sigma^2,\widehat\mu_3,\widehat\mu_4)$, followed by
Cauchy--Schwarz, yields, for any fixed $\bm u\in\mathbb R^n$ and symmetric
$\bm B\in\mathbb R^{n\times n}$,
\begin{align*}
 \operatorname{Var}_{\bm e}(\bm u^\top\bm e+\bm e^\top\bm B\bm e)
 \le C_n(\|\bm u\|_2^2+\|\bm B\|_F^2),\quad C_n&=\widehat\sigma^2+|\widehat\mu_3|+\widehat\mu_4
 +2\widehat\sigma^4=O_p(1).
\end{align*}
For a nonsymmetric matrix, its symmetric part gives the same
quadratic form and has no larger Frobenius norm.
To bound the full score, put
$\widehat{\bm C}=(\bm X,\widehat{\bm Z},
\widehat{\bm W}\widehat{\bm\mu}_y)$.
Direct substitution of the artificial response into
\eqref{eq: full response score} gives
\[
 \sqrt n\widehat{\bm q}_y^{(\bm e)}
 =\frac1{\sqrt n}\begin{pmatrix}
 \widehat\tau\widehat{\bm C}^{\top}\bm e\\
 -\tfrac12\sum_i(e_i^2-\widehat\sigma^2)
 \end{pmatrix}+\widehat{\bm r}_{q,n}^{(\bm e)},\qquad
 (\widehat{\bm r}_{q,n}^{(\bm e)})_\gamma
 =\frac{\widehat\tau}{\sqrt n}
 \{\bm e^\top\widehat{\bm Q}\bm e-\widehat\sigma^2\tr\widehat{\bm Q}\},
\]
with all other remainder coordinates zero. Since $\widehat\tau=O_p(1)$ by
\eqref{eq: supplementary profile estimator rate} and
$\|\widehat{\bm Q}\|_F^2\le CK_n$ by
\eqref{eq: response operator bounds} with $t=1$ and $\gamma=\widehat\gamma$,
the conditional variance bound gives
\[
 \mathbb E_{\bm e}\|\widehat{\bm r}_{q,n}^{(\bm e)}\|_2^2
 \le\frac{C_n\widehat\tau^2}{n}\|\widehat{\bm Q}\|_F^2
 =O_p(K_n/n)=o_p(1).
\]
The bounded design rows, latent recovery and
\eqref{eq: fitted mean row consistency} imply
$\max_i\|\widehat{\bm C}_{i\cdot}\|_2=O_p(1)$.
Consequently, with the observed data held fixed,
\[
 \mathbb E_{\bm e}\left\|
 \frac1{\sqrt n}\begin{pmatrix}
 \widehat\tau\widehat{\bm C}^{\top}\bm e\\
 -\tfrac12\sum_i(e_i^2-\widehat\sigma^2)
 \end{pmatrix}\right\|_2^2
 =\frac{\widehat\tau^2\widehat\sigma^2}{n}\|\widehat{\bm C}\|_F^2
 +\frac{\widehat\mu_4-\widehat\sigma^4}{4}=O_p(1).
\]
Thus $n\mathbb E_{\bm e}\|\widehat{\bm q}_y^{(\bm e)}\|_2^2=O_p(1)$.
Since $\mathbb E_{\bm e}\bigl(\sqrt n\widehat{\bm q}_y^{(\bm e)}
 -\widehat{\bm r}_{q,n}^{(\bm e)}\bigr)=\bm0$ and $\mathbb E_{\bm e}\widehat{\bm r}_{q,n}^{(\bm e)}=\bm0$,
Cauchy--Schwarz gives
\begin{align*}
 \left\|\widehat{\bm\Xi}_n
 -\operatorname{Var}_{\bm e}\bigl(\sqrt n\widehat{\bm q}_y^{(\bm e)}
 -\widehat{\bm r}_{q,n}^{(\bm e)}\bigr)\right\|_2
 &\le\mathbb E_{\bm e}\|\widehat{\bm r}_{q,n}^{(\bm e)}\|_2^2\\
 &\quad+2\left(\mathbb E_{\bm e}\|\sqrt n\widehat{\bm q}_y^{(\bm e)}
 -\widehat{\bm r}_{q,n}^{(\bm e)}\|_2^2\right)^{1/2}
 \left(\mathbb E_{\bm e}\|\widehat{\bm r}_{q,n}^{(\bm e)}\|_2^2\right)^{1/2}
 =o_p(1).
\end{align*}
Moreover,
\[
 \operatorname{Var}_{\bm e}\bigl(\sqrt n\widehat{\bm q}_y^{(\bm e)}
 -\widehat{\bm r}_{q,n}^{(\bm e)}\bigr)
 =\begin{pmatrix}
 \widehat\tau^2\widehat\sigma^2\widehat{\bm C}^{\top}\widehat{\bm C}/n
 &-\widehat\tau\widehat\mu_3\widehat{\bm C}^{\top}\bm1_n/(2n)\\
 -\widehat\tau\widehat\mu_3\bm1_n^\top\widehat{\bm C}/(2n)
 &(\widehat\mu_4-\widehat\sigma^4)/4
 \end{pmatrix}.
\]
Similarly, applying the same Cauchy--Schwarz covariance estimate to
\eqref{eq: reduced score representation}, with $\bm\Xi_n$ defined in
\eqref{eq: response information definitions}, gives
\[
\bm\Xi_n
 =\operatorname{Var}_\epsilon\left\{
 \frac1{\sqrt n}\begin{pmatrix}
 \tau^*\bm C_n^\top\bm\epsilon\\
 -\tfrac12\sum_i(\epsilon_i^2-\sigma^{*2})
 \end{pmatrix}\right\}+o(1) 
 =\begin{pmatrix}
 \tau^{*2}\sigma^{*2}\bm G_n
 &-\tau^*\mu_3\bar{\bm c}_n/2\\
 -\tau^*\mu_3\bar{\bm c}_n^\top/2
 &(\mu_4-\sigma^{*4})/4
 \end{pmatrix}+o(1).
\]
By \eqref{eq: first stage recovery} and \eqref{eq: fitted mean row consistency},
$n^{-1/2}\|\widehat{\bm C}-\bm C_n\|_F=o_p(1)$.
Comparing the two covariance matrices and using moment consistency and
\eqref{eq: supplementary profile estimator rate} proves
\begin{align}
\label{eq: fitted full covariance consistency}
 \|\widehat{\bm\Xi}_n-\bm\Xi_n\|_2=o_p(1).
\end{align}

For the Jacobian, consistency of $\widehat{\bm\zeta}$ and
\eqref{eq: supplementary fitted Jacobian consistency} directly give $\|\widehat{\mathsf J}_{y,n}-\mathsf J_{y,n}\|_2=o_p(1)$. Since $\|\mathsf J_{y,n}^{-1}\|_2\le C$ by Lemma~\ref{lem: response conditioning},
$\widehat{\mathsf J}_{y,n}$ is invertible with probability tending to one and
$\|\widehat{\mathsf J}_{y,n}^{-1}-\mathsf J_{y,n}^{-1}\|_2=o_p(1)$.
Lemma~\ref{lem: response conditioning} and
\eqref{eq: fitted full covariance consistency} then give $\|\bm\Xi_n\|_2\le C$, $\|\widehat{\bm\Xi}_n\|_2+\|\widehat{\mathsf J}_{y,n}^{-1}\|_2=O_p(1)$.
The triangle inequality applied to \eqref{eq: joint feasible covariance}
and \eqref{eq: joint response covariance} gives
\begin{align*}
 \|\widehat{\bm V}_n-\bm V_n\|_2
 &\le\|\widehat{\mathsf J}_{y,n}^{-1}-\mathsf J_{y,n}^{-1}\|_2
 \|\widehat{\bm\Xi}_n\|_2\|\widehat{\mathsf J}_{y,n}^{-1}\|_2+\|\mathsf J_{y,n}^{-1}\|_2
 \|\widehat{\bm\Xi}_n-\bm\Xi_n\|_2
 \|\widehat{\mathsf J}_{y,n}^{-1}\|_2\\
 &\quad+\|\mathsf J_{y,n}^{-1}\|_2\|\bm\Xi_n\|_2
 \|\widehat{\mathsf J}_{y,n}^{-1}-\mathsf J_{y,n}^{-1}\|_2=o_p(1),
\end{align*}
which proves \eqref{eq: joint covariance consistency}.
Under \eqref{eq: response nondegeneracy}, Lemma~\ref{lem: response conditioning}
gives $\lambda_{\min}(\bm V_n)\ge c>0$; hence
\eqref{eq: joint covariance consistency} implies
$\Pr\{\widehat{\bm V}_n\succ0\}\longrightarrow1$.
\hfill$\blacksquare$

\putbib[NLDA]

\end{bibunit}
\end{document}